\documentclass[final,3p,times]{elsarticle}

\usepackage{amsmath,hyperref}
\usepackage{lineno}

\usepackage{epstopdf}
\journal{Journal}

\usepackage{lineno,hyperref}
\usepackage{graphicx}
\usepackage{dcolumn}
\usepackage{bm}
\usepackage{epsfig}
\usepackage{booktabs}
\usepackage{subfigure}
\usepackage{graphics}
\usepackage{amssymb,enumitem}
\usepackage{amsmath}
\usepackage{array}
\usepackage{color}
\usepackage{booktabs}
\usepackage{multirow}
\usepackage{caption}
\usepackage{chngpage}
\usepackage{subfigure}
\usepackage{mathrsfs,gensymb,float}

\usepackage{verbatim}

\newtheorem{theorem}{Theorem}

\biboptions{numbers,sort&compress}
\modulolinenumbers[1]
\begin{document}

	\captionsetup[figure]{labelfont={bf},name={Fig.},labelsep=period}        
	
	\begin{frontmatter}
		
		
		\title{A thermodynamically consistent and conservative phase-field model for surfactant contact line dynamics with solid adsorption}
		
		\author[1]{Jiangxu Huang}
		\author[1,3,4,5]{Zhenhua Chai\corref{mycorrespondingauthor}}	
		\ead{hustczh@hust.edu.cn}
		\cortext[mycorrespondingauthor]{Corresponding author}
		\author[1]{Xi Liu}
		\author[1]{Changsheng Huang}

		\address[1]{ School of Mathematics and Statistics, Huazhong University of Science and Technology, Wuhan 430074, China}
		\address[3]{Institute of Interdisciplinary Research for Mathematics and Applied Science,Huazhong University of Science and Technology, Wuhan 430074, China}	
		\address[4]{Hubei Key Laboratory of Engineering Modeling and Scientific Computing, Huazhong University of Science and Technology, Wuhan 430074, China}	
		
		\address[5]{The State Key Laboratory of Intelligent Manufacturing Equipment and Technology, Huazhong University of Science and Technology, Wuhan 430074, China}


		\begin{abstract}
			Surfactant adsorption at solid–fluid interfaces has an important influence on the contact line dynamics, yet most surfactant-laden phase-field models neglect solid adsorption and are still based on Cahn--Hilliard-type equations. In this work, we develop a thermodynamically consistent and conservative Allen--Cahn--Navier--Stokes model for surfactant-laden contact line dynamics with solid adsorption. In this model, a unified free-energy formulation accounts for surfactant transport, adsorption, and wetting dynamics, from which consistent wetting and adsorption boundary conditions are derived. To overcome the numerical complexity of fourth-order formulations in the Cahn--Hilliard-type equations, the second-order Allen--Cahn--type equations are employed for both the phase field and the surfactant concentration. Within the Allen--Cahn framework, a classical linear constraint ensures exact conservation of the total surfactant mass, whereas a generalized nonlinear conserved quantity is introduced for the phase field to improve geometric volume conservation. The resulting weighted Lagrange-multiplier correction vanishes in the bulk phases and is confined to the diffuse interface, thereby suppressing the leading-order curvature-induced bulk shift associated with the classical linear constraint while preserving the energy-dissipation structure of the model. Matched asymptotic analysis is further performed to show that this nonlinear formulation can achieve formal second-order geometric-volume accuracy. Finally, a lattice Boltzmann method is developed to solve the coupled system, and the numerical results confirm the predicted second-order convergence, demonstrate significantly enhanced volume preservation for high-curvature droplets, and reproduce characteristic wetting responses induced by preferential solid adsorption, including hydrophilization and autophobing. Under shear flow, the results also illustrate that solid adsorption can significantly influence contact line recession and droplet detachment.
		\end{abstract}

		
	\end{frontmatter}
	
	\section{Introduction}
	Surfactants, owing to their amphiphilic nature, can adsorb at fluid--fluid and solid--liquid interfaces, reducing interfacial tension \cite{TakagiARFM2011}. The nonlinear coupling between their bulk transport and interfacial adsorption/desorption kinetics drives spatial gradients in interfacial tension, thereby generating Marangoni stresses. These effects have some significant influences on the interface deformation, breakup, coalescence, and contact line dynamics, with broad applications in emulsification, microfluidics, enhanced oil recovery, and biomedicine \cite{MyersJWS2020,SeethepalliSPEJ2004,BaretLC2012}. However, the strong nonlinear coupling among surfactant transport, interfacial adsorption/desorption, interfacial-tension variation, and bulk flow poses a great challenge for numerical simulation and has attracted considerable research interest.

	Existing numerical models for surfactant-laden multiphase flows mainly describe surfactant transport on fluid–fluid interfaces and the associated coupling with interfacial tension and Marangoni stresses. Based on the interface representation, these models can be broadly classified into sharp-interface and diffuse-interface approaches. Sharp-interface methods, including front-tracking \cite{ZhangJCP2006,MuradogluJCP2008}, immersed-boundary \cite{LaiJCP2008}, level-set \cite{XuJCP2006,XuJCP2018}, and volume-of-fluid methods \cite{JamesJCP2004,XueJCP2026}, consider the surfactant transport equation and stress jump conditions directly on a zero-thickness interface. Although these methods can accurately capture interfacial dynamics, maintaining mass conservation while retaining a simple implementation remains challenging during severe deformation and topological changes. Diffuse-interface methods, by contrast, spread the interface over a finite thickness, which makes them better suited for handling the droplet breakup and coalescence \cite{TeigenCMS2009,TeigenJCP2011}. The diffuse-interface surfactant models can be generally classified into two categories. The first one is to introduce a regularized delta function to reformulate the sharp-interface surfactant transport equation over the diffuse-interface region, including the diffuse-interface transport models \cite{TeigenCMS2009,TeigenJCP2011,LiuJFM2018,HuAML2021,YamashitaJCP2024}, the interface-confined scalar model \cite{JainJCP2024}, and recent profile-preserving or lattice-Boltzmann extensions \cite{HaoJCP2025,HaoJCP2025a,ZhanPRE2025}. The other one is to derive the governing equations from a free-energy functional, with chemical potentials obtained as variational derivatives and the phase field and the surfactant concentration evolving through Cahn–Hilliard-type or Allen--Cahn-type dynamics \cite{VanRA2006,LiuJCP2010,YunAMC2014,ZongPOF2020,YangCMAME2021,LiangJCP2026,YangJCP2022}. In contrast to the first kind of transport-based models \cite{TeigenCMS2009,TeigenJCP2011,LiuJFM2018,HuAML2021,YamashitaJCP2024,JainJCP2024,HaoJCP2025,HaoJCP2025a,ZhanPRE2025}, models of the second type feature a much clearer thermodynamic structure and an explicit energy-dissipation law.

	When surfactant-laden multiphase flows come into contact with solid substrates, the surfactants would modulate solid wettability and introduce complex contact line dynamics. In standard phase-field models for surfactant-free contact-line dynamics, the wall wettability is typically incorporated via the solid–fluid interfacial free energy, from which wetting boundary conditions can be derived \cite{CahnJCP1977,YueJCP2010}. In surfactant–substrate systems, Zhu et al. \cite{ZhuJFM2019,ZhuJCP2020} developed thermodynamically consistent phase-field models in which soluble surfactants adsorb at the fluid–fluid interface and reduce the interfacial tension. Similar assumptions have been adopted in subsequent studies \cite{WangJCP2022,WangJCP2024,WuAML2024,ZhangICHMT2025,YaoIJMF2026,WuCMAME2026}. However, when surfactants are allowed to modify only the fluid–fluid interfacial tension, Young's equation predicts an amplification of the intrinsic wettability: the hydrophilic surfaces become more hydrophilic, while the hydrophobic surfaces become more hydrophobic, a trend consistent with some experimental observations \cite{StoebeLangmuir1996,KarapetsasJFM2011}. In addition to these common phenomena, the surfactants can also produce counterintuitive behaviors: they can either induce an overall hydrophilic shift with a reduced contact angle or cause autophobic droplet retraction with an increased contact angle \cite{BeraLangmuir2016,TadmorLangmuir2019}.
	Existing models \cite{ZhuJFM2019,ZhuJCP2020,WangJCP2022,WangJCP2024,WuAML2024,ZhangICHMT2025,YaoIJMF2026,WuCMAME2026} fail to explain these distinct behaviors, including an overall hydrophilic shift and autophobicity, as such behaviors require an asymmetric modification of the solid–liquid and solid–gas interfacial tensions, indicating that a key physical mechanism remains unaccounted for. Recently, Kannan et al. \cite{KannanJFM2026} pointed out that this discrepancy stems from the omission of solid-surface adsorption: the surfactants not only alter the fluid–fluid interfacial tension but also adsorb at the solid–liquid and solid–gas interfaces, changing the solid surface energies and reshaping Young's force balance. Therefore, the incorporation of solid adsorption is essential for accurately describing contact line dynamics in surfactant–substrate systems \cite{KannanJFM2026}.

	Furthermore, most existing models for surfactant-laden contact lines are based on Cahn-Hilliard-type fourth-order equations \cite{ZhuJFM2019,ZhuJCP2020,WangJCP2022,WangJCP2024,WuAML2024,ZhangICHMT2025,YaoIJMF2026,WuCMAME2026,KannanJFM2026}. Although this type of model naturally satisfies mass conservation as an $H^{-1}$ gradient flow, its fourth-order spatial derivatives significantly increase the complexity of numerical implementation \cite{CahnJCP1958}. In contrast, the second-order Allen--Cahn equation, formulated as an $L^2$ gradient flow, only involves the second-order spatial derivatives and is therefore attractive for its computational simplicity and easier treatment of boundary conditions \cite{AllenAM1976}. However, the classical Allen--Cahn equation lacks intrinsic mass conservation, which poses substantial challenges when it is directly applied to study multiphase flow problems requiring accurate conservation of both droplet volume and total surfactant mass. To overcome this limitation, the existing strategies fall primarily into two categories: (i) incorporating local penalty terms or localized fluxes to devise conservative Allen--Cahn variants \cite{ChaiIJHMT2018,ChiuJCP2011}, and (ii) introducing a global Lagrange multiplier to construct a nonlocal Allen--Cahn model that enforces mass conservation through a spatial integral constraint \cite{RubinsteinIMAJAM1992,BrasselMMAS2011,KimIJES2014}. The former can improve local conservation and profile preservation, but many such corrections are not derived from a unified variational structure and therefore do not generally retain a strict energy-dissipation law. The latter preserves the spatial integral of the phase variable, but this constraint does not necessarily ensure conservation of the true geometric volume represented by the diffuse-interface profile. For the high-curvature droplets or large-contact-angle configurations, the mismatch between integral conservation and geometric volume conservation may lead to artificial droplet shrinkage or loss of fine interfacial features, as reported in the previous work \cite{BrasselMMAS2011}. Although subsequent refinements, such as curvature-dependent multipliers \cite{KwakAML2022}, local multipliers \cite{ChoiEABE2023}, and maximum-principle-preserving schemes \cite{YangCNSNS2024}, have improved the performance of conservative Allen--Cahn models, these modified methods generally break the underlying variational structure and cannot preserve the strict energy-dissipation law of the original nonlocal Allen--Cahn system \cite{HwangEABE2026}. This motivates the development of a second-order Allen--Cahn formulation that can simultaneously ensure geometric volume conservation, total surfactant mass conservation, and thermodynamic consistency for the surfactant-laden moving contact lines with solid adsorption.

	In this paper, we develop a thermodynamically consistent and conservative Allen--Cahn--Navier--Stokes framework for surfactant-laden moving contact line problems with solid adsorption. Within the conservative Allen–Cahn formulation, distinct nonlocal constraints are applied to the phase field and the surfactant concentration. A standard Lagrange multiplier is adopted for the surfactant concentration to guarantee mass conservation, whereas a more general conserved quantity is introduced for the phase field, giving rise to a new weighted Lagrange-multiplier correction. Unlike the classical Lagrange multiplier, which shifts the chemical potential uniformly throughout the domain, the present weighting-function derivative vanishes in the pure bulk phases, thereby localizing the volume-conservation correction strictly within the diffuse interface. Consequently, this construction removes the leading-order curvature-induced bulk shift associated with standard Allen–Cahn models and achieves accurate geometric volume conservation while preserving the strict energy-dissipation law. In addition, the asymptotic analysis and numerical results demonstrate that compared to the classical linear constraint, the proposed nonlinear constraint substantially improves geometric-volume preservation, achieving second-order accuracy in zero-level-set geometric-volume conservation under the stated assumptions.
	
	
	The remainder of this paper is organized as follows. In Section \ref{sec2}, we present the physical problem of surfactant-laden contact-line dynamics with solid adsorption and the phase-field model, together with theoretical analyses of its conservative properties and thermodynamic consistency. In Section \ref{sec3}, the lattice Boltzmann method is developed to solve the phase-field model. In Section \ref{sec4}, some numerical experiments are conducted to validate the phase-field model, and finally, the main conclusions are summarized in Section \ref{sec5}.

	\section{Physical problem and mathematical model}
	\label{sec2}
	In this section, an Allen–Cahn–Navier–Stokes model is developed for surfactant-laden contact-line dynamics with solid adsorption, and the physical configuration is shown in Fig. \ref{fig0}. In the following, the equilibrium distributions of the phase-field variable and the surfactant concentration, as well as the energy-dissipation law of the system, are derived and discussed in detail.
	\begin{figure}[H]
		\centering
		\includegraphics[scale=0.45]{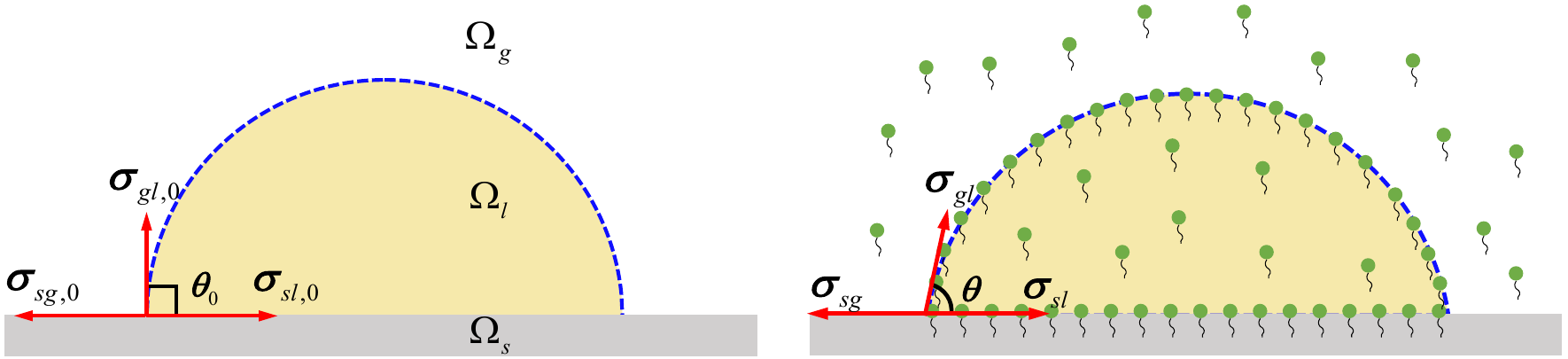} 
		\put(-385,80){(\textit{a})}
		\put(-185,80){(\textit{b})}
		\caption{ Schematic illustration of droplet spreading on a solid surface (a) without surfactant and (b) with surfactant. Surfactant molecules adsorb at both the gas–liquid and solid–liquid interfaces, thereby influencing the motion of the three-phase contact line and altering the spreading rate and the equilibrium contact angle. $\theta$ denotes the equilibrium contact angle, $\sigma_{sg}$, $\sigma_{sl}$ and $\sigma_{gl}$ are the solid–gas, solid–liquid and gas–liquid interfacial tensions, respectively. The subscript $0$ denotes the corresponding clean-system values.}
		\label{fig0}
	\end{figure}

	\subsection{Phase-field method for interface capturing and surfactant}
	We introduce two scalars $\phi$ and $\psi$ to represent the local phase state and the surfactant concentration, respectively. The phase-field variable $\phi$ changes smoothly from -1 in phase 1 (the gas phase) to 1 in phase 2 (the liquid phase), while $\psi$ denotes the normalized surfactant concentration and is constrained in the range [0, 1]. The total free-energy functional of the system, excluding the kinetic energy, is composed of four parts,
	\begin{equation}
		\mathcal{E}_f(\phi, \psi) = \mathcal{E}_{GL}(\phi) + \mathcal{E}_{sur}(\psi) + \mathcal{E}_{ad}(\phi, \psi) + \mathcal{E}_{w}(\phi,\psi).
		\label{eq1}
	\end{equation}
	The Ginzburg–Landau free energy $\mathcal{E}_{GL}(\phi)$ accounts for the structure of the gas–liquid interface and takes the standard double-well form \cite{LiangJCP2026}
	\begin{equation}
		\mathcal{E}_{GL}(\phi) = \int_{\Omega} \left[ \beta(\phi^2 - 1)^2 + \frac{\kappa}{2} |\nabla \phi|^2 \right] d\Omega,
	\end{equation}	
	where $\kappa=3 \varepsilon \sigma_{gl,0}/8$ and $\beta=3\sigma_{gl,0}/(4\varepsilon)$ are constants, and determined by the interface thickness parameter $\varepsilon$ and the surfactant-free gas–liquid interfacial tension $\sigma_{gl,0}$. The surfactant mixing free energy $\mathcal{E}_{sur}(\psi)$ describes the entropic contribution of the binary surfactant–fluid system, supplemented by a gradient penalty term to capture interfacial effects associated with solid-phase adsorption \cite{KannanJFM2026}
	\begin{equation}
		\mathcal{E}_{sur}(\psi) = \int_{\Omega} \left\{ k_B T [\psi \ln \psi + (1 - \psi) \ln(1 - \psi)] {+ \frac{\lambda_s}{2} |\nabla \psi|^2} \right\} d\Omega,
	\end{equation}
	where $k_B$ is the Boltzmann constant, $T$ is the absolute temperature and $\lambda_s = \kappa/2$ is a gradient penalty coefficient (or regularization parameter) that controls the physical thickness of the diffuse adsorption layer and preserves thermodynamic consistency with the Langmuir isotherm \cite{KannanJFM2026}. The adsorption coupling free energy $\mathcal{E}_{ad}(\phi, \psi)$ characterizes the surfactant adsorption behavior at the gas–liquid interface \cite{LiangJCP2026}
	\begin{equation}
		\mathcal{E}_{ad}(\phi, \psi) = \int_{\Omega} \left[ -\beta \psi (1 - \phi^2)^2 + \frac{W}{2} \psi \phi^2 \right] d\Omega.
	\end{equation}
	The first term accounts for the adsorption and desorption of surfactant at the interface, while the second term serves as a penalty potential that suppresses free surfactant in the bulk phases, ensuring a low local surfactant concentration away from the interface. The parameter $W>0$ controls the strength of this penalty. Finally, the wetting free energy $\mathcal{E}_{w}(\phi, \psi)$ accounts for the short-range interactions between the fluid and the solid substrate \cite{KannanJFM2026}
	\begin{equation}
		\mathcal{E}_{w}(\phi, \psi) = \int_{\partial \Omega} f_w(\phi, \psi) dS,
	\end{equation}
	where $f_w(\phi, \psi)$ is the surface energy per unit area on the wall. For a clean two-phase system, Young’s law, $\sigma_{gl,0} \cos \theta_0 = \sigma_{sg,0} - \sigma_{sl,0}$, gives the following standard form for the wall free energy per unit area \cite{CahnJCP1977,YueJCP2010},
	\begin{equation}
		f_{w,0}(\phi) = \left[ \frac{1}{2} - g(\phi) \right] \sigma_{sg,0}+ \left[ \frac{1}{2} + g(\phi) \right]\sigma_{sl,0} = -\sigma_{gl,0}\cos\theta_{0}g(\phi) + \frac{\sigma_{sg,0} + \sigma_{sl,0}}{2},
		\label{fw0}
	\end{equation}
	where $\sigma_{sg,0}$ and $\sigma_{sl,0}$ denote the solid–gas and solid–liquid interfacial tensions of the clean system, $\theta_0$ represents the surfactant-free equilibrium contact angle, $g(\phi)$ is a smooth switching function such that $g(-1)=-0.5$ for the pure gas phase and $g(1)=0.5$ for the pure liquid phase. When a soluble surfactant is introduced into the system, it adsorbs not only at the fluid–fluid interface but also at the solid–liquid and solid–gas interfaces, thereby reducing the corresponding interfacial tensions. Here we adopt the Langmuir–Szyszkowski equation of state to describe the effect of surfactant adsorption on the interfacial tensions \cite{LiuJCP2010, KannanJFM2026},
	\begin{equation}
		\sigma_{gl}(\psi) = \sigma_{gl,0} \left[ 1 + \beta_{gl} \ln(1-\psi) \right],\quad \sigma_{sg}(\psi) = \sigma_{sg,0} \left[ 1 + \beta_{sg} \ln(1-\psi) \right],\quad  \sigma_{sl}(\psi) = \sigma_{sl,0} \left[ 1 + \beta_{sl} \ln(1-\psi) \right],
		\label{langmuir}
	\end{equation}
	where $\beta_{gl}$, $\beta_{sg}$, and $\beta_{sl}$ represent the interfacial elasticity coefficients associated with surfactant adsorption at the gas-liquid, solid-gas and solid-liquid interfaces. This formulation ensures that the system naturally recovers the clean-system reference values when $\psi=0$. Substituting Eq. (\ref{langmuir}) into Eq. (\ref{fw0}) and replacing $\sigma_{sg,0} \rightarrow \sigma_{sg}(\psi)$ and $\sigma_{sl,0} \rightarrow \sigma_{sl}(\psi)$, while retaining the explicit contribution of the gas–liquid interfacial tension to the contact angle, yields the surfactant-dependent wall free energy density
	\begin{equation}
		\begin{aligned}
			f_{w}(\phi, \psi) &= \left[ \frac{1}{2} - g(\phi) \right] \left\{ \sigma_{sg,0} \left[ 1 + \beta_{sg} \ln(1-\psi) \right] \right\} + \left[ \frac{1}{2} + g(\phi) \right] \left\{ \sigma_{sl,0} \left[ 1 + \beta_{sl} \ln(1-\psi) \right] \right\} \\
			&= {-\left[ \sigma_{gl,0} \cos \theta_{0} + \ln(1-\psi)(\sigma_{sg,0} \beta_{sg} - \sigma_{sl,0} \beta_{sl}) \right] g(\phi)}+ {\ln(1-\psi) \frac{\sigma_{sg,0} \beta_{sg} + \sigma_{sl,0} \beta_{sl}}{2}} + {\frac{\sigma_{sg,0} + \sigma_{sl,0}}{2}},
		\end{aligned}
		\label{eq8}
	\end{equation}
	where Young's equation for the clean system is used to simplify the first term. Following the previous studies \cite{YueJCP2010,ZhuJFM2019}, we derive $g(\phi) = 3\phi_b^{-1}\phi / 4 - \phi_b^{-3}\phi^3 / 4$. It can be found that the expression (\ref{eq8}) clearly delineates three different contributions: the first term governs the surfactant-modified contact angle behavior, the second one captures the overall shift in solid–fluid interfacial tension induced by surfactant adsorption, and the third part represents the clean-system baseline. 
	
	
	The variations of the total free-energy functional $\mathcal{E}_f(\phi,\psi)$ with respect to $\phi$ and $\psi$, respectively, are given by
	\begin{subequations}
		\begin{equation}
			\delta_\phi \mathcal{E}_f(\phi,\psi) = \int_{\Omega} \left[ 4\beta(1-\psi)(\phi^3 - \phi) - \kappa \nabla^2\phi + W\psi\phi \right]\delta \phi d\Omega + \int_{\partial\Omega} \left( \kappa \mathbf{n} \cdot \nabla\phi + \frac{\partial f_w}{\partial \phi} \right)\delta \phi dS,
		\end{equation}
		\begin{equation}
			\delta_\psi \mathcal{E}_f(\phi,\psi) = \int_{\Omega} \left[ k_B T \ln\left(\frac{\psi}{1-\psi}\right) - \lambda_s\nabla^2\psi - \beta(1-\phi^2)^2 + \frac{W}{2}\phi^2 \right] \delta \psi d\Omega + \int_{\partial\Omega} \left( \lambda_s\mathbf{n} \cdot \nabla\psi + \frac{\partial f_w}{\partial \psi} \right) \delta \psi dS.
		\end{equation}	
		\label{free_energy}		
	\end{subequations}	
	Here, the integration by parts has been used, and $\mathbf n$ denotes the unit outward normal vector on the solid boundary $\partial\Omega$ (pointing from the fluid domain towards the solid substrate). Then, the chemical potentials of the phase-field variable and the surfactant concentration can be determined as 
	\begin{subequations}
		\begin{equation}
			\mu_\phi = 4\beta(1-\psi)(\phi^3 - \phi) - \kappa \nabla^2\phi + W\psi\phi,
			\label{eq_muphi}
		\end{equation}
		\begin{equation}
			\mu_\psi = k_B T \ln\left(\frac{\psi}{1-\psi}\right) - \lambda_s\nabla^2\psi 	- \beta(1-\phi^2)^2 + \frac{W}{2}\phi^2.
			\label{eq_mupsi}
		\end{equation}			
	\end{subequations}

	To model the moving interface and surfactant transport, we employ a unified Allen–Cahn framework augmented with nonlocal Lagrange multipliers. To this end, we introduce $c_i(\mathbf{x}, t)$ to denote a generic scalar field, representing either the phase field $\phi$ or the surfactant concentration $\psi$, which is governed by the unified conservative Allen-Cahn equation
	\begin{equation}
		\frac{\partial c_i}{\partial t} + \nabla \cdot (c_i \mathbf{u}) = -M_i \left[ \mu_i - \lambda_i(t) h_i'(c_i) \right],
		\label{ac1}
	\end{equation}	
	where $\mathbf{u}$ is the fluid velocity, $M_i$, $\mu_i$, and $\lambda_i(t)$ are the mobility, chemical potential, and time-dependent Lagrange multiplier. We note that in some previous conservative Allen–Cahn models \cite{YangCMAME2021,LiangJCP2026,RubinsteinIMAJAM1992,BrasselMMAS2011,KimIJES2014,KwakAML2022,ChoiEABE2023,YangCNSNS2024,HwangEABE2026}, the conservation constraints are imposed by preserving the spatial integral of the order parameter, i.e., $\int_{\Omega} c_i \, d\Omega=constant$. However, this standard linear formulation suffers from an unphysical shrinkage artifact: curvature-induced bulk plateau shifts drive an artificial loss of geometric volume for high-curvature droplets \cite{BrasselMMAS2011}. To improve geometric volume conservation, we introduce a more general conserved quantity defined through a tailored monotonic mapping  $h_i(c_i)$ \cite{ZhouPRSA2025,MusilArxiv2026}, subject to the following constraint,
	\begin{equation}
		\frac{\mathrm{d}}{\mathrm{d}t}\int_\Omega h_i(c_i) \, \mathrm{d}\Omega = \int_\Omega h_i'(c_i) \frac{\partial c_i}{\partial t} \, \mathrm{d}\Omega = 0.
		\label{conservation}
	\end{equation}
	By combining Eqs. (\ref{ac1}) and (\ref{conservation}), we obtain
	\begin{equation}
		\int_\Omega h_i'(c_i) \frac{\partial c_i}{\partial t} \, \mathrm{d}\Omega = -\int_\Omega h_i'(c_i) \nabla \cdot (c_i \mathbf{u}) \, \mathrm{d}\Omega - M_i \int_\Omega h_i'(c_i) \mu_i \, \mathrm{d}\Omega + M_i \lambda_i(t) \int_\Omega [h_i'(c_i)]^2 \, \mathrm{d}\Omega = 0.
		\label{eq13}
	\end{equation}
	With the help of the chain rule, the divergence theorem, and the no-penetration boundary condition $\mathbf u\cdot\mathbf n=0$ on $\partial\Omega$, we have
	\begin{equation}
		-\int_{\Omega} h_i'(c_i) \nabla \cdot (c_i \mathbf{u}) \, d\Omega
		= -\int_{\Omega} h_i'(c_i) \mathbf{u} \cdot \nabla c_i \, d\Omega
		= -\int_{\Omega} \mathbf{u} \cdot \nabla h_i(c_i) \, d\Omega
		= -\int_{\Omega} \nabla \cdot \left[ h_i(c_i) \mathbf{u} \right] \, d\Omega
		= -\int_{\partial \Omega} h_i(c_i) (\mathbf{u} \cdot \mathbf{n}) \, dS
		= 0,
	\end{equation}
	where the incompressible condition ($\nabla \cdot \mathbf{u} = 0$) has been used. Then from Eq. (\ref{eq13}) one can derive the generalized Lagrange multiplier,
	\begin{equation}
		\lambda_i(t) = \frac{\int_\Omega h_i'(c_i) \mu_i \, \mathrm{d}\Omega}{\int_\Omega [h_i'(c_i)]^2 \, \mathrm{d}\Omega}.
		\label{eq15}
	\end{equation}
	
	Although both variables evolve within the same constrained Allen–Cahn framework, they are subject to different conservation requirements: the phase-field constraint preserves a diffuse approximation of the geometric phase volume, whereas the surfactant constraint ensures conservation of the total surfactant mass. For this reason, different constraint functions are adopted for $\phi$ and $\psi$. Actually, for the surfactant concentration $\psi$, the total surfactant mass over the entire domain, $\int_\Omega \psi \, \mathrm{d}\Omega$, must be strictly conserved, and hence a linear function $h_\psi(\psi)=\psi$ (i.e., $h_\psi'(\psi)=1$) is the natural choice. This constraint based on the linear function $h_\psi$ applies the mass-conservation correction uniformly across the entire computational domain, acting as a spatially uniform shift in the chemical potential. Substituting $h_\psi'(\psi)=1$ into the multiplier formula (\ref{eq15}), $\lambda_\psi$ would reduce to the spatial average of the chemical potential, $\bar{\mu}_\psi = \frac{1}{\vert{}\Omega\vert{}} \int_\Omega \mu_\psi \, \mathrm{d}\Omega$. Consequently, the governing equation for the surfactant concentration simplifies to
	\begin{equation}
		\frac{\partial \psi}{\partial t} + \nabla \cdot (\psi \mathbf{u}) = -M_\psi (\mu_\psi - \bar{\mu}_\psi).
		\label{AC_psi}
	\end{equation}
	In contrast, the geometric volume conservation is required for the phase field $\phi$. If a global constraint based on a linear function $h_\phi$ ($h_\phi(\phi) = \phi$, $h_\phi'(\phi) = 1$) is applied, a spatially uniform correction would be uniformly distributed throughout the entire domain. This artificially shifts the bulk values away from their equilibrium minima ($\pm \phi_b$), thereby forcing an unphysical displacement of the interface to maintain the global integral \cite{BrasselMMAS2011}. To prevent this unphysical bulk shift, the volume correction is designed to be localized within the diffuse-interface region rather than distributed uniformly across the entire domain. Therefore, we consider a nonlinear function $h_\phi (\phi)=0.5+g(\phi)=0.5+3\phi_b^{-1}\phi / 4 - \phi_b^{-3}\phi^3 / 4$ with $h_\phi(-\phi_b) = 0$, $h_\phi(\phi_b) = 1$, and $h_\phi'(\pm \phi_b) = 0$, ensuring the Lagrange-multiplier correction to vanish in the pure bulk phases. In this case, the phase-field equation with the geometric volume correction takes the following form
	\begin{equation}
		\frac{\partial \phi}{\partial t} + \nabla \cdot (\phi \mathbf{u}) = -M_\phi \left( \mu_\phi - \tilde{\mu}_\phi \right), \quad \tilde{\mu}_\phi=\lambda_\phi(t) h_\phi'(\phi).
		\label{AC_phi}
	\end{equation}	
	Owing to the weighting function $h_\phi'(\phi)$ vanishing at the bulk equilibrium states $\phi = \pm \phi_b$, the volume correction term $\tilde{\mu}_\phi = \lambda_\phi(t) h_\phi'(\phi)$ is predominantly localized within the diffuse interface and decays exponentially to zero in the bulk phases. Additionally, the boundary conditions for phase-field variable $\phi$ and surfactant concentration $\psi$ are derived from the variations of the total free energy, 
	\begin{subequations}
		\begin{equation}
			\mathcal{B}_\phi = \kappa \mathbf{n} \cdot \nabla\phi + \frac{\partial f_w}{\partial \phi}=0, \quad \quad \frac{\partial f_w}{\partial \phi} = -\left[ \sigma_{gl,0}\cos\theta_0 + \ln(1-\psi)(\sigma_{sg,0}\beta_{sg} - \sigma_{sl,0}\beta_{sl}) \right] g'(\phi),
		\end{equation}
		\begin{equation}
			\mathcal{B}_\psi = \lambda_s \mathbf{n} \cdot \nabla\psi  + \frac{\partial f_w}{\partial \psi}=0, \quad \quad \frac{\partial f_w}{\partial \psi} = \frac{1}{1-\psi} \left[ (\sigma_{sg,0}\beta_{sg} - \sigma_{sl,0}\beta_{sl})g(\phi) - \frac{\sigma_{sg,0}\beta_{sg} + \sigma_{sl,0}\beta_{sl}}{2} \right],
		\end{equation}	
		\label{boundary}		
	\end{subequations}	
	where the first equation represents the generalized Young boundary condition, which balances the diffuse-interface gradient contribution against the wall free energy $f_w$ and determines the equilibrium contact angle. The second equation, in turn, denotes the local thermodynamic equilibrium condition for the surfactant at the solid–fluid interface, balancing the diffuse-layer gradient contribution against the wall free-energy contribution and thereby determining the equilibrium surfactant adsorption on the substrate.

	Finally, the incompressible Navier–Stokes equations are applied to describe two-phase flows,
	\begin{subequations}
		\begin{equation}
			\nabla \cdot \mathbf{u} = 0,
		\end{equation}
		\begin{equation}
			\frac{\partial(\rho \mathbf{u})}{\partial t} + \nabla \cdot (\rho \mathbf{uu}) = -\nabla p + \nabla \cdot \left[ \mu \left( \nabla \mathbf{u} + (\nabla \mathbf{u})^\mathrm{T} \right) \right] + \mathbf{F}_s + \mathbf{G},
		\end{equation}	
		\label{eqns}		
	\end{subequations}	
	where $\rho$, $p$, and $\mu$ are the density, pressure, and dynamic viscosity, respectively. $\mathbf{G}$ is the body force. Based on the formulation of surface tension force for clean two-phase flows, the potential-form expression for surfactant-laden two-phase systems is adopted \cite{YangCMAME2021,LiangJCP2026},
	\begin{equation}
		\mathbf{F}_s = -\phi\nabla\mu_\phi - \psi\nabla\mu_\psi.
		\label{eq_Fs0}
	\end{equation}
	
	\noindent \textbf{Remark 1} In the surfactant–substrate systems, the solid–fluid interfacial tensions approach a saturation limit at high surfactant coverage rather than decreasing indefinitely \cite{ChangCSA1995,JuJML2017,SoligoJCP2019}. In addition, the logarithmic term in the standard Langmuir–Szyszkowski equation diverges as $\psi \to 1$. To account for this adsorption saturation and prevent numerical singularity at a high local surfactant concentration, the solid–fluid interfacial tensions are prevented from falling below half of their clean-surface values \cite{KannanJFM2026}, i.e., $\sigma_{s\alpha}(\psi) = \sigma_{s\alpha,0} \max\left[0.5, 1+\beta_{s\alpha}\ln(1-\psi)\right]$ for $\alpha \in \{g,l\}$. In the present study, to retain a smooth free-energy functional and thereby preserve thermodynamic consistency and the energy-dissipation law established in Theorem~1, the wall free-energy density $f_w(\phi,\psi)$ and its variational derivatives are formulated using the standard smooth logarithmic relation. In all simulations performed herein, the local wall surfactant concentration remains below the saturation threshold (i.e., $1+\beta_{s\alpha}\ln(1-\psi_w)>0.5$). Therefore, the system remains entirely within the smooth, unsaturated regime without reaching the nonsmooth cutoff, ensuring that the variational boundary conditions in Eq. (\ref{boundary}) and the energy-dissipation law remain rigorously valid throughout the computation.

	\subsection{The thermodynamic equilibrium}

	At thermodynamic equilibrium state, the surfactant chemical potential is spatially uniform, whereas the phase-field variable satisfies the constraint equilibrium condition. Based on these equilibrium conditions, the corresponding equilibrium profiles of $\phi$ and $\psi$ can be derived, together with the interfacial adsorption isotherm and the modified Young’s relation at the solid boundary. In the bulk phases, the surfactant chemical potential can be simplified by
	\begin{equation}
		\mu_{\psi, b} = k_B T \ln\left( \frac{\psi_b}{1-\psi_b} \right) - \beta(1-\phi_b^2)^2 + \frac{W}{2}\phi_b^2,
		\label{eq_psib}
	\end{equation}	
	where $\phi_b$ denotes the bulk value of the phase field and $\psi_b$ denotes the bulk surfactant concentration. Under the equilibrium condition, i.e., $\mu_{\phi, b}=4\beta(1-\psi_b)(\phi_b^3 - \phi_b) + W\psi_b\phi_b = 0$, one can obtain $\phi_b^2 = 1 - W\psi_b/[4\beta(1-\psi_b)]$. The thermodynamic equilibrium indicates that the surfactant chemical potential remains uniform throughout the domain, i.e., $\mu_\psi(x) = \mu_{\psi, b}$. Combining Eqs. (\ref{eq_mupsi}) and (\ref{eq_psib}) yields
	\begin{equation}
		k_B T \ln\left(\frac{\psi}{1-\psi}\right) - \lambda_s\nabla^2\psi - \beta(1-\phi^2)^2 + \frac{W}{2}\phi^2 = k_B T \ln\left(\frac{\psi_b}{1-\psi_b}\right) - \beta(1-\phi_b^2)^2 + \frac{W}{2}\phi_b^2,
	\end{equation}
	for which the local surfactant concentration $\psi(x)$ can be expressed as
	\begin{equation}
		\psi(x)=\frac{\psi_b}{\psi_b+(1-\psi_b)\exp\left[-R(\phi,\psi)/(k_B T)\right]},
	\end{equation}
	where $R(\phi,\psi) = (\phi_b^2 - \phi^2) \left[ \beta(2 - \phi_b^2 - \phi^2) + W/2 \right]+\lambda_s\nabla^2\psi$ encapsulates the interfacial coupling and square-gradient regularization effects. At the center of the fluid–fluid interface ($x=0$), we have $\phi(0)=0$, $\phi_b \approx 1$, and $\psi_b\ll1$. Substituting these results into $R(\phi,\psi)$ allows the exponential term to be written as $\exp\left[-R(0,\psi_0)/(k_B T)\right] = \psi_c \mathcal{R}_s$, yielding the equilibrium interfacial concentration,
	\begin{equation}
		\psi_0 = \frac{\psi_b}{\psi_b+(1-\psi_b)\psi_c\mathcal{R}_s} = \frac{\psi_b}{\psi_b+\psi_c\mathcal{R}_s} +\mathcal{O}(\psi_b), \quad \psi_c = \exp\left[-\frac{\beta+W/2}{k_B T}\right], \quad \mathcal{R}_s = \exp\left[-\frac{\lambda_s\nabla^2\psi_0}{k_B T}\right].
	\end{equation}
	Here, $\psi_c$ corresponds to the classical Langmuir adsorption constant, whereas $\mathcal{R}_s$ represents a nonlocal correction factor induced by the surfactant square-gradient penalty. When the gradient term is omitted ($\lambda_s=0, \mathcal{R}_s=1$), the expression reduces to the standard closed-form Langmuir adsorption isotherm, $\psi_0=\psi_b/(\psi_b+\psi_c)+\mathcal{O}(\psi_b)$. For a finite $\lambda_s$, however, $\mathcal{R}_s$ intrinsically depends on the local Laplacian of the surfactant concentration $\nabla^2\psi_0$, precluding an explicit analytical expression for the isotherm. Nevertheless, for the parameter range considered, our numerical results demonstrate that the fluid--fluid interfacial tension is well described by the Langmuir--Szyszkowski relation, without evident systematic deviation and in good agreement with the results reported by Kannan et al. \cite{KannanJFM2026} and Engblom et al. \cite{EngblomCICP2013}.

	We next derive the equilibrium profile of the phase-field variable. In the absence of surfactant, the chemical potential can be written as
	\begin{equation}
		\mu_{\phi, 0} = 4\beta(\phi^3 - \phi) - \kappa \nabla^2\phi = 4\beta \left( \phi^3 - \phi - \frac{\varepsilon^2}{8}\nabla^2\phi \right).
		\label{eq_T1}
	\end{equation}
	The corresponding equilibrium solution is 
	$\phi_0(x) = \tanh (2x/\varepsilon)$, with $\varepsilon = \sqrt{2 \kappa/\beta}$ being the characteristic interface thickness parameter. However, the presence of surfactant introduces strong nonlinear coupling into the phase-field chemical potential $\mu_\phi$. To avoid directly solving $\mu_\phi=0$, we approximate the surfactant concentration entering the phase-field equilibrium relation by its bulk value, $\psi\approx\psi_b$ \cite{LiangJCP2026}. Under this approximation, Eq. (\ref{eq_muphi}) can be simplified as
	\begin{equation}
		\mu_\phi = (4\beta - 4\beta\psi)\phi_b^3 \left[ (\phi_b^{-1}\phi)^3 - (\phi_b^{-1}\phi) - \frac{\varepsilon_\phi^2}{8}\nabla^2(\phi_b^{-1}\phi) \right],
		\label{eq_T2}
	\end{equation}
	where $\varepsilon_\phi =\varepsilon \phi_b^{-1} / \sqrt{1-\psi_b}$ is the effective interface-thickness parameter. Since Eq. (\ref{eq_T2}) has the same form as Eq. (\ref{eq_T1}), the classical hyperbolic tangent solution can be directly applied, giving $\phi_b^{-1}\phi = \tanh(2x/\varepsilon_\phi)$. Substituting $\varepsilon_\phi$ back into the solution yields an approximate equilibrium profile of the phase-field variable in the coupled surfactant system,
	\begin{equation}
		\phi(x) = \phi_b \tanh \left( \frac{2 \phi_b \sqrt{1-\psi_b} \, x}{\varepsilon} \right).
	\end{equation}
	
	\subsection{Formal second-order accuracy of geometric-volume conservation}

	We now examine how accurately the exactly conserved diffuse-volume measure approximates the geometric measure enclosed by the zero-level set of the phase field. Let $\Omega \subset \mathbb{R}^d$ ($d=2,3$) be a fixed, smooth, and bounded domain. We define the normalized phase-field variable $u(\mathbf{x},t)=\phi(\mathbf{x},t)/\phi_b \in [-1,1]$, where $\phi_b>0$ is a fixed reference bulk value throughout the thin-interface asymptotic analysis. The zero-level set $\Gamma_\varepsilon(t)$ and the phase regions with positive and negative values are defined as
	\begin{equation}
		\Gamma_\varepsilon(t) = \{\mathbf{x}\in\Omega : u(\mathbf{x},t) = 0\}, \quad \Omega_\varepsilon^+(t) = \{\mathbf{x}\in\Omega : u(\mathbf{x},t) > 0\}, \qquad \Omega_\varepsilon^-(t) = \{\mathbf{x}\in\Omega : u(\mathbf{x},t) < 0\},
	\end{equation}
	and the geometric volume (or area in two-dimensional space) of the phase with positive value is $ \vert{}\Omega_\varepsilon^+(t)\vert{}$. Additionally, we define the intrinsic diffuse-interface transition thickness $\ell_\phi$ and the interfacial-energy scale parameter $\gamma_\phi$ as 
	\begin{equation}
		\ell_\phi = \sqrt{\frac{\kappa}{2\beta_{\text{eff}}\phi_b^2}} = \frac{\varepsilon}{2\phi_b\sqrt{1-\psi_b}} = \mathcal{O}(\varepsilon), \quad \gamma_\phi = \frac{\kappa\phi_b}{\ell_\phi} = 2\beta_{\text{eff}}\phi_b^3\ell_\phi = \mathcal{O}(1),
	\end{equation}
	where $\beta_{\text{eff}} = \beta(1-\psi_b)$. Here, the scaling $\gamma_\phi=\mathcal O(1)$ corresponds to the standard thin-interface scaling in which the interfacial-energy scale remains finite as $\varepsilon\to 0$. The nonlinear constrained Allen–Cahn equation (\ref{AC_phi}) and the chemical potential (\ref{eq_muphi}) can be rewritten as
	\begin{equation}
		\phi_b\left( \frac{\partial u}{\partial t} + \mathbf{u}\cdot\nabla u \right) = -M_\phi \left[ \mu_\phi - \lambda_\phi(t) w(u) \right], \quad \mu_\phi = \frac{\gamma_\phi}{\ell_\phi} \left[ 2(u^3 - u) - \ell_\phi^2\nabla^2 u \right],
		\label{eq_AC}
	\end{equation}
	where $\mathbf{u}$ denotes the fluid velocity satisfying $\nabla \cdot \mathbf{u} = 0$ in $\Omega$ and $\mathbf{u}\cdot\mathbf{n} = 0$ on $\partial\Omega$. The weight function $w(u)$ and the nonlocal Lagrange multiplier $\lambda_\phi(t)$ are given by
	\begin{equation}
		w(u) = h_\phi'(\phi_b u) = \frac{3}{4\phi_b}(1-u^2), \quad \lambda_\phi(t) = \frac{\int_\Omega w(u)\mu_\phi\,\mathrm{d}\Omega}{\int_\Omega [w(u)]^2\,\mathrm{d}\Omega}.
		\label{eq_wu}
	\end{equation}
	We now introduce the normalized odd constraint function $Q(u) = 2h_\phi(\phi_b u) - 1 = \frac{3}{2}u - \frac{1}{2}u^3$, which satisfies $Q(-u) = -Q(u)$, $Q(\pm 1) = \pm 1$, $Q'(u) = 3(1-u^2)/2$, and $Q'(\pm 1) = 0$. 
	Consequently, the exact conservation law (\ref{conservation}) at the continuous level established for the model is expressed as
	\begin{equation}
		\int_\Omega Q(u(\mathbf{x},t))\,\mathrm{d}\Omega = \mathcal{C}_0,
		\label{conservation_Q}
	\end{equation}
	where $\mathcal{C}_0$ is a constant. The asymptotic analysis is conducted under the following assumptions: (i) The flow field is incompressible and satisfies the no-penetration boundary condition. (ii) All model parameters are independent of time. (iii) The interface $\Gamma_\varepsilon(t)$ remains a smooth, closed hypersurface separated from the solid boundaries, with uniformly bounded surface measure and principal curvatures. (iv) The interface $\Gamma_\varepsilon(t)$ admits a tubular neighborhood of uniform width independent of $\ell_\phi$. (v) The phase-field solution is well prepared and admits a formal matched asymptotic expansion. The geometric assumptions in (iii) and (iv) exclude contact-line boundary layers, topological pinch-off, and interface–boundary collisions. Under these conditions, the phase field variable admits the following second-order uniform asymptotic approximation,
	\begin{equation}
		u_\varepsilon(\mathbf x,t)
		=q_0\!\left(\frac{r(\mathbf x,t)}{\ell_\phi}\right)
		+\ell_\phi u_{\text{unif},1}(\mathbf x,t)
		+\mathcal{O}(\ell_\phi^2),
		\label{eq:gv_uniform_expansion}
	\end{equation}
	where $r(\mathbf{x},t)$ denotes the signed distance to $\Gamma_\varepsilon(t)$ and $q_0(z)=\tanh z$ is the leading-order interfacial profile. The first-order correction $u_{\text{unif},1}$ is uniformly bounded throughout $\Omega$, satisfying
	\begin{equation}
		\|u_{\text{unif},1}\|_{L^\infty(\Omega)}\le C, 
		\label{eq:gv_q0_u1_bound}
	\end{equation}
	where $C$ is a generic positive constant independent of $\ell_\phi$ and time. Moreover, $u_{\text{unif},1}$ decays exponentially away from the diffuse interface and vanishes asymptotically in the bulk phases. A detailed analysis of the matched asymptotic expansion is provided in \ref{APA} \cite{LiuMMS2022}. Substituting the uniform asymptotic expansion (\ref{eq:gv_uniform_expansion}) into the global conservation law (\ref{conservation_Q}) and expanding $Q(u_\varepsilon)$ in a Taylor series yields
	\begin{equation}
		\mathcal{C}_0 = \int_\Omega Q(u_\varepsilon)\,\mathrm{d}\Omega=\int_\Omega Q\left(q_0\left(\frac{r}{\ell_\phi}\right)\right)\mathrm{d}\Omega + \ell_\phi \int_\Omega Q'\left(q_0\left(\frac{r}{\ell_\phi}\right)\right) u_{\text{unif},1}\,\mathrm{d}\Omega+ \mathcal{O}(\ell_\phi^2) = I + II + \mathcal{O}(\ell_\phi^2),
	\end{equation}
	where $I$ and $II$ denote the leading contribution and that arising from the first-order correction in the phase-field expansion, respectively. Let $\mathcal{A}(\zeta,t) := \mathcal{H}^{d-1}(\{\mathbf{x} \in \Omega : r(\mathbf{x},t) = \zeta\})$ denote the $(d-1)$-dimensional Hausdorff measure of the parallel level sets, and extend $\mathcal A(\zeta,t)$ by zero outside the attainable range of the signed-distance function in $\Omega$. Applying the coarea formula, the integral $I$ is represented as 
	\begin{equation}
		I = \int_{-\infty}^{+\infty} Q\left(q_0\left(\frac{\zeta}{\ell_\phi}\right)\right) \mathcal{A}(\zeta,t)\,\mathrm{d}\zeta.
	\end{equation}
	With the sign convention $r>0$ in $\Omega_\varepsilon^+$ and $r<0$ in $\Omega_\varepsilon^-$, the exact geometric phase volumes satisfy $\vert{}\Omega_\varepsilon^+(t)\vert{} = \int_0^{+\infty} \mathcal{A}(\zeta, t)\,\mathrm{d}\zeta$ and $\vert{}\Omega_\varepsilon^-(t)\vert{} = \int_{-\infty}^0 \mathcal{A}(\zeta, t)\,\mathrm{d}\zeta$, then we can decompose $I$ as
	\begin{equation}
		\begin{aligned} I ={} \int_0^{+\infty} \mathcal{A}(\zeta,t)\,\mathrm{d}\zeta + \int_0^{+\infty} \left[ Q\left( q_0\left(\frac{\zeta}{\ell_\phi}\right) \right) - 1 \right] \mathcal{A}(\zeta,t)\,\mathrm{d}\zeta - \int_{-\infty}^0 \mathcal{A}(\zeta,t)\,\mathrm{d}\zeta + \int_{-\infty}^0 \left[ Q\left( q_0\left(\frac{\zeta}{\ell_\phi}\right) \right) + 1 \right] \mathcal{A}(\zeta,t)\,\mathrm{d}\zeta. \end{aligned}
	\end{equation}
	Because $q_0(z) = \tanh z$ and $Q(u) = \frac{3}{2}u - \frac{1}{2}u^3$ are both odd functions, we have $Q(q_0(-z)) = -Q(q_0(z))$. Making the change of variable $\zeta \mapsto -\zeta$ in the negative half-line integral yields
	\begin{equation}
		I = \vert{}\Omega_\varepsilon^+(t)\vert{} - \vert{}\Omega_\varepsilon^-(t)\vert{} + \int_0^{+\infty} \left[ Q\left( q_0\left(\frac{\zeta}{\ell_\phi}\right) \right) - 1 \right] \left[ \mathcal{A}(\zeta,t) - \mathcal{A}(-\zeta,t) \right] \mathrm{d}\zeta.
	\end{equation}
	Introducing the stretched variable $\eta = \zeta/\ell_\phi$ gives rise to the following relation,
	\begin{equation}
		I = \vert{}\Omega_\varepsilon^+(t)\vert{} - \vert{}\Omega_\varepsilon^-(t)\vert{} + \ell_\phi \int_0^{+\infty} [Q(q_0(\eta)) - 1] [\mathcal{A}(\ell_\phi\eta,t) - \mathcal{A}(-\ell_\phi\eta,t)]\,\mathrm{d}\eta.
	\end{equation}
	By the tubular neighborhood theorem, there exists a reach radius $\delta > 0$ such that for $\vert{}\zeta\vert{} < \delta$, the level set area obeys the first variation expansion $\mathcal{A}(\zeta,t) - \mathcal{A}(-\zeta,t) = 2\zeta \mathcal{A}'(0,t) + \mathcal{O}(\zeta^2)$. Thus, there exists a constant $C$ such that
	\begin{equation}
		\vert{}\mathcal{A}(\ell_\phi\eta,t) - \mathcal{A}(-\ell_\phi\eta,t)\vert{} \le C \ell_\phi \eta, \qquad 0 \le \eta \le \delta/\ell_\phi.
	\end{equation}
	Moreover, the algebraic identity $1 - Q(q_0(\eta)) = \frac{1}{2}(1 - \tanh\eta)^2(2 + \tanh\eta)$ implies that $\vert{}1 - Q(q_0(\eta))\vert{} \le C e^{-4\eta}$. Therefore, the near-interface contribution can be estimated by
	\begin{equation}
		\begin{aligned}
			\ell_\phi \left\vert{} \int_0^{\delta/\ell_\phi} [Q(q_0(\eta)) - 1] [\mathcal{A}(\ell_\phi\eta,t) - \mathcal{A}(-\ell_\phi\eta,t)]\,\mathrm{d}\eta \right\vert{} \le C \ell_\phi^2 \int_0^{\delta/\ell_\phi} \eta \vert{}1 - Q(q_0(\eta))\vert{}\,\mathrm{d}\eta \le C \ell_\phi^2 \int_0^\infty \eta e^{-4\eta}\,\mathrm{d}\eta = \mathcal{O}(\ell_\phi^2).
		\end{aligned}
	\end{equation}  
	For the far-field tail $\vert{}\zeta\vert{} \ge \delta$, the exponential profile ensures $\vert{}Q(q_0(\zeta/\ell_\phi)) - \operatorname{sgn}(\zeta)\vert{} \le C e^{-4\delta/\ell_\phi}$. Since $\delta = \mathcal{O}(1)$, the contribution is $\mathcal{O}(e^{-4\delta/\ell_\phi}) = \mathcal{O}(\ell_\phi^\infty)$, which is transcendentally small. Consequently,
	\begin{equation}
		I = \vert{}\Omega_\varepsilon^+(t)\vert{} - \vert{}\Omega_\varepsilon^-(t)\vert{} + \mathcal{O}(\ell_\phi^2).
		\label{eq_I}
	\end{equation}
	
	Applying the coarea formula, using the uniform bound $\Vert{}u_{\text{unif},1}\Vert{}_{L^\infty(\Omega)} \le C$, and introducing the stretched coordinate $\eta = \zeta/\ell_\phi$, the contribution \(II\) arising from the first-order phase-field correction can be bounded as
	\begin{equation}
		\begin{aligned} \vert{}II\vert{} \le \ell_\phi \Vert{}u_{\text{unif},1}\Vert{}_{L^\infty(\Omega)} \int_{-\infty}^{+\infty} \left\vert{} Q'\left(q_0\left(\frac{\zeta}{\ell_\phi}\right)\right) \right\vert{} \mathcal{A}(\zeta,t)\,\mathrm{d}\zeta \le C \ell_\phi^2 \int_{-\infty}^{+\infty} \vert{}Q'(q_0(\eta))\vert{} \mathcal{A}(\ell_\phi \eta, t)\,\mathrm{d}\eta = C \ell_\phi^2 J(t). \end{aligned}
	\end{equation}
	We now split the integral using the uniform tubular-neighborhood radius $\delta>0$,
	\begin{equation}
		J(t) = \int_{\vert{}\eta\vert{} < \delta/\ell_\phi} \vert{}Q'(q_0(\eta))\vert{} \mathcal{A}(\ell_\phi \eta, t)\,\mathrm{d}\eta + \int_{\vert{}\eta\vert{} \ge \delta/\ell_\phi} \vert{}Q'(q_0(\eta))\vert{} \mathcal{A}(\ell_\phi \eta, t)\,\mathrm{d}\eta.
	\end{equation}
	In the near-interface region ($\vert{}\eta\vert{} < \delta/\ell_\phi$), since the parallel hypersurfaces within the tubular neighborhood $\vert{}\zeta\vert{} < \delta$ are smooth and satisfy $\sup_{\vert{}\zeta\vert{}<\delta} \mathcal{A}(\zeta,t) \le C_{\mathcal{A}} < \infty$, using $Q'(q_0(\eta)) = \frac{3}{2}\operatorname{sech}^2\eta$ yields
	\begin{equation}
		\int_{\vert{}\eta\vert{} < \delta/\ell_\phi} \left\vert{} Q'(q_0(\eta)) \right\vert{} \mathcal{A}(\ell_\phi \eta, t)\,\mathrm{d}\eta\le C_{\mathcal{A}} \int_{-\delta/\ell_\phi}^{\delta/\ell_\phi} \left\vert{} Q'(q_0(\eta)) \right\vert{}\,\mathrm{d}\eta \le C_{\mathcal{A}} \int_{-\infty}^{+\infty} \left\vert{} Q'(q_0(\eta)) \right\vert{}\,\mathrm{d}\eta= 3 C_{\mathcal{A}} = \mathcal{O}(1).
		\label{eq_nei}
	\end{equation}
	In the far-field region ($\vert{}\eta\vert{} \ge \delta/\ell_\phi$), the integral is evaluated by transforming back to the physical coordinate $\zeta = \ell_\phi \eta$. Utilizing the exponential decay $\vert{}Q'(q_0(\zeta/\ell_\phi))\vert{} \le 6e^{-2\delta/\ell_\phi}$ along with the domain-volume identity $\int_{-\infty}^{+\infty} \mathcal{A}(\zeta,t)\,\mathrm{d}\zeta = \int_\Omega \vert{}\nabla r\vert{}\,\mathrm{d}\Omega = \vert{}\Omega\vert{}$, we can obtain
	\begin{equation}
		\begin{aligned}
			\int_{\vert{}\eta\vert{} \ge \delta/\ell_\phi} \vert{}Q'(q_0(\eta))\vert{} \mathcal{A}(\ell_\phi \eta, t)\,\mathrm{d}\eta &= \frac{1}{\ell_\phi} \int_{\vert{}\zeta\vert{} \ge \delta} \left\vert{} Q'\left(q_0\left(\frac{\zeta}{\ell_\phi}\right)\right) \right\vert{} \mathcal{A}(\zeta,t)\,\mathrm{d}\zeta \\ &\le \frac{6e^{-2\delta/\ell_\phi}}{\ell_\phi} \int_{-\infty}^{+\infty} \mathcal{A}(\zeta,t)\,\mathrm{d}\zeta = \frac{6\vert{}\Omega\vert{}}{\ell_\phi} e^{-2\delta/\ell_\phi} = o(1).
		\end{aligned}
		\label{eq_wai}
	\end{equation}
	Combining Eqs. (\ref{eq_nei}) and (\ref{eq_wai}) confirms that $J(t) = \mathcal{O}(1)$, which directly gives
	\begin{equation}
		II = \mathcal{O}(\ell_\phi^2) = \mathcal{O}(\varepsilon^2).
		\label{eq_II}
	\end{equation}
	Additionally, combining Eqs. (\ref{eq_I}) and (\ref{eq_II}) with the domain volume identity $\vert{}\Omega_\varepsilon^+(t)\vert{} + \vert{}\Omega_\varepsilon^-(t)\vert{} = \vert{}\Omega\vert{}$ yields
	\begin{equation}
		\vert{}\Omega_\varepsilon^+(t)\vert{} = \frac{\vert{}\Omega\vert{} + \mathcal{C}_0}{2} + \mathcal{O}(\ell_\phi^2).
	\end{equation}
	Since $\vert{}\Omega\vert{}$ and $\mathcal{C}_0$ are time-invariant, subtracting the initial state at $t = 0$ and using $\ell_\phi = \mathcal{O}(\varepsilon)$, we have 
	\begin{equation}
		\vert{}\Omega_\varepsilon^+(t)\vert{} = \vert{}\Omega_\varepsilon^+(0)\vert{} + \mathcal{O}(\varepsilon^2),
		\label{second}
	\end{equation} 
	which shows the second-order accuracy in preserving the enclosed sharp-interface geometric volume throughout the evolution.

	\subsection{The energy dissipation law}
	In this part, we demonstrate that the proposed coupled phase-field–Navier–Stokes system for surfactant-laden two-phase flow is thermodynamically consistent and satisfies the second law of thermodynamics. For an isothermal multiphase system without external energy input, thermodynamic consistency requires the total energy, consisting of the free and kinetic energies, to be non-increasing in time. Here, the total energy of the system is defined as $\mathcal{E}_{total} = \mathcal{E}_{f}(\phi, \psi) + \mathcal{E}_{K}$, where $\mathcal{E}_{f}(\phi, \psi)$ is the total free energy functional defined by Eq. (\ref{eq1}), and $\mathcal{E}_{K} = \int_{\Omega} \frac{1}{2}\rho|\mathbf{u}|^2 d\Omega$ is the kinetic energy.
	\begin{theorem}
		Consider the coupled phase-field–Navier–Stokes system for surfactant laden two-phase flow in a closed domain $\Omega$, and assume that the generalized natural boundary conditions $\mathcal{B}_{\phi} = 0$, $\mathcal{B}_{\psi} = 0$ and the no-slip boundary condition $\mathbf{u} = \mathbf{0}$ on $\partial\Omega$ are imposed. Then, in the absence of external force, the coupled system satisfies the energy-dissipation law,
		\begin{equation}
			\frac{d\mathcal{E}_{total}}{dt} \le 0.
		\end{equation}
	\end{theorem}
	
	\noindent \textbf{Proof.} We show the energy dissipation law by analyzing the temporal evolution of the free energy $\mathcal{E}_f(\phi, \psi)$ and the kinetic energy $\mathcal{E}_K$ in the $L^2$ Hilbert space. We first consider the temporal evolution of the total free energy $\mathcal E_f(\phi,\psi)$, and take its time derivative,
	\begin{equation}
		\frac{d\mathcal{E}_{f}}{dt} = \int_{\Omega} \left( \frac{\delta \mathcal{E}_{f}}{\delta \phi}\frac{\partial \phi}{\partial t} + \frac{\delta \mathcal{E}_{f}}{\delta \psi}\frac{\partial \psi}{\partial t} \right) d\Omega + \int_{\partial\Omega} \left( \mathcal{B}_{\phi}\frac{\partial \phi}{\partial t} + \mathcal{B}_{\psi}\frac{\partial \psi}{\partial t} \right) dS.
	\end{equation}
	With the help of the natural boundary conditions $\mathcal{B}_{\phi} = 0$ and $\mathcal{B}_{\psi} = 0$, the boundary integral strictly vanishes. Based on the definitions of the chemical potentials $\mu_{\phi} = \delta\mathcal{E}_{f}/\delta\phi$ and $\mu_{\psi} = \delta\mathcal{E}_{f}/\delta\psi$, we can obtain
	\begin{equation}
		\frac{d\mathcal{E}_{f}}{dt} = (\mu_{\phi}, \partial_t\phi) + (\mu_{\psi}, \partial_t\psi).
		\label{Dfree0}
	\end{equation}
	Taking the $L^2$ inner product of the Allen-Cahn equation (\ref{AC_phi}) with $\mu_{\phi}$, we obtain
	\begin{equation}
		\begin{aligned}
			(\mu_\phi, \partial_t \phi)
			&= -(\mu_\phi, \nabla \cdot (\phi \mathbf{u})) - M_\phi (\mu_\phi, \mu_\phi - \tilde{\mu}_\phi) \\
			&= -\int_\Omega \mu_\phi \nabla \cdot (\phi \mathbf{u}) \, d\Omega - M_\phi \int_\Omega \mu_\phi (\mu_\phi - \tilde{\mu}_\phi) \, d\Omega \\
			&= -\int_{\partial \Omega} \mu_\phi \phi(\mathbf{u} \cdot \mathbf{n}) \, dS + \int_\Omega \phi \mathbf{u} \cdot \nabla \mu_\phi \, d\Omega - M_\phi \int_\Omega (\mu_\phi - \tilde{\mu}_\phi)^2 \, d\Omega \\
			&= \int_\Omega \phi \mathbf{u} \cdot \nabla \mu_\phi \, d\Omega - M_\phi \left\| \mu_\phi - \tilde{\mu}_\phi \right\|_{L^2(\Omega)}^2,
		\end{aligned}
		\label{Dfreea}
	\end{equation}
	where the divergence theorem and the no-slip boundary condition ($\mathbf{u} = \mathbf{0}$ on $\partial\Omega$) have been adopted to eliminate the boundary integral. Furthermore, according to the definition of $\tilde{\mu}_\phi = \lambda_\phi(t) h_\phi'(\phi)$ and the explicit form of the Lagrange multiplier $\lambda_\phi(t)$, the relaxation term satisfies the following orthogonality identity,
	\begin{equation}
		\begin{aligned}
			(\mu_\phi, \mu_\phi - \tilde{\mu}_\phi) 
			&= (\mu_\phi - \tilde{\mu}_\phi, \mu_\phi - \tilde{\mu}_\phi) + (\tilde{\mu}_\phi, \mu_\phi - \tilde{\mu}_\phi) \\
			&= \|\mu_\phi - \tilde{\mu}_\phi\|_{L^2(\Omega)}^2 + \int_\Omega \lambda_\phi(t) h_\phi'(\phi) (\mu_\phi - \tilde{\mu}_\phi) \, d\Omega \\
			&= \|\mu_\phi - \tilde{\mu}_\phi\|_{L^2(\Omega)}^2 + \lambda_\phi(t) \int_\Omega h_\phi'(\phi) \mu_\phi \, d\Omega - \lambda^2_\phi(t) \int_\Omega [h_\phi'(\phi)]^2 \, d\Omega \\
			&= \|\mu_\phi - \tilde{\mu}_\phi\|_{L^2(\Omega)}^2 + \lambda_\phi(t) \left( \lambda_\phi(t) \int_\Omega [h_\phi'(\phi)]^2 \, d\Omega \right) - \lambda^2_\phi(t) \int_\Omega [h_\phi'(\phi)]^2 \, d\Omega \\
			&= \|\mu_\phi - \tilde{\mu}_\phi\|_{L^2(\Omega)}^2.
		\end{aligned}
	\end{equation}
	Following the same procedure for the surfactant transport equation (\ref{AC_psi}) yields
	\begin{equation}
		(\mu_{\psi}, \partial_t\psi) = \int_{\Omega} \psi \mathbf{u} \cdot \nabla \mu_{\psi} d\Omega - M_{\psi} \|\mu_{\psi} - \bar{\mu}_{\psi}\|_{L^2(\Omega)}^2.
		\label{Dfreeb}
	\end{equation}
	Combining Eqs~(\ref{Dfree0}), (\ref{Dfreea})and (\ref{Dfreeb}), the time rate of total free energy can be rewritten as
	\begin{equation}
		\frac{d\mathcal{E}_{f}}{dt} = \int_{\Omega} \mathbf{u} \cdot \left( \phi \nabla \mu_{\phi} + \psi \nabla \mu_{\psi} \right) d\Omega - M_{\phi} \|\mu_{\phi} - \tilde{\mu}_{\phi}\|_{L^2(\Omega)}^2 - M_{\psi} \|\mu_{\psi} - \bar{\mu}_{\psi}\|_{L^2(\Omega)}^2.
		\label{Dfree}
	\end{equation}
	Taking the $L^2$ inner product of the fluid velocity $\mathbf{u}$ with the momentum equation \ref{eqns}(b), one can obtain
	\begin{equation}
		\frac{d\mathcal{E}_{K}}{dt} = \int_{\Omega} \mathbf{u} \cdot \left( \rho \frac{\partial \mathbf{u}}{\partial t} + \rho \mathbf{u} \cdot \nabla \mathbf{u} \right) d\Omega = \int_{\Omega} \mathbf{u} \cdot \left( -\nabla p + \nabla \cdot \left[ \mu (\nabla \mathbf{u} + \nabla \mathbf{u}^\mathrm{T}) \right] + \mathbf{F}_s \right) d\Omega.
	\end{equation}
	Using the integration by parts and the divergence-free condition $\nabla \cdot \mathbf{u} = 0$, we have
	\begin{equation}
		-\int_{\Omega} \mathbf{u} \cdot \nabla p d\Omega = -\int_{\partial\Omega} p(\mathbf{u} \cdot \mathbf{n}) dS + \int_{\Omega} p (\nabla \cdot \mathbf{u}) d\Omega = 0.
	\end{equation}
	For the viscous dissipation term, one can derive 
	\begin{equation}
		\begin{aligned} 	\int_{\Omega} \mathbf{u} \cdot \nabla \cdot \left[ \mu (\nabla \mathbf{u} + \nabla \mathbf{u}^\mathrm{T}) \right] \mathrm{d}\Omega 	&= \int_{\partial \Omega} \mathbf{u} \cdot \left[ \mu (\nabla \mathbf{u} + \nabla \mathbf{u}^\mathrm{T}) \mathbf{n} \right] \mathrm{d}S 	- \int_{\Omega} \nabla \mathbf{u} : \left[ \mu (\nabla \mathbf{u} + \nabla \mathbf{u}^\mathrm{T}) \right] \mathrm{d}\Omega \\ 	&= -\int_{\Omega} \nabla \mathbf{u} : \left[ \mu (\nabla \mathbf{u} + \nabla \mathbf{u}^\mathrm{T}) \right] \mathrm{d}\Omega \\ 	&= -\int_{\Omega} \frac{\mu}{2} \vert{}\nabla \mathbf{u} + \nabla \mathbf{u}^\mathrm{T}\vert{}^2 \, \mathrm{d}\Omega, \end{aligned}
	\end{equation}
	where the integration by parts and the no-slip condition ($\mathbf{u}\vert{}_{\partial\Omega} = \mathbf{0}$) have been used, and the last equality is derived from the pointwise tensor identity $\nabla \mathbf{u} : \left[ \mu (\nabla \mathbf{u} + \nabla \mathbf{u}^\mathrm{T}) \right] = \frac{\mu}{2} \vert{}\nabla \mathbf{u} + \nabla \mathbf{u}^\mathrm{T}\vert{}^2$, which stems from the symmetry of the strain-rate tensor $\mathbf{D} = \frac{1}{2}(\nabla \mathbf{u} + \nabla \mathbf{u}^\mathrm{T})$ and the orthogonality between symmetric and antisymmetric tensors. Thus, the time derivative of the kinetic energy can be reformulated by
	\begin{equation}
		\frac{d\mathcal{E}_{K}}{dt} = \int_{\Omega} \mathbf{u} \cdot \mathbf{F}_s d\Omega - \int_{\Omega} \frac{\mu}{2} |\nabla \mathbf{u} + \nabla \mathbf{u}^\mathrm{T}|^2 d\Omega.
		\label{Dkinetic}
	\end{equation}
	Summing the evolution equations for the free energy (\ref{Dfree}) and the kinetic energy (\ref{Dkinetic}), the total energy dissipation rate of the system can be given by
	\begin{equation}
		\frac{d\mathcal{E}_{total}}{dt} = \int_{\Omega} \mathbf{u} \cdot \left( \phi \nabla \mu_{\phi} + \psi \nabla \mu_{\psi} + \mathbf{F}_s \right) d\Omega - M_{\phi} \|\mu_{\phi} - \tilde{\mu}_{\phi}\|_{L^2(\Omega)}^2 - M_{\psi} \|\mu_{\psi} - \bar{\mu}_{\psi}\|_{L^2(\Omega)}^2 - \int_{\Omega} \frac{\mu}{2} |\nabla \mathbf{u} + \nabla \mathbf{u}^\mathrm{T}|^2 d\Omega.
	\end{equation}
	Substituting Eq. (\ref{eq_Fs0}) back into the total energy equation, we can obtain the energy dissipation law,
	\begin{equation}
		\frac{d\mathcal{E}_{total}}{dt} = -M_{\phi} \|\mu_{\phi} - \tilde{\mu}_{\phi}\|_{L^2(\Omega)}^2 - M_{\psi} \|\mu_{\psi} - \bar{\mu}_{\psi}\|_{L^2(\Omega)}^2 - \int_{\Omega} \frac{\mu}{2} |\nabla \mathbf{u} + \nabla \mathbf{u}^\mathrm{T}|^2 d\Omega \le 0.
		\label{law}
	\end{equation}
	This completes the proof. \quad $\square$

	\section{Numerical methods}
	\label{sec3}
	In this section, we will present a lattice Boltzmann framework for the coupled flow, phase-field, and surfactant transport equations. Actually, a detailed derivation and analysis of the lattice Boltzmann framework for the Navier–Stokes equations and the convection–diffusion equation can be found in Ref. \cite{ChaiPRE2020}. Following a similar approach \cite{ChaiPRE2020}, we develop a unified multiple-relaxation-time lattice Boltzmann framework, in which the evolution equation for the distribution function associated with the $k$-th physical field can be expressed as
	\begin{equation}
		f_{i,k}(\mathbf{x} + \mathbf{c}_i \Delta t, t + \Delta t) = f_{i,k}(\mathbf{x}, t) - (\mathbf{M}^{-1}\mathbf{S}^k\mathbf{M})_{ij} [f_{j,k}(\mathbf{x}, t) - f_{j,k}^{eq}(\mathbf{x}, t)] + \Delta t F_{i,k}(\mathbf{x}, t),
	\end{equation}
	where $f_{i,k}(\mathbf{x}, t)$ denotes the distribution function associated with the $k$-th physical field at position $\mathbf{x}$ and time $t$, with $k=1,2,$ and $3$ corresponding to the flow, phase, and surfactant concentration fields, respectively. Here, $\Delta t$ is the time step, $\mathbf{c}_i$ is the discrete velocity ($i = 0, 1, \dots, q - 1$ with $q$ being the number of discrete velocity directions), $\mathbf{M}$ is a $q \times q$ transformation matrix projecting the distribution functions onto moment space, $\mathbf{S}^k$ is the diagonal relaxation matrix associated with the \(k\)-th distribution function. To recover the macroscopic equations for different fields, the local equilibrium distribution functions for the respective fields are designed as
	\begin{subequations}
		\begin{equation}
			f_{i,1}^{eq}(\mathbf{x}, t) = 
			\begin{cases}
				\dfrac{p}{c_s^2}(\omega_i - 1) + \rho \omega_i \left[ \dfrac{\mathbf{c}_i \cdot \mathbf{u}}{c_s^2} + \dfrac{(\mathbf{c}_i \cdot \mathbf{u})^2}{2c_s^4} - \dfrac{\mathbf{u} \cdot \mathbf{u}}{2c_s^2} \right], & i = 0, \\
				\dfrac{p}{c_s^2}\omega_i + \rho \omega_i \left[ \dfrac{\mathbf{c}_i \cdot \mathbf{u}}{c_s^2} + \dfrac{(\mathbf{c}_i \cdot \mathbf{u})^2}{2c_s^4} - \dfrac{\mathbf{u} \cdot \mathbf{u}}{2c_s^2} \right], & i \neq 0,
			\end{cases}
		\end{equation}
		\begin{equation}
			\begin{aligned}
				f_{i,2}^{eq}(\mathbf{x}, t) = 
				\begin{cases} 
					\phi + (\omega_i - 1)\eta_\phi \phi, & i = 0, \\[6pt]
					\omega_i \eta_\phi \phi + \dfrac{\omega_i \mathbf{c}_i \cdot (\phi \mathbf{u})}{c_s^2} , & i \neq 0,
				\end{cases} \quad \quad
				f_{i,3}^{eq}(\mathbf{x}, t) = 
				\begin{cases} 
					\psi + (\omega_i - 1)\eta_\psi \psi, & i = 0, \\[6pt]
					\omega_i \eta_\psi \psi + \dfrac{\omega_i \mathbf{c}_i \cdot (\psi \mathbf{u})}{c_s^2}, & i \neq 0,
				\end{cases} 
			\end{aligned}
		\end{equation}
	\end{subequations}
	where $\omega_i$ is the lattice weight, $c_s = c /\sqrt 3$ is the lattice sound speed, $c=\Delta x/\Delta t$ is the lattice speed, $\eta_\phi$ and $\eta_\psi$ are adjustable parameters controlling the diffusion coefficients. In addition, the discrete source terms $F_{i,k}$ are defined as 
	\begin{subequations}
		\begin{equation}
			F_{i,1}(\mathbf{x}, t) = A_{ij}^1 \omega_j \left[ \mathbf{u} \cdot \nabla \rho + \frac{\mathbf{c}_j \cdot \mathbf{F}}{c_s^2} + \frac{\mathbf{u} \nabla \rho : (\mathbf{c}_j \mathbf{c}_j - c_s^2 \mathbf{I})}{c_s^2} \right],
		\end{equation}
		\begin{equation}
			F_{i,2}(\mathbf{x}, t) = A_{ij}^2 \frac{\omega_j \mathbf{c}_j \cdot \partial_t (\phi \mathbf{u})}{c_s^2} + \omega_i \left( R_\phi + \frac{\Delta t}{2} \partial_t R_\phi \right), \quad R_\phi = M_\phi \left[ -4\beta(1-\psi)(\phi^3 - \phi) - W\psi\phi + \tilde{\mu}_\phi \right],
		\end{equation}
		\begin{equation}
			F_{i,3}(\mathbf{x}, t) = A_{ij}^3 \omega_j \frac{\mathbf{c}_j \cdot \partial_t (\psi \mathbf{u})}{c_s^2} + \omega_i \left( R_\psi + \frac{\Delta t}{2} \partial_t R_\psi \right), 
		\end{equation}
		\begin{equation}
			R_\psi = M_\psi \left[ -\frac{M_0}{M_\psi} \nabla^2 \psi -k_B T \ln\left( \frac{\psi}{1-\psi} \right) + \beta(1-\phi^2)^2 - \frac{W}{2}\phi^2 + \bar{\mu}_\psi \right],
		\end{equation}
	\end{subequations}
	where $ \mathbf{F}=\mathbf{F}_{\mathrm{s}}+\mathbf{G}$ is the total force, $A_{ij}^k=[\mathbf{M}^{-1}(\mathbf{I} - \mathbf{S}^k/2 )\mathbf{M}]_{ij}$, and $\mathbf{I}$ is the identity matrix. Here, the relaxation matrices are specified as $\mathbf{S}^k = \mathbf{diag}(s_0^k, s_1^k, s_1^k, s_2^k, s_2^k, s_2^k, s_3^k, s_3^k, s_4^k)$, where \(s_\alpha^k\) denotes the relaxation rate associated with the corresponding moment of distribution function $f_{i,k}(\mathbf{x}, t)$. The relaxation rates governing the viscous and diffusive modes are related to the transport coefficients through $\nu = c_s^2 \Delta t (1/s_2^1 - 0.5)$, $M_\phi \kappa  = \eta_{\phi} c_s^2 \Delta t (1/s_1^2 - 0.5)$, and $M_\psi \lambda_s + M_0 = \eta_{\psi} c_s^2 \Delta t (1/s_1^3 - 0.5)$. Furthermore, the macroscopic velocity $\mathbf{u}$, pressure $p$, order parameter $\phi$ and concentration $\psi$ are calculated by
	\begin{equation}
		\rho \mathbf{u} = \sum_i \mathbf{c}_i f_{i,1} + \frac{\Delta t}{2} \mathbf{F}, \quad
		p = \frac{c_s^2}{1 - \omega_0} \left( \sum_{i \neq 0} f_{i,1} + \frac{\Delta t}{2} \mathbf{u} \cdot \nabla \rho - \rho \omega_0 \frac{\mathbf{u} \cdot \mathbf{u}}{2c_s^2} \right), \quad 
		\phi = \sum_{i} f_{i,2}(\mathbf{x}, t), \quad 
		\psi = \sum_{i} f_{i,3}(\mathbf{x}, t).
	\end{equation}
	
	Finally, the  standard lattice-based numerical scheme \cite{ZhangICHMT2025} is applied to discretize the wetting and solid-surface adsorption boundary conditions given in Eq. (\ref{boundary}), the second-order isotropic lattice stencils and first-order backward-difference formula are used to compute gradient, Laplacian and time-derivative terms,  
	\begin{equation}
		\begin{aligned}
			\nabla \zeta(\mathbf{x}, t)=\sum_{i \neq 0} \frac{\omega_i \mathbf{c}_i \zeta\left(\mathbf{x}+\mathbf{c}_i \Delta t, t\right)}{c_s^2 \Delta t},\: \nabla^2 \zeta(\mathbf{x}, t)=\sum_{i \neq 0} \frac{2 \omega_i\left[\zeta\left(\mathbf{x}+\mathbf{c}_i \Delta t, t\right)-\zeta(\mathbf{x}, t)\right]}{c_s^2 \Delta t^2},\: 
			\frac{\partial  \zeta(\mathbf{x}, t) }{\partial t}= \frac{\zeta(\mathbf{x}, t) - \zeta(\mathbf{x}, t - \Delta t)}{\Delta t}.
		\end{aligned}
	\end{equation}

	\section{Numerical results and discussion}
	\label{sec4}
	
	In this section, several numerical examples are performed to test the proposed phase-field model and investigate surfactant-laden contact line dynamics with solid adsorption. To facilitate parametric analysis, the characteristic scales, including the initial droplet radius $R_0$, reference velocity $U_0$, reference density $\rho_0$, dynamic viscosity $\mu_0$, and clean fluid--fluid interfacial tension $\sigma_{gl,0}$ are introduced, and some key dimensionless parameters governing the simulations can be given by
	\begin{equation}
		\mathrm{Re}=\frac{\rho_0U_0R_0}{\mu_0},
		\qquad
		\mathrm{Ca}=\frac{\mu_0U_0}{\sigma_{gl,0}},
		\qquad
		\mathrm{Cn}=\frac{\varepsilon}{R_0},
		\qquad
		\mathrm{Ex}=\frac{4\beta}{W},
		\qquad
		\mathrm{Pi}=\frac{k_BT}{4\beta},
	\end{equation}
	where $\mathrm{Re}$ and $\mathrm{Ca}$ denote the Reynolds and capillary numbers, which are defined as the ratios of inertial force to viscous forces and viscous to capillary force, respectively. The Cahn number $\mathrm{Cn}$ characterizes the relative diffuse-interface thickness. The solubility parameter $\mathrm{Ex}$ represents the energetic penalty that suppresses surfactant dissolution into the bulk phases, and $\mathrm{Pi}$ quantifies the thermal mixing entropy of the surfactant relative to the phase-field interfacial energy. For the present free-energy parametrization, $\mathrm{Pi}$ is related to the Langmuir adsorption constant $\psi_c$,
	\begin{equation}
		\mathrm{Pi}=-\frac{1+2/\mathrm{Ex}}{4\ln\psi_c}.
	\end{equation}

	\subsection{Accuracy and mass conservation test}
	To assess the accuracy of the coupled phase-field lattice Boltzmann formulation, we consider a simple problem in a doubly periodic domain $\Omega=[0,1]^2$, where a constant divergence-free velocity $\mathbf{u}=(0.05,-0.03)$ is prescribed, the analytical solutions of phase-field variable and the surfactant concentration are given by $\phi^e(x,y,t)=0.5e^{-t}\sin[2\pi(x-0.05t)]\cos[2\pi(y+0.03t)]$ and $\psi^e(x,y,t)=0.5+0.05e^{-2t}\cos[2\pi(x-0.05t)]\sin[2\pi(y+0.03t)]$, respectively. Substituting these manufactured solutions into the conservative Allen-Cahn equations for phase and concentration fields, one can determine  $S_\phi=-\phi^e+M(\mu_\phi^e-\bar{\mu}_\phi^e)$ and $S_\psi=-2(\psi^e-0.5)+M(\mu_\psi^e-\bar{\mu}_\psi^e)$. Here, the global Lagrange multipliers $\bar{\mu}_\phi^e$ and $\bar{\mu}_\psi^e$, representing the spatial averages of the analytical chemical potentials evaluated from the manufactured solutions, are given by $\bar{\mu}_\phi^e = 0$, $ \bar{\mu}_\psi^e = -\beta\left(1- a_\phi^2/2 + 9a_\phi^4/64 \right)+W_{\mathrm{ads}}a_\phi^2/8$,
	with $a_\phi = 0.5e^{-t}$. The model parameters are set to $M=0.01$, $\beta=0.03$, $\kappa=9.375\times10^{-6}$, $k_BT=0.162$, $W=0.24$, and $\lambda_s=0$. In our simulations, the grids with resolutions ranging from $32^2$ to $512^2$ are used, and the time step is determined by $\Delta t=\Delta x^2/0.15$. To evaluate the numerical accuracy, we consider the relative $L_2$ errors of both phase and concentration fields,
	\begin{equation}
		E_2^{\phi} = \sqrt{\frac{\sum_{i,j}[\phi_{i,j}(t) - \phi^e(x_i, y_j, t)]^2}{\sum_{i,j}[\phi^e(x_i, y_j, t)]^2}}, \quad E_2^{\psi} = \sqrt{\frac{\sum_{i,j}[\psi_{i,j}(t) - \psi^e(x_i, y_j, t)]^2}{\sum_{i,j}[\psi^e(x_i, y_j, t)]^2}}.
	\end{equation}
	
	We first conduct a numerical test on the convergence rate of the lattice Boltzmann method, and present the results in Table \ref{table1}. From this table, one can observe that under diffusive scaling, the numerical method has a second-order convergence rate in space. Meanwhile, it is also found that the relative deviations of the globally conserved quantities associated with $\phi$ and $\psi$ remain below $5.0\times10^{-17}$ and $1.9\times10^{-14}$, respectively, confirming global conservation to within round-off error.
	\begin{table}[H]
		\centering
		\caption{Relative $L_2$ errors and convergence orders of lattice Boltzmann method for the phase field variable $\phi$ and surfactant concentration $\psi$ at $t=5$.}
		\begin{tabular}{ccccc}
			\toprule
			$\Delta x$ & $E_2^{\phi}$             & Order   & $E_2^{\psi}$             & Order   \\ \midrule
			1/32       & $1.49467 \times 10^{-3}$ & -       & $1.22546 \times 10^{-4}$ & -       \\
			1/64       & $3.75136 \times 10^{-4}$ & 1.99434 & $3.06679 \times 10^{-5}$ & 1.99853 \\
			1/128      & $9.38754 \times 10^{-5}$ & 1.99860 & $7.66864 \times 10^{-6}$ & 1.99969 \\
			1/256      & $2.34745 \times 10^{-5}$ & 1.99965 & $1.91726 \times 10^{-6}$ & 1.99992 \\
			1/512      & $5.86899 \times 10^{-6}$ & 1.99991 & $4.79321 \times 10^{-7}$ & 1.99999 \\ \bottomrule
		\end{tabular}
		\label{table1}
	\end{table}

	Furthermore, to numerically evaluate the theoretically predicted second-order accuracy of zero-level-set geometric-volume preservation, we simulate the capillary-driven relaxation of an initially stationary elliptical droplet in a doubly periodic square domain $\Omega=[0,1]\times[0,1]$. The initial zero-level set is centered at $(0.5, 0.5)$ with semi-axes $a = 0.28$ and $b = 0.18$, corresponding to an analytical sharp-interface area of $\mathcal{A}_\Gamma(0; \varepsilon)= \pi a b$. The phase field is initialized by the equilibrium profile $\phi(\boldsymbol{x}, 0) = \phi_b \tanh[2\phi_b d(\boldsymbol{x})/\varepsilon]$, where $d(\boldsymbol{x})$ denotes the signed distance to the elliptical interface, $\varepsilon$ is the parameter related to interfacial thickness, and $\phi_b = 1$. To isolate the phase-field conservation properties, surfactant transport is disabled and matched fluid properties are assigned: $\rho_l = \rho_g = 1$, $\mu_l = \mu_g = 0.01$, surface tension $\sigma = 0.01$, and mobility $M_\phi = 0.01$. To minimize spatial-discretization effects, a large grid size $1024 \times 1024$ is employed, so that the error related to the parameter $\varepsilon$ dominates within the fitted asymptotic interval. All simulations are terminated at $t = 100$, this is because at this moment, the maximum deviation of the droplet axis ratio from unity falls below $1.5 \times 10^{-4}$ and the maximum velocity magnitude remains below $2.3 \times 10^{-5}$, indicating that the droplet has reached a quiescent, nearly circular equilibrium state. The relative geometric-area drift is evaluated directly from the reconstructed zero-level-set area $\mathcal A_\Gamma$ as $E(\varepsilon) = \vert{}\mathcal{A}_\Gamma(t; \varepsilon) - \mathcal{A}_\Gamma(0; \varepsilon)\vert{} / \mathcal{A}_\Gamma(0; \varepsilon)$. Here the area $\mathcal{A}_\Gamma(t; \varepsilon)$ enclosed by the zero-level set is measured by dividing each Cartesian cell into two triangles, linearly reconstructing the \(\phi=0\) contour within each triangle, and summing the area fraction $\phi>0$ over the domain.

	We plot the relative geometric-volume drift $E(\varepsilon)$ at $t=100$ as a function of the interfacial thickness parameter $\varepsilon$ in Fig.~\ref{fig1}. From this figure, one can observe that the linear-constraint formulation produces larger geometric-volume errors and does not exhibit quadratic convergence over the tested range. In contrast, the proposed nonlinear formulation follows the $\mathcal{O}(\varepsilon^2)$ scaling within the asymptotic interval $\varepsilon \in [0.0175, 0.0300]$, yielding a second-order convergence rate. Furthermore, the nonlinear formulation yields errors that are one order of magnitude smaller than those of the linear formulation at $\varepsilon = 0.01$. These results provide numerical evidence of second-order geometric-area preservation over the stated asymptotic interval and are consistent with the theoretical prediction in Eq.~(\ref{second}).

	\begin{figure}[H]
		\centering
		\includegraphics[scale=0.5]{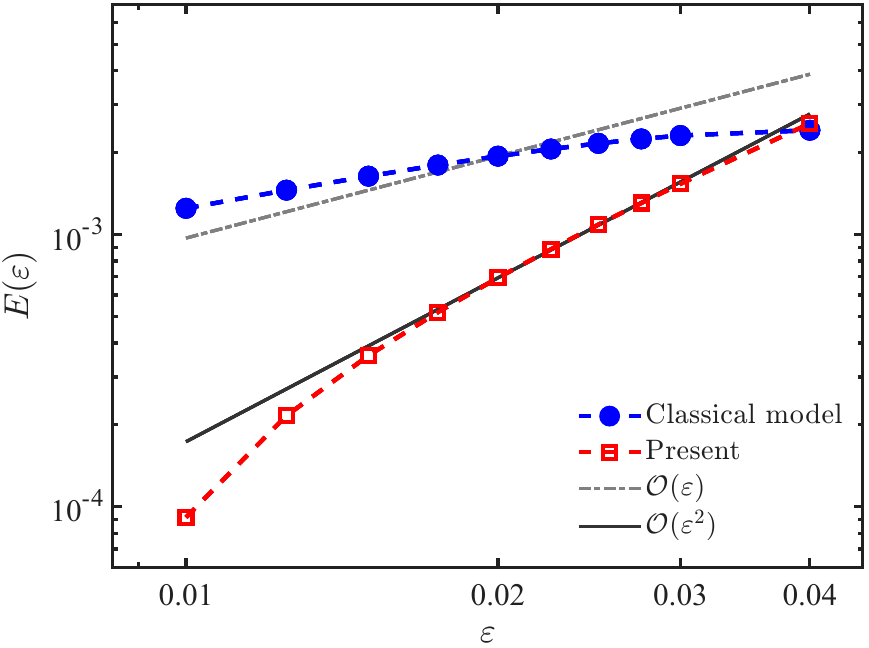} 
		\caption{ The relative geometric-area drift $E(\varepsilon)$ with respect to the interface-thickness parameter $\varepsilon$. }
		\label{fig1}
	\end{figure}

	\subsection{Static surfactant-covered droplet}
	
	To show the capability of the proposed method in maintaining coupled thermodynamic equilibrium, preserving the prescribed global mass/volume constraints, and suppressing spurious velocities, we perform numerical simulations of a stationary surfactant-laden droplet in the physical domain $\Omega=[0,1]\times[0,1]$ with periodic boundary conditions imposed in both directions. Initially, a circular droplet with the radius $R=0.25$ is placed at the center of the domain, and the physical and numerical parameters are specified as: $\rho_l/\rho_g=1000$, $\mu_l/\mu_g=100$, $\rho_g=1$, $\mu_g=0.1$, $\sigma_{gl,0}=0.001$, $\varepsilon=0.025$, $M_\phi=M_\psi=0.01$, $\lambda_s=0$, $k_BT=0.162$, and $W=0.24$.

	\begin{figure}[H]
		\centering
		\includegraphics[scale=0.5]{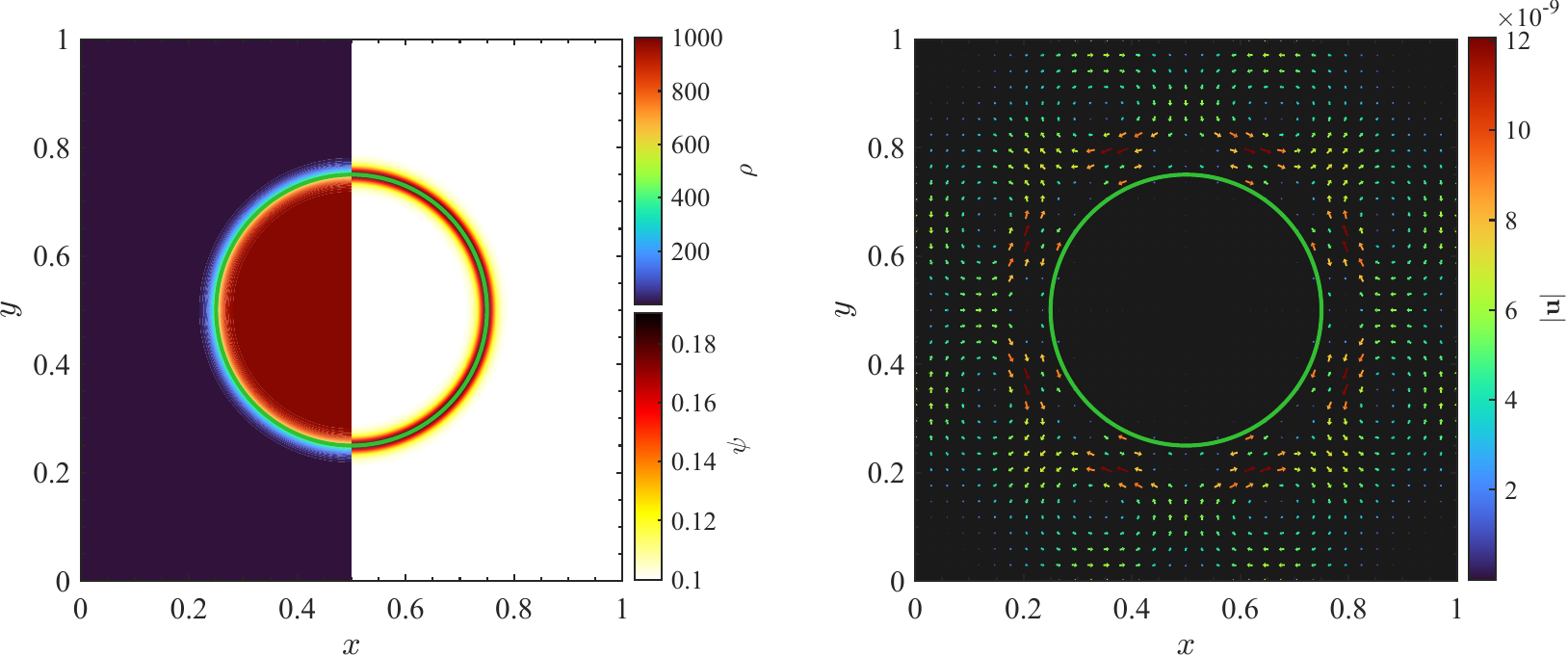} 
		\put(-395,150){(\textit{a})}
		\put(-190,150){(\textit{b})}
		\caption{The equilibrium distributions of the density, surfactant concentration, and corresponding velocity field at $\psi_b=0.1$ and $\Delta x=\Delta y=1/200$. (a) Distributions of the density $\rho$ (left half) and  concentration $\psi$ (right half). (b) Velocity field, with vectors colored by the velocity magnitude \(|\mathbf{u}|\) and the green line indicating the \(\phi=0\) contour. }
		\label{fig2}
	\end{figure}

	We present the equilibrium distributions of the order parameter and the surfactant concentration, together with the corresponding velocity field in Fig. \ref{fig2} where $\psi_b=0.1$ and $\Delta x=\Delta y=1/200$. As seen from this figure, both the density and surfactant concentration vary smoothly across the diffuse interface, with pronounced surfactant enrichment localized within the interfacial region. The velocity field exhibits weak, approximately symmetric recirculating vortices near the interface that decay rapidly into the bulk phases. Such a velocity pattern is characteristic of spurious currents in the diffuse-interface simulations of static equilibria \cite{YueJCP2010}. It should be noted, however, that the maximum velocity magnitude is of the order $10^{-8}$, indicating that the balance of discrete force in the present formulation can be well preserved.

	To quantitatively evaluate the equilibrium solutions, Figs. \ref{fig3}(a) and \ref{fig3}(b) show the numerical and analytical profiles of the order parameter $\phi$ and the surfactant concentration $\psi$ along the droplet centerline at $t=40$. It is found from this figure that for all considered values of bulk surfactant concentration $\psi_b$, there is a good agreement between them. The surfactant concentration reaches a maximum within the diffuse-interface region and approaches the specified bulk value $\psi_b$ away from the interface, which is consistent with preferential surfactant adsorption at the fluid--fluid interface. As $\psi_b$ increases, the bulk value of the order parameter $\phi_b$ decreases slightly, while the interfacial surfactant concentration increases markedly. Since the values of $\psi_b$ considered here are relatively close, the profiles of $\phi$ differ only slightly and therefore nearly overlap, as shown in Fig. \ref{fig3}(a). For the present nonlinear conservative formulation, the globally conserved quantities $m_\phi$ and $m_\psi$ defined in Eq. (\ref{conservation}) are considered, and the temporal evolutions of them at different values of $\psi_b$ are plotted in Figs. \ref{fig3}(c) and \ref{fig3}(d), respectively. 
	The results illustrate that both quantities remain constants throughout the simulations for all considered bulk surfactant concentrations, confirming that the prescribed global constraints are preserved.
	
	\begin{figure}[H]
		\centering
		\includegraphics[scale=0.45]{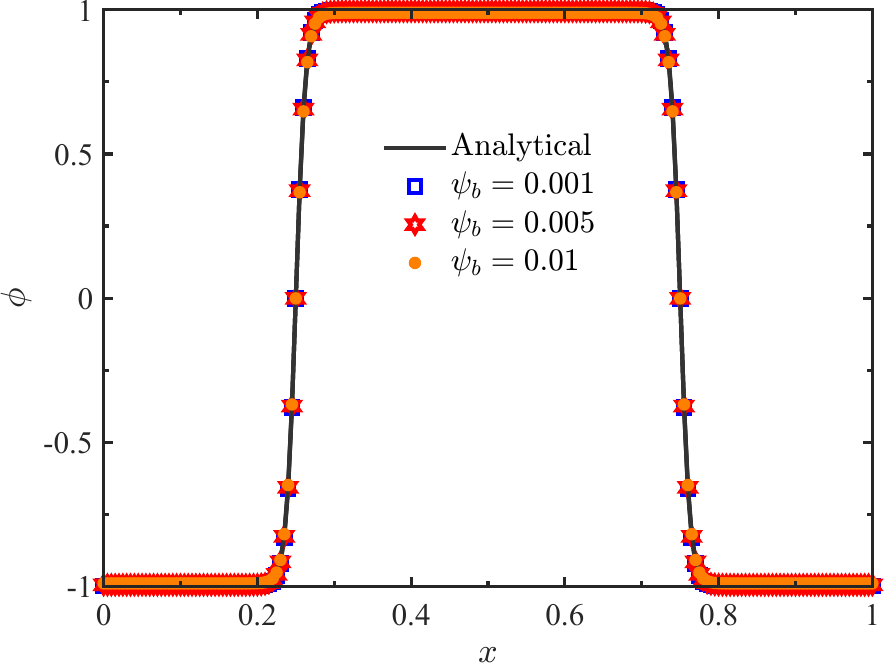} 
		\put(-200,140){(\textit{a})}
		\hspace{5mm}  
		\includegraphics[scale=0.45]{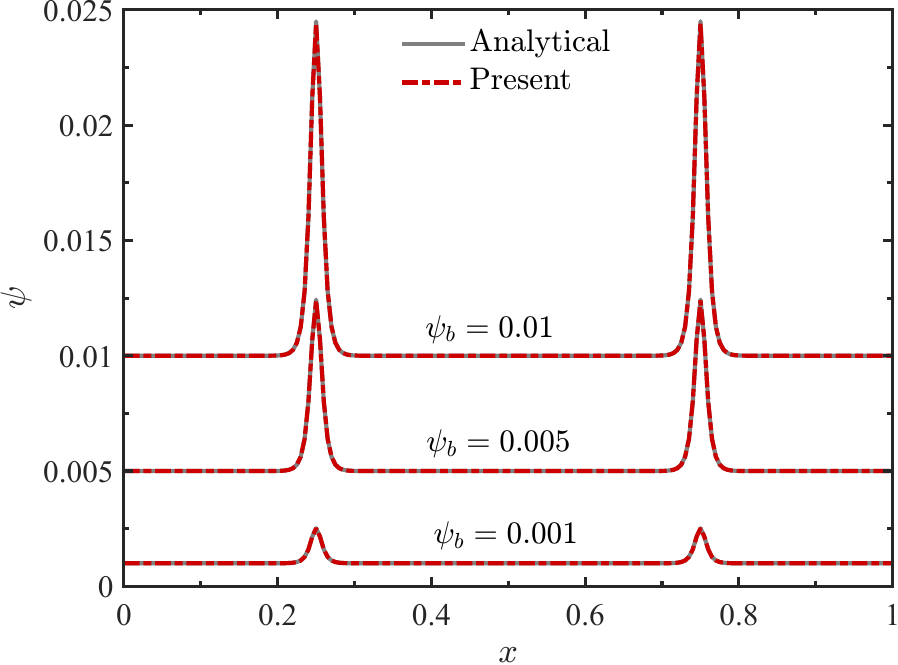} 
		\put(-200,140){(\textit{b})}
		\vspace{1mm}  
		\includegraphics[scale=0.45]{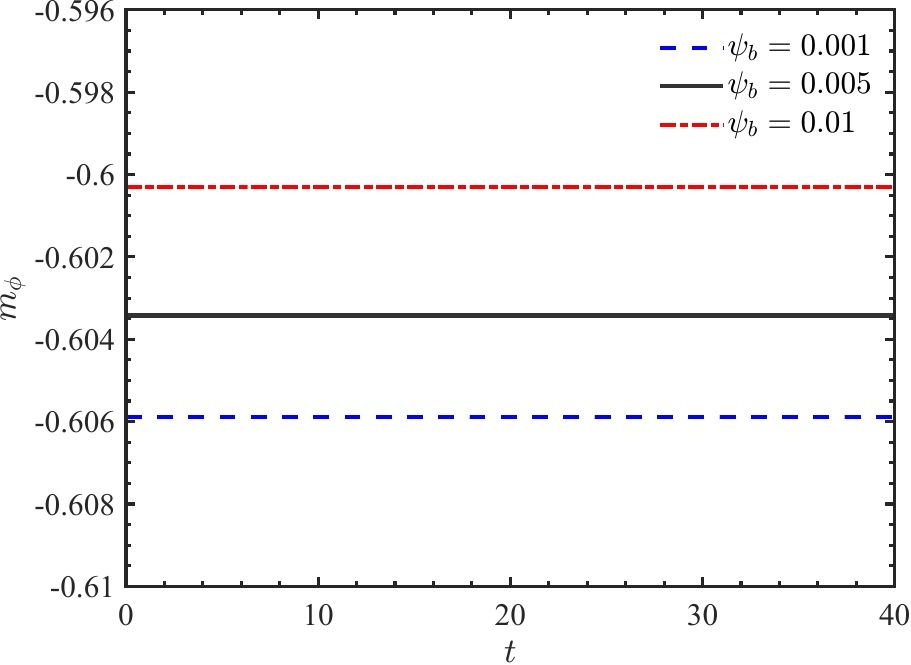} 
		\put(-203,140){(\textit{c})}
		\hspace{5mm}
		\includegraphics[scale=0.45]{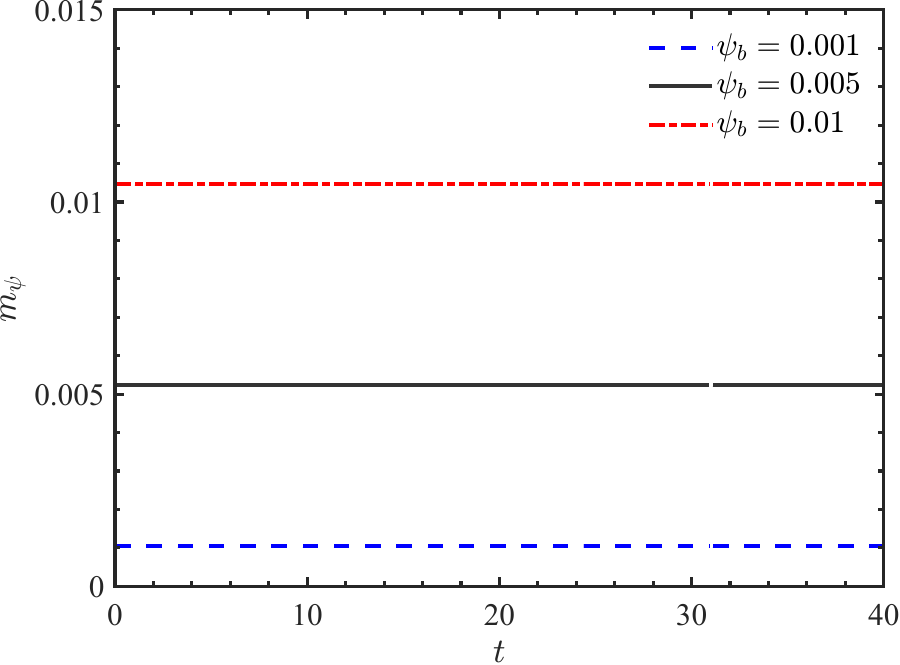} 
		\put(-203,140){(\textit{d})}
		\caption{ The equilibrium solutions and conservation properties at different values of bulk surfactant concentration $\psi_b$. (a) Numerical and analytical profiles of order parameter $\phi$. (b) Numerical and analytical profiles of surfactant concentration $\psi$. (c) Temporal evolution of the nonlinear phase measure $m_\phi$. (d) Temporal evolution of the total surfactant mass $m_\psi$. }
		\label{fig3}
	\end{figure}

	Spurious velocities, also referred to as parasitic currents, are a well-known numerical artifact in diffuse-interface simulations of multiphase flows. Table \ref{tab2} summarizes the maximum values of spurious velocity magnitude $\vert{}\mathbf{u}\vert{}_{\max}$ at different grid resolutions and bulk surfactant concentrations. As the grid is refined from $100\times100$ to $800\times800$, $\vert{}\mathbf{u}\vert{}_{\max}$ decreases from $\mathcal{O}(10^{-6})$ to below $\mathcal{O}(10^{-11})$, while showing only a weak dependence on $\psi_b$. These results demonstrate that the proposed method can accurately reproduce the coupled thermodynamic equilibrium in the presence of interfacial surfactant adsorption, preserve the imposed global constraints, and effectively suppress the spurious currents.

	\begin{table}[H]
		\centering
		\caption{Maximum value of spurious velocity $\vert{}\mathbf{u}\vert{}_{\max}$ at different bulk surfactant concentrations $\psi_b$ and grid resolutions at $t=40$.}
		\begin{tabular}{ccccc}
			\toprule
			$\psi_b$
			& $100 \times 100$
			& $200 \times 200$
			& $400 \times 400$
			& $800 \times 800$ \\ \midrule
			0.001
			& $2.73065 \times 10^{-6}$
			& $2.01373 \times 10^{-8}$
			& $5.16011 \times 10^{-10}$
			& $2.55204 \times 10^{-11}$ \\
			0.005
			& $2.49576 \times 10^{-6}$
			& $1.92287 \times 10^{-8}$
			& $5.37177 \times 10^{-10}$
			& $2.91039 \times 10^{-11}$ \\
			0.01
			& $2.25421 \times 10^{-6}$
			& $1.81499 \times 10^{-8}$
			& $5.76585 \times 10^{-10}$
			& $3.39571 \times 10^{-11}$ \\
			0.1
			& $3.19836 \times 10^{-7}$
			& $1.26912 \times 10^{-8}$
			& $7.62299 \times 10^{-10}$
			& $4.98816 \times 10^{-11}$ \\ \bottomrule
		\end{tabular}
		\label{tab2}
	\end{table}

	\subsection{Droplet spreading on a wetting solid}
	To evaluate the effect of the constraint function $h_\phi(\phi)$ used in the volume constraint on the equilibrium droplet geometry, the relaxation of a two-dimensional surfactant-free droplet on a homogeneous solid wall is investigated. The computational domain is $\Omega=[0,2]\times[0,1]$, and is discretized with a uniform $200\times100$ grid. Initially, a semicircular droplet of radius $R_0=0.35$ is placed at $(x_0,y_0)=(1,0)$, and the order-parameter field is given by
	\begin{equation}
		\phi(x,y,0) = \phi_b\tanh \left[ \frac{2\phi_b(R_0-\sqrt{(x-x_0)^2+(y-y_0)^2})}{\varepsilon} \right],
	\end{equation}
	where $\phi_b=1$. The density and viscosity ratios are fixed as $\rho_l/\rho_g=1000$ and $\mu_l/\mu_g=100$, and other simulation parameters are set as $\Delta x=0.01$, $\sigma_{gl,0}=0.1$, $M_\phi=0.1$, $\rho_g=1$, $\nu_g=0.1$, and $\varepsilon=6\Delta x$. The periodic boundary conditions are imposed in the horizontal direction, while the no-slip boundary conditions are applied at the top and bottom boundaries. The wetting boundary condition given in Eq. (\ref{boundary}) is imposed at the bottom wall. According to the previous work \cite{ZhuJCP2020}, the analytical equilibrium shape is a two-dimensional circular cap uniquely determined by the initial geometric area $\mathcal{A}_0=\pi R_0^2/2$ and the specified contact angle $\theta_0$. The corresponding analytical expressions of the droplet radius $R$, spreading length $L$, height $H$, and interface contour are given by
	\begin{equation}
		R=R_0\sqrt{\frac{\pi/2}{\theta_0-\sin\theta_0\cos\theta_0}},\ L=2R\sin\theta_0, \ H=R(1-\cos\theta_0), \ (x-x_0)^2+(y+R\cos\theta_0)^2=R^2, y\ge0.
		\label{a}
	\end{equation}
	Figs. \ref{fig4} and \ref{fig5} present the analytical and numerical results of steady-state spreading length $L$, droplet height $H$ and droplet profile. From these figures, one can find that although the classical linear model and the present nonlinear model can accurately predict the spreading length, the former brings larger errors in predicting the droplet height, and the shrinkage of the droplet profile becomes more pronounced with the increase of contact angle $\theta_0$. On the contrary, the present nonlinear model with improved volume conservation maintains high accuracy across all contact angles considered, indicating that the nonlinear volume constraint substantially reduces curvature-induced geometric volume loss.

	\begin{figure}[H]
		\centering
		\includegraphics[scale=0.45]{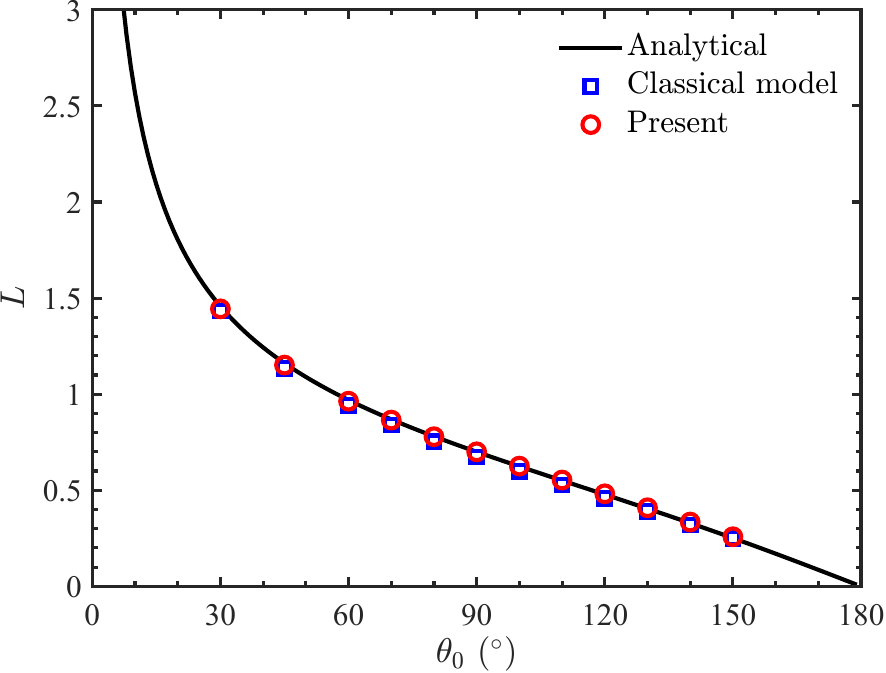} 
		\put(-195,135){(\textit{a})}
		\hspace{5mm}  
		\includegraphics[scale=0.45]{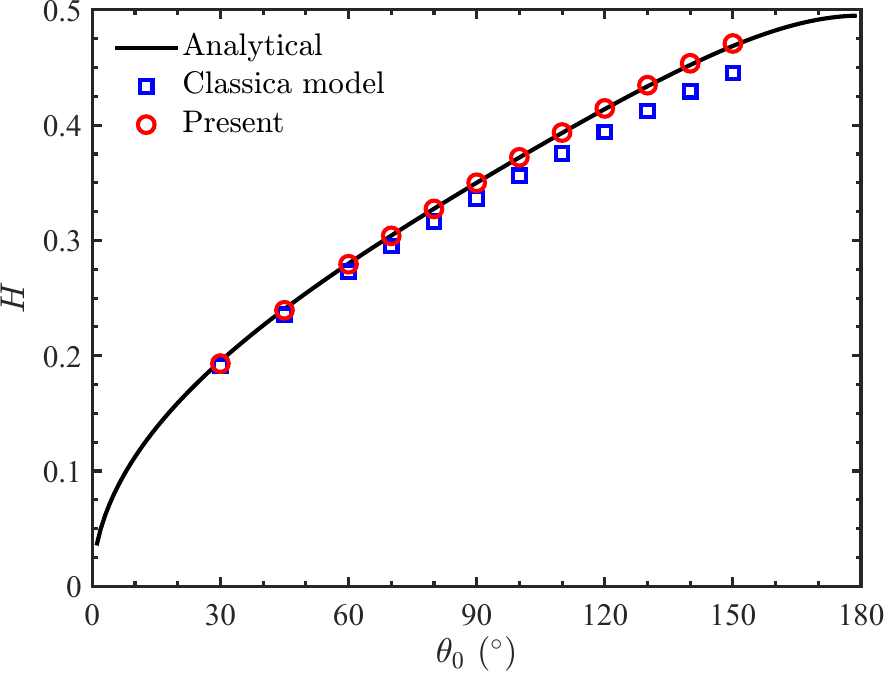} 
		\put(-195,135){(\textit{b})}
		\caption{ Comparisons of the analytical predictions with the numerical results based on the present nonlinear model and the classical linear-constraint model: (a) spreading length $L$ and (b) droplet height $H$. }
		\label{fig4}
	\end{figure}
	
	To quantitatively explain the observed droplet shrinkage, we derive leading-order asymptotic estimates for the bulk-phase shift $\delta_\phi$ and geometric-area loss $\Delta \mathcal{A}/\mathcal{A}_0$ for the classical linear model. At the static equilibrium state, the chemical potential in the linear model is spatially uniform, i.e., $\mu_\phi = \bar{\mu}_\phi$. In the homogeneous bulk regions away from the diffuse interface, the gradient terms would vanish, and the chemical potential derived from the double-well potential $f(\phi) = \beta(\phi^2-1)^2$ can be given by $\mu_\phi = f'(\phi) = 4\beta\phi(\phi^2-1)$. Assuming a small uniform shift of the bulk phase values $\phi_\pm = \pm 1 + \delta_\phi$ ($\vert{}\delta_\phi\vert{} \ll 1$), a Taylor expansion about $\phi=\pm 1$ yields
	\begin{equation}
		\mu_\phi(\pm 1 + \delta_\phi) = 8\beta\delta_\phi+\mathcal{O}(\delta_\phi^2).
		\label{a1}
	\end{equation}
	For a two-dimensional droplet with the radius $R$, the equilibrium chemical potential satisfies the leading-order Gibbs--Thomson relation,
	\begin{equation}
		\bar{\mu}_\phi = \frac{\sigma}{2 R}+\mathcal{O}\left[\left(\frac{\varepsilon}{R}\right)^2\right]. 
		\label{a2}
	\end{equation}
	Combining Eqs. (\ref{a}), (\ref{a1}) and \ref{a2} gives rise to the bulk phase shift $\delta_\phi$,
	\begin{equation}
		\delta_\phi = \frac{\varepsilon}{12R} + \mathcal{O}\!\left[\left(\frac{\varepsilon}{R}\right)^2\right] = \frac{\varepsilon}{12}\sqrt{\frac{\theta_0-\sin\theta_0\cos\theta_0}{\mathcal{A}_0}} + \mathcal{O}\!\left[\left(\frac{\varepsilon}{R}\right)^2\right].
		\label{ca1}
	\end{equation}
	Assuming that this uniform bulk shift $\delta_\phi$ develops at equilibrium  state, the phase-field variable in the liquid and gas bulk regions can be approximated as $\phi = +1 + \delta_\phi$ in the liquid region of area $\mathcal{A}$, and $\phi = -1 + \delta_\phi$ in the gas region of area $\vert{}\Omega\vert{} - \mathcal{A}$. Neglecting the higher-order contribution from the thin diffuse-interface region, the global integral of $\phi$ can be approximated by
	\begin{equation}
		\int_\Omega \phi \, \mathrm{d}\Omega = (+1 + \delta_\phi)\mathcal{A} + (-1 + 	\delta_\phi)(\vert{}\Omega\vert{} - \mathcal{A}) = 2\mathcal{A} - \vert{}\Omega\vert{} + \delta_\phi \vert{}\Omega\vert{}.
	\end{equation}	
	In the initial state, the bulk phases take the ideal equilibrium values ($\delta_\phi = 0$), and the same approximation gives $\int_\Omega \phi \, \mathrm{d}\Omega = 2\mathcal{A}_0 - \vert{}\Omega\vert{}$, with $\mathcal{A}_0$ being the initial geometric area. Based on the conservation of $\int_\Omega \phi \, \mathrm{d}\Omega$ and Eq. (\ref{ca1}), the relative geometric area loss can be obtained,
	\begin{equation}
		\frac{\Delta \mathcal{A}}{\mathcal{A}_0} = \frac{\vert{}\Omega\vert{}}{2\mathcal{A}_0}\delta_\phi = \frac{\vert{}\Omega\vert{}}{\mathcal{A}_0} \frac{\varepsilon}{24R},
		\label{ca2}
	\end{equation}
	where $\Delta \mathcal{A} = \mathcal{A}_0 - \mathcal{A}$ is the geometric-area loss. Eq. (\ref{ca2}) indicates that the classical linear constraint produces a leading-order geometric-area loss through the curvature-induced shift of the bulk phase values.

	\begin{figure}[htbp]
		\centering
		\includegraphics[scale=0.45]{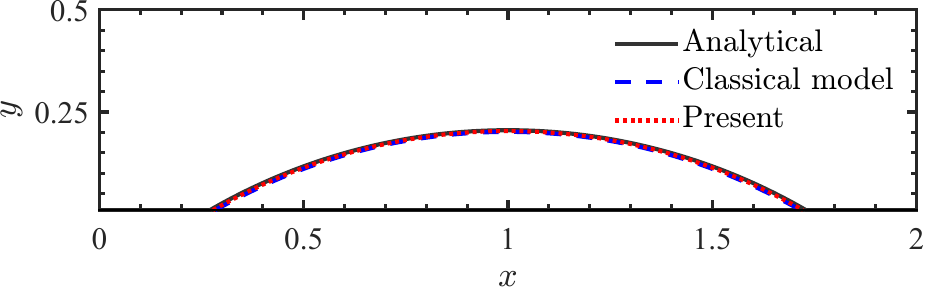} 
		\put(-205,55){(\textit{a})}
		\hspace{5mm}  
		\includegraphics[scale=0.45]{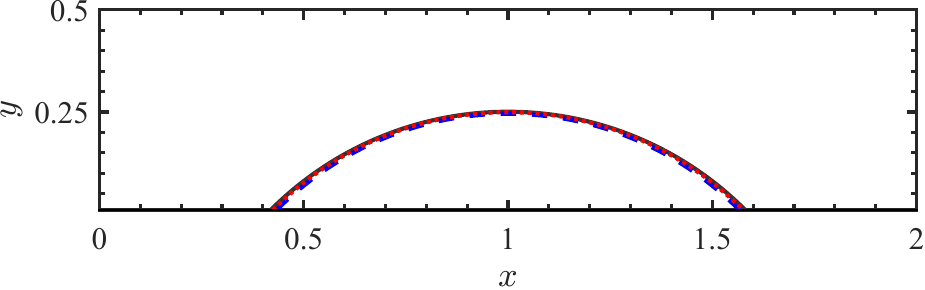} 
		\put(-205,55){(\textit{b})}
		\vspace{1mm}  
		\includegraphics[scale=0.45]{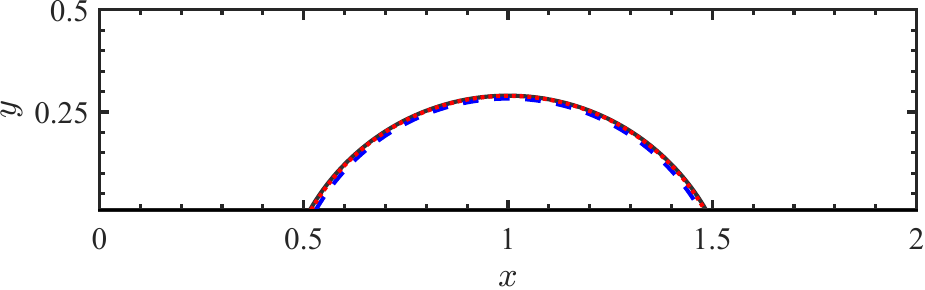} 
		\put(-205,55){(\textit{c})}
		\hspace{5mm}
		\includegraphics[scale=0.45]{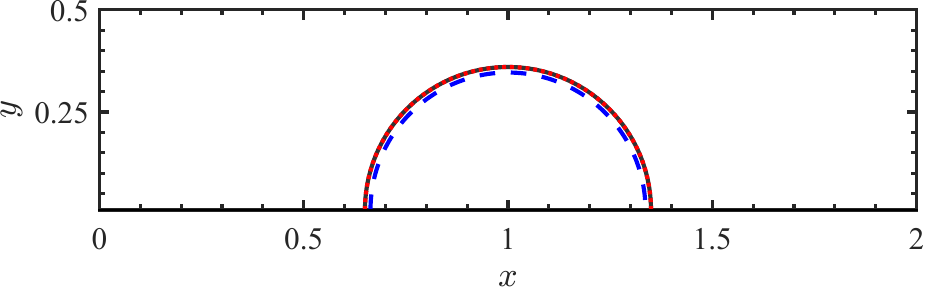} 
		\put(-205,55){(\textit{d})}
		\vspace{1mm}  
		\includegraphics[scale=0.45]{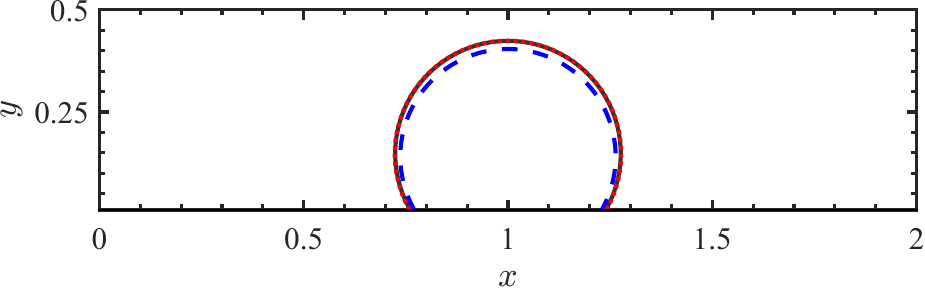} 
		\put(-205,55){(\textit{e})}
		\hspace{5mm}
		\includegraphics[scale=0.45]{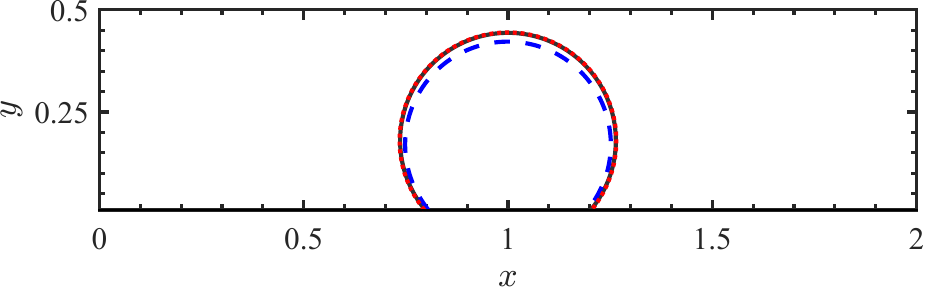} 
		\put(-205,55){(\textit{f})}
		\vspace{1mm}  
		\includegraphics[scale=0.45]{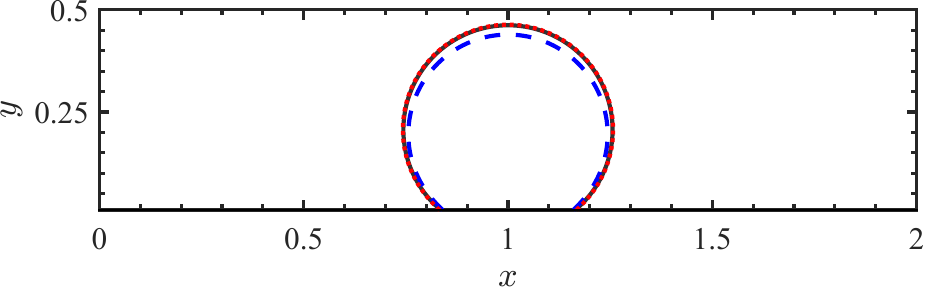} 
		\put(-205,55){(\textit{g})}
		\hspace{5mm}
		\includegraphics[scale=0.45]{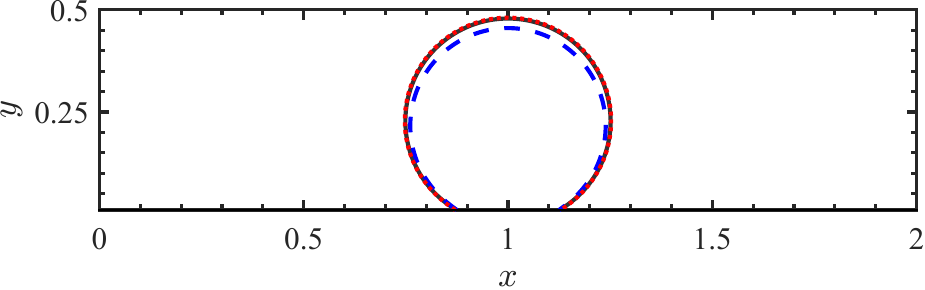} 
		\put(-205,55){(\textit{h})}
		\caption{ Comparison of the analytical solutions with the steady-state \(\phi=0\) contours based on the present nonlinear model and the classical linear-constraint model at (a) $\theta_0\textbf{e}=30^\circ$, (b) $45^\circ$, (c) $60^\circ$, (d) $90^\circ$, (e) $120^\circ$, (f) $130^\circ$, (g) $140^\circ$ and (h) $150^\circ$. }
		\label{fig5}
	\end{figure}

	In contrast, the static equilibrium in the present nonlinear model with improved volume conservation requires $\mu_\phi = \lambda_\phi(t) h_\phi'(\phi)$. In the homogeneous bulk regions far from the interface where $\phi \to \pm 1$, $h_\phi'(\pm 1) = 0$, the equilibrium condition would lead to $\mu_\phi = 0$. For the double-well free-energy density, the stable bulk equilibria corresponding to $\mu_\phi=0$ are $\phi=\pm1$. Hence, the nonlinear constraint does not induce the curvature-dependent bulk-phase shift that occurs in the linear model. Moreover, because the Lagrange-multiplier correction $\lambda_\phi h_\phi'(\phi)$ vanishes at the bulk equilibrium states, it only acts within the diffuse-interface region. This eliminates the leading-order $\mathcal{O}(\varepsilon/R)$ bulk shift and the associated additional geometric-area shrinkage, leaving only higher-order diffuse-interface contributions, expected to scale as $\mathcal{O}((\varepsilon/R)^2)$, together with spatial discretization errors.

	We first present a comparison of the numerical bulk-phase shift $\delta_\phi$ with the theoretical prediction given by Eq.~\eqref{ca1} in Fig. \ref{fig6}(a). It can be found that for the classical linear model, the shift $\delta_\phi$ increases with $1/R$ and agrees well with Eq.~\eqref{ca1}. A linear regression constrained to zero intercept yields a slope within $2.5\%$ of the theoretical value $\varepsilon/12$. This confirms that the bulk-phase shift in the linear model is primarily induced by the interfacial curvature. In contrast, the value of $\delta_\phi$ in the present model remains below $9.2 \times 10^{-4}$ for all cases considered, and exhibits no 
	apparent increase with $1/R$, which is consistent with the fact that $h_\phi'(\pm1) = 0$ eliminates the leading-order curvature-induced bulk shift. Fig. \ref{fig6}(b) compares the additional geometric area shrinkage of the classical linear model relative to the present nonlinear model, $(\mathcal{A}_N-\mathcal{A}_L)/\mathcal{A}_0$, with the theoretical prediction given by Eq.~\eqref{ca2}. The additional shrinkage increases monotonically with curvature and shows excellent agreement with the theoretical prediction. This result confirms that the curvature-induced bulk shift is the primary mechanism responsible for the additional geometric shrinkage in the linear model. Although the nonlinear model still suffers from the variation of residual area due to the effect of finite-interface-thickness and the error of spatial discretization, it eliminates the leading-order curvature-induced shrinkage inherent in the linear constraint.

	\begin{figure}[H]
		\centering
		\includegraphics[scale=0.45]{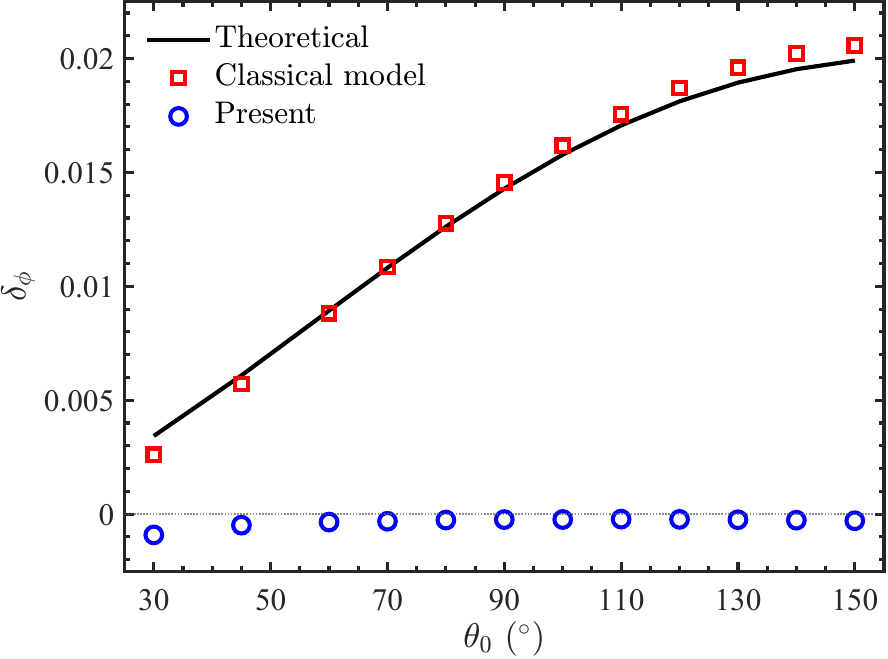} 
		\put(-195,135){(\textit{a})}
		\hspace{5mm}  
		\includegraphics[scale=0.45]{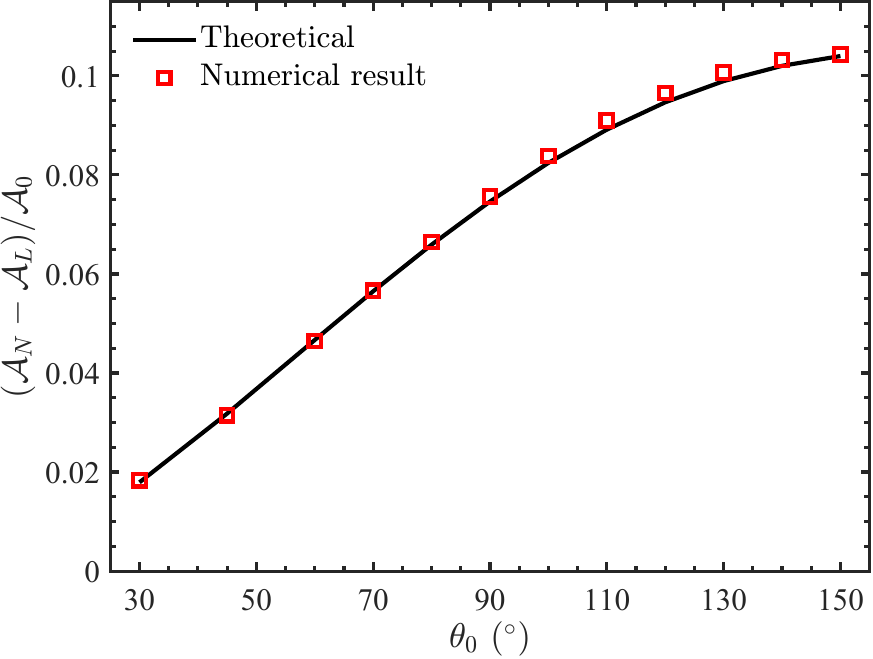} 
		\put(-195,135){(\textit{b})}
		\caption{ Curvature-induced bulk-phase shift and associated geometric-area loss. (a) Bulk-phase shift $\delta=(\phi_l+\phi_g)/2$ as a function of the specified equilibrium contact angle $\theta_0$. The solid line denotes the leading-order prediction \(\delta_{\rm th}=\varepsilon/(12R)\). The linear constraint exhibits a curvature-dependent bulk shift, whereas the nonlinear constraint maintains the shift near zero. (b) Additional geometric shrinkage of the classical model relative to the present model. Symbols denote $S_{\rm extra,num}=(\mathcal{A}_N-\mathcal{A}_L)/\mathcal{A}_0$, and the solid line denotes $S_{\rm extra,th}=|\Omega|\varepsilon/(24\mathcal{A}_0R)$. }
		\label{fig6}
	\end{figure}

	\subsection{Surfactant-covered droplet spreading on a wetting solid}

	To demonstrate the capacity of the present model in predicting the wetting response induced by solid adsorption, we conduct a comparison of the equilibrium contact angle $\theta$ with the numerical results reported by Kannan et al. \cite{KannanJFM2026}. To this end, three representative surface-adsorption conditions are considered: no solid adsorption $(\beta_{sl}=\beta_{sg}=0)$, preferential adsorption on the solid–liquid interface $(\beta_{sl}=0.9>\beta_{sg}=0.01)$, and preferential adsorption on the solid–gas interface $(\beta_{sl}=0.01<\beta_{sg}=0.9)$. For this problem, the computational domain is $\Omega=[0,3]\times[0,1]$ and is discretized by a $600\times200$ uniform lattice with $\Delta x=\Delta y=0.005$. As in the previous example, the periodic boundary conditions are imposed in the horizontal direction, the no-slip, wetting, and surfactant-adsorption boundary conditions are applied on the bottom wall, while the no-flux conditions are imposed to the phase-field variable and the surfactant concentration on the top boundary. A semicircular droplet with the radius $R_0=0.35$ is initially placed at the center $(1.5,0)$, and gravitational effect is neglected. Following the previous work \cite{KannanJFM2026}, the dimensionless parameters are set to be $\mathrm{Re}=10$, $\mathrm{Ca}=0.1$, and the clean fluid--fluid interfacial tension is $\sigma_{gl,0}=0.0107$. The bulk surfactant concentration and thermodynamic parameters are fixed as $\psi_b=0.01$, $\psi_c=0.017$, $\mathrm{Ex}=1$, $\mathrm{Pi}=0.18407$. Preferential solid--gas adsorption is realized through setting $\beta_{sg}=0.9, \beta_{sl}=0.01$, the clean solid--fluid interfacial tensions are determined from $\sigma_{sg,0}-\sigma_{sl,0}=\sigma_{gl,0}\cos\theta_0$, $\sigma_{sg,0}+\sigma_{sl,0}=16\sigma_{gl,0}/15$, giving $\sigma_{sg,0}=0.00709135$ and
	$\sigma_{sl,0}=0.00432198$. 
	
	\begin{figure}[H]
		\centering
		\includegraphics[scale=0.5]{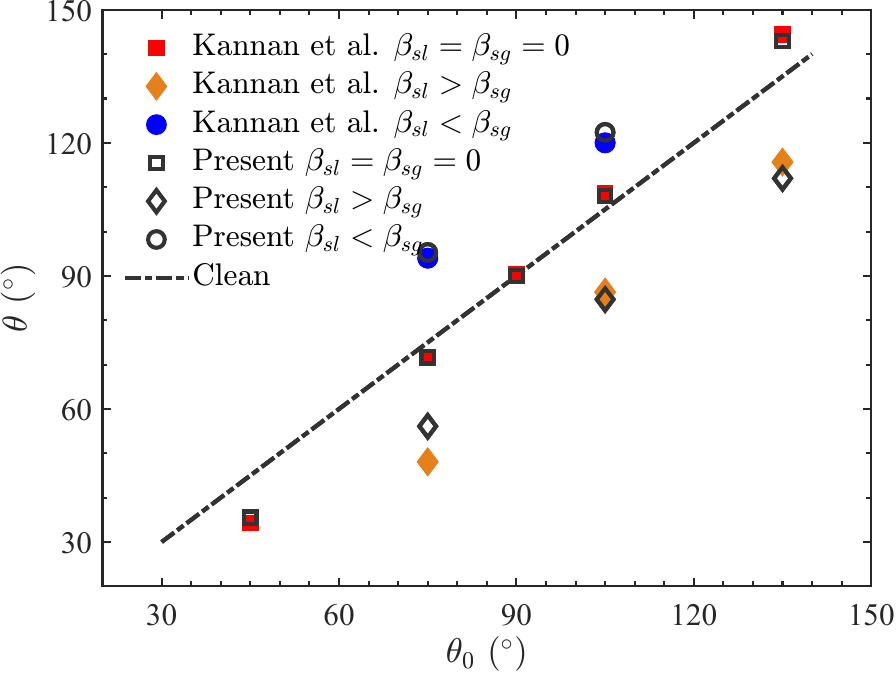} 
		\caption{ Comparison of the equilibrium contact angles $\theta$ predicted by the present method and those reported by Kannan et al. \cite{KannanJFM2026} for surfactant-laden droplets under three solid-adsorption conditions: without solid adsorption ($\beta_{sl}=\beta_{sg}=0$), preferential solid-liquid adsorption ($\beta_{sl}>\beta_{sg}$), and preferential solid-gas adsorption ($\beta_{sg}>\beta_{sl}$). }
		\label{fig7}
	\end{figure}

	We first present a comparison of the equilibrium contact angles $\theta$ predicted by the present model with those reported by Kannan et al.~\cite{KannanJFM2026} at three solid-adsorption conditions in Fig.~\ref{fig7}, in which the dashed line, $\theta=\theta_0$, represents the equilibrium contact angle of a surfactant-free droplet. Overall, the present results agree well with the reference data and reproduce the characteristic wetting behavior of all three cases. Without solid adsorption, the addition of surfactant amplifies the original wettability: the equilibrium contact angle decreases for hydrophilic case with $\theta_0<90^\circ$, remains approximately unchanged at $\theta_0=90^\circ$, and increases for hydrophobic case with $\theta_0>90^\circ$. This behavior is consistent with previous models neglecting solid adsorption \cite{ZhuJFM2019,ZhuJCP2020,WangJCP2024}. When the preferential solid--liquid adsorption is considered, $\beta_{sl}>\beta_{sg}$, the equilibrium contact angles are less than those of the clean droplet for all cases examined, indicating enhanced wetting. In contrast, the preferential solid--gas adsorption, $\beta_{sl}<\beta_{sg}$, shifts the equilibrium contact angles above those of the clean droplet, and produces an autophobing response. In general, the present model captures the three wetting regimes reported by Kannan et al. \cite{KannanJFM2026}, with only small differences.

	Fig.~\ref{fig8} illustrates the equilibrium distributions of surfactant concentration for the three solid-adsorption conditions ($\theta_0=105^\circ$, $\psi_b=0.01$), while Fig.~\ref{fig9} shows the corresponding equilibrium profiles of droplet over a range of intrinsic contact angles. According to the modified Young relation,
	\begin{equation}
		\sigma_{gl}(\psi)\cos\theta = \sigma_{sg}(\psi) - \sigma_{sl}(\psi) = \sigma_{gl,0}\cos\theta_0 + \sigma_{sg,0}\beta_{sg}\ln(1-\psi) - \sigma_{sl,0}\beta_{sl}\ln(1-\psi),
		\label{modified_young}
	\end{equation}
	the wetting response is governed by the bulk surfactant concentration $\psi$, fluid--fluid adsorption and preferential adsorption on the two solid--fluid interfaces. In the absence of solid adsorption ($\beta_{sl}=\beta_{sg}=0$), the surfactant accumulates primarily within the diffuse fluid–fluid interface, with no distinct adsorption layer forming along the solid wall, and then the modified Young’s relation reduces to $\sigma_{gl}(\psi_\Gamma)\cos\theta = \sigma_{gl,0}\cos\theta_0$. Because the surfactant adsorption would reduce $\sigma_{gl}$, the magnitude of $\cos\theta$ increases relative to $\cos\theta_0$. Consequently, the hydrophilic droplet ($\theta_0<90^\circ$) becomes more wetting, whereas the hydrophobic droplet ($\theta_0>90^\circ$) becomes less wetting; the neutral case ($\theta_0=90^\circ$) remains unchanged. These phenomena are reproduced by the green dashed profiles in Fig.~\ref{fig9}, and are consistent with the previous models that neglect solid adsorption \cite{ZhuJFM2019,ZhuJCP2020,WangJCP2022,WangJCP2024,WuAML2024,ZhangICHMT2025,YaoIJMF2026,WuCMAME2026}. Under preferential solid--liquid adsorption ($\beta_{sl}>\beta_{sg}$), the adsorption layer beneath the droplet, as shown in Fig.~\ref{fig8}(b), would reduce $\sigma_{sl}$ and thereby increase $\sigma_{sg}-\sigma_{sl}$. As a result, the blue dash-dotted profiles in Fig.~\ref{fig9} become wider and flatter, indicating enhanced wetting. The predicted surfactant distribution and wetting response are consistent with previous experimental and numerical results \cite{KannanJFM2026,JohnsonCS1986}. Conversely, when the adsorption at the solid--gas interface dominates ($\beta_{sg}>\beta_{sl}$), the surfactant accumulates along the solid wall outside the droplet, as shown in Fig.~\ref{fig8}(c). The resulting reduction in $\sigma_{sg}$ decreases $\sigma_{sg}-\sigma_{sl}$ and therefore increases the equilibrium contact angle. Consequently, the droplet becomes narrower and taller, corresponding to a reduction in wettability. This surfactant-induced dewetting behavior is commonly referred to as autophobing \cite{KannanJFM2026,BeraLangmuir2016}, and is clearly observed in Fig.~\ref{fig9}.

	\begin{figure}[H]
		\centering
		\includegraphics[scale=0.5]{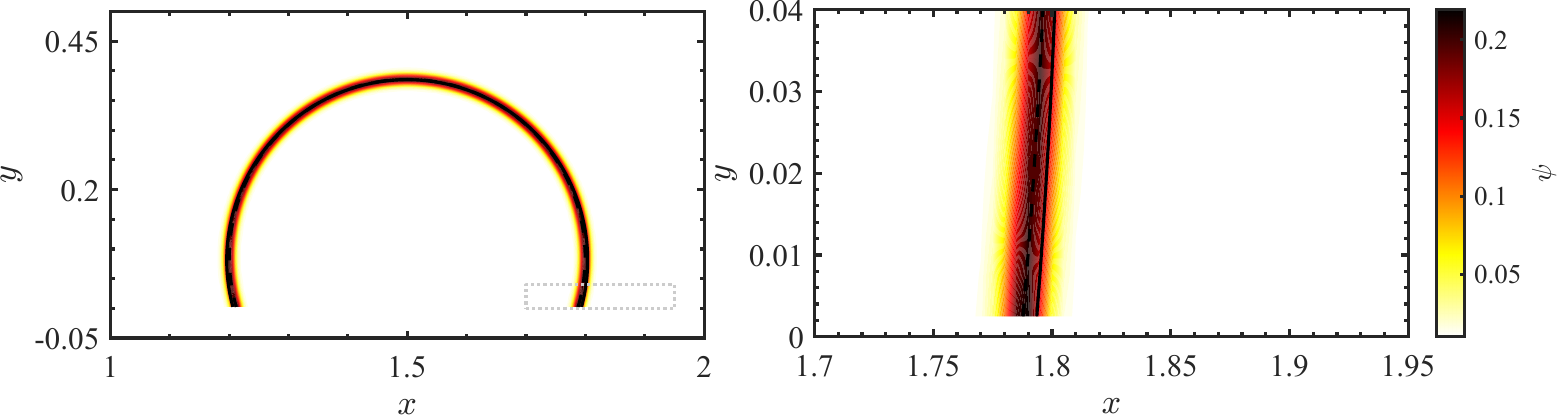} 
		\put(-380,90){(\textit{a})}
		\vspace{1mm}  
		\includegraphics[scale=0.5]{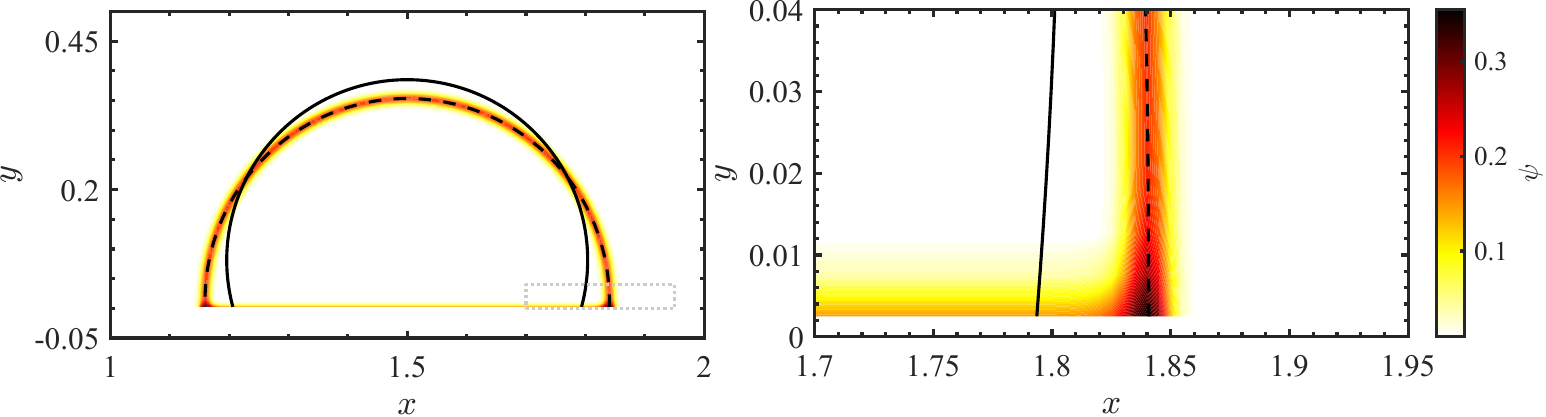} 
		\put(-380,90){(\textit{b})}
		\vspace{1mm}  
		\includegraphics[scale=0.5]{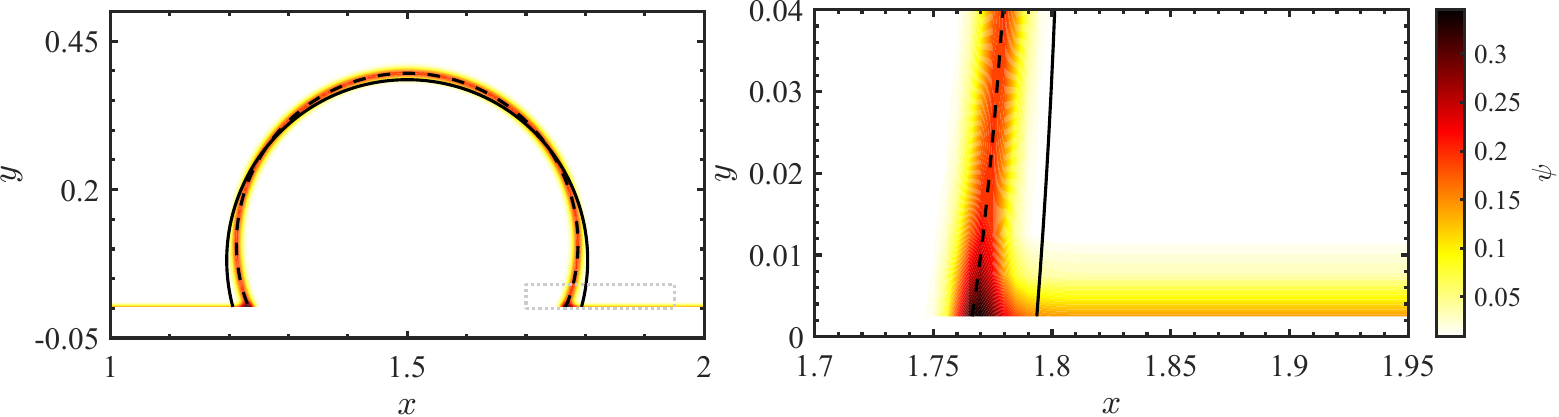} 
		\put(-380,90){(\textit{c})}		
		\caption{ Equilibrium distributions of surfactant concentration for a droplet with $\theta_0=105^\circ$ and $\psi_b=0.01$: (a) without solid adsorption ($\beta_{sl}=\beta_{sg}=0$), (b) preferential solid--liquid adsorption, $\beta_{sl}=0.9$ and $\beta_{sg}=0.01$ and (c) preferential solid--gas adsorption, $\beta_{sg}=0.9$ and $\beta_{sl}=0.01$. The solid black line indicates the equilibrium profile of a clean droplet, and the dashed black line denotes the $\phi=0$ contour representing the fluid–fluid interfaces. The right panels show enlarged views of the dotted boxes near the three-phase contact line in the left panels. }
		\label{fig8}
	\end{figure}

	In addition, the results in Fig.~\ref{fig10} demonstrate that increasing the bulk surfactant concentration enhances the influence of the adsorption preference determined by $\beta_{sl}$ and $\beta_{sg}$. When $\beta_{sl}>\beta_{sg}$, the increase of $\psi_b$ would reduce $\sigma_{sl}$, thus decreasing the equilibrium contact angle and promoting droplet spreading, whereas the opposite trend is observed for $\beta_{sl}<\beta_{sg}$, which follows directly from the surfactant-modified Young’s relation, $\sigma_{gl}\cos\theta=\sigma_{sg}-\sigma_{sl}$. Actually, within the parameter range considered here, a higher bulk surfactant concentration amplifies the adsorption-induced change in wettability, while the direction of this change is determined by the relative adsorption strengths at the solid--liquid and solid--gas interfaces. For the case of $\psi_b=0.05$, a localized deformation of the fluid--fluid interface is observed near the three-phase contact line, and the distribution of surfactant concentration is more nonuniform in the region where the solid--fluid adsorption layer intersects the diffuse fluid--fluid interface. This localized surfactant enrichment modifies the interfacial tensions and the associated force balance near the contact line. Consequently, the fluid--fluid interface undergoes a local curvature adjustment to re-establish the local capillary and wetting force balance, giving rise to the interface bending. This localized deformation is also reported in some previous works \cite{KannanJFM2026,GokhaleLangmuir2005}, indicating that phase-dependent solid-surface adsorption can affect not only the macroscopic equilibrium contact angle but also the local interface profile near the contact line.

	\begin{figure}[H]
		\centering
		\includegraphics[scale=0.45]{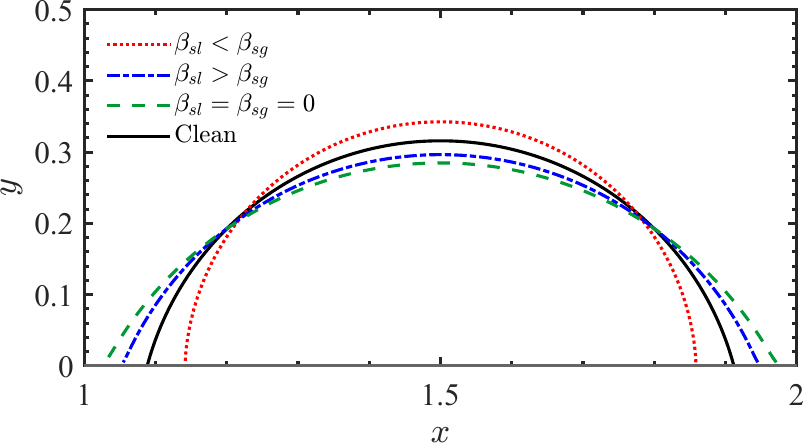} 
		\put(-185,90){(\textit{a})}
		\hspace{5mm}  
		\includegraphics[scale=0.45]{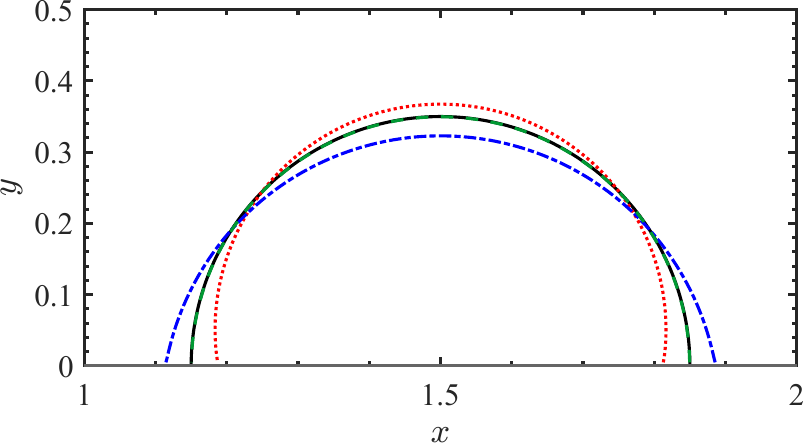} 
		\put(-185,90){(\textit{b})}
		\vspace{1mm}  
		\includegraphics[scale=0.45]{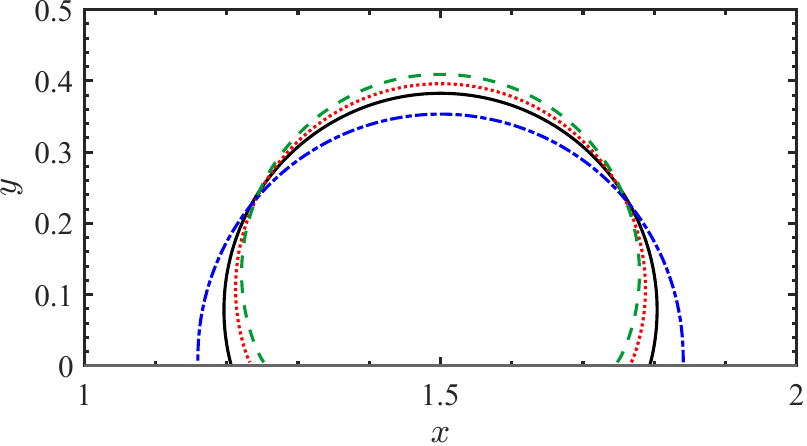} 
		\put(-185,90){(\textit{c})}
		\hspace{5mm}
		\includegraphics[scale=0.45]{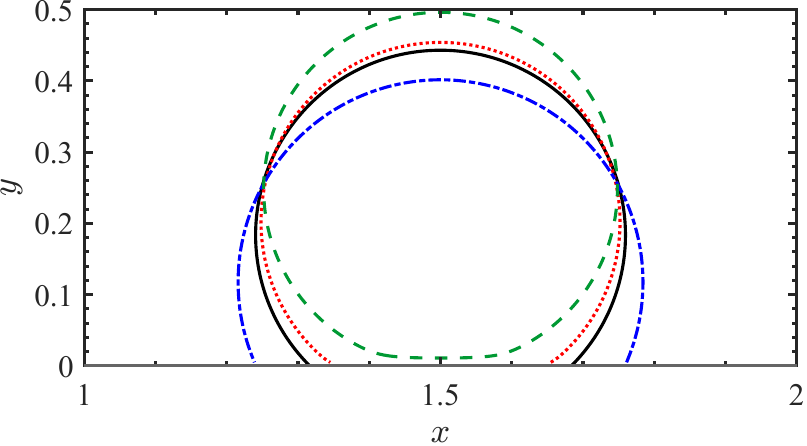} 
		\put(-185,90){(\textit{d})}
		\caption{Comparison of equilibrium profiles of droplet under different solid-adsorption conditions for intrinsic contact angles (a) $\theta_0=75^\circ$, (b) $90^\circ$, (c) $105^\circ$, and (d) $135^\circ$. The red dotted, blue dash-dotted, and green dashed curves represent preferential solid--gas adsorption ($\beta_{sg}=0.9$, $\beta_{sl}=0.01$), preferential solid--liquid adsorption ($\beta_{sl}=0.9$, $\beta_{sg}=0.01$), and no solid adsorption ($\beta_{sl}=\beta_{sg}=0$), respectively. The black solid curves denote the analytical equilibrium profiles of the corresponding surfactant-free droplets, while the numerical results are represented by the $\phi=0$ contours.}
		\label{fig9}
	\end{figure}

	\begin{figure}[H]
		\centering
		\includegraphics[scale=0.5]{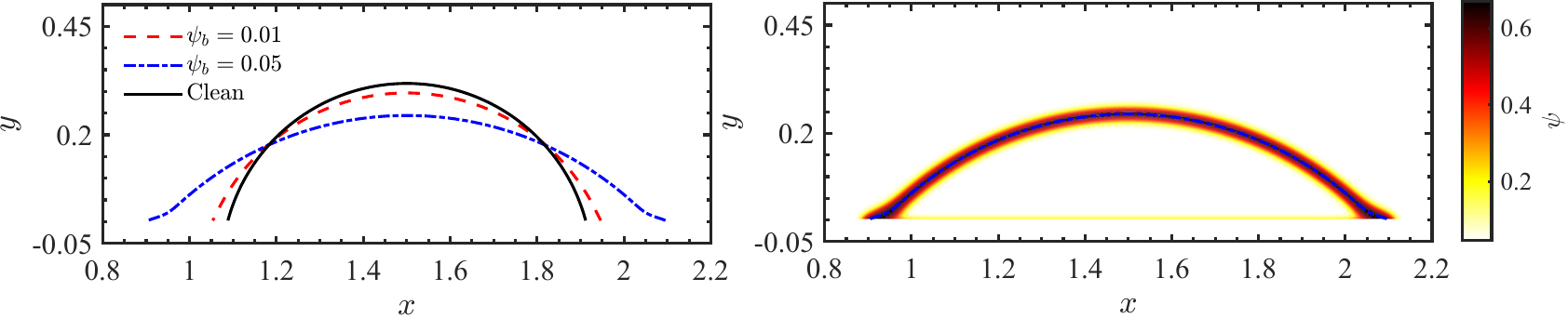} 
		\put(-410,75){(\textit{a})}
		\vspace{1mm}  
		\includegraphics[scale=0.5]{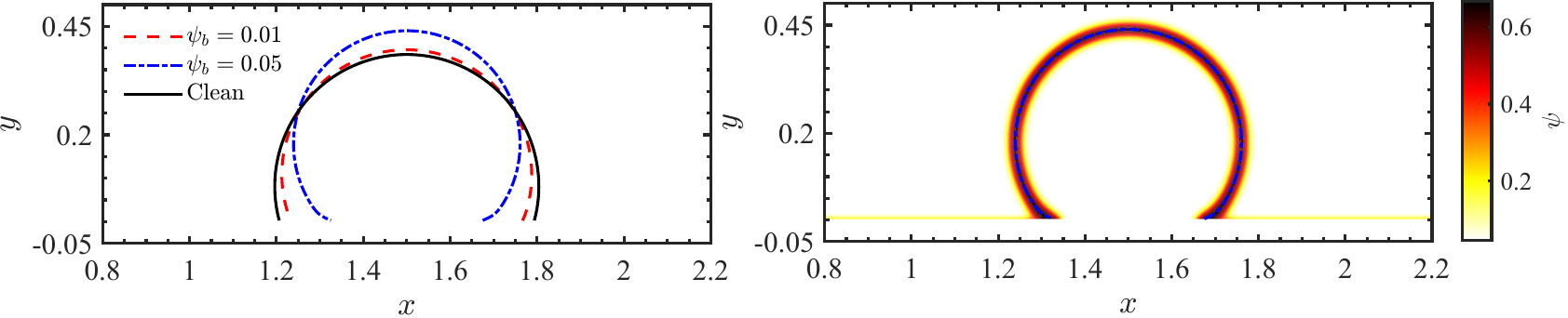} 
		\put(-410,75){(\textit{b})}
		\caption{ Effects of the bulk surfactant concentration $\psi_b$ on the equilibrium morphology and surfactant distribution of a sessile droplet with solid adsorption. (a) Preferential solid-liquid adsorption, $\beta_{sl}>\beta_{sg}$, at $\theta_0=75^\circ$; (b) preferential solid-gas adsorption, $\beta_{sl}<\beta_{sg}$, at $\theta_0=105^\circ$. In each panel, the left plot shows the equilibrium profiles at different values of $\psi_b$, while the right plot presents the distribution of surfactant concentration at $\psi_b=0.05$. }
		\label{fig10}
	\end{figure}

	\subsection{Surfactant-laden droplet on a wetting substrate under shear flow}
	To investigate the influence of preferential solid adsorption on the shear-induced deformation and detachment of a surfactant-laden droplet, a two-dimensional Couette-flow configuration is considered. The computational domain is $\Omega=[0,5]\times[0,1]$ and is discretized with $\Delta x=\Delta y=0.005$. Initially, a semicircular droplet with the radius $R_0=0.35$ is centered at $(x_c,y_c)=(1,0)$, and the two fluids have matched densities and viscosities. The periodic boundary conditions are imposed in the streamwise direction. The bottom wall is stationary and is subject to the wetting and surfactant-adsorption boundary conditions, whereas the top wall moves in the positive $x$ direction with the velocity $U_0=0.5$. The clean fluid--fluid interfacial tension, interface thickness, and bulk surfactant concentration are set to be $\sigma_{gl,0}=0.1$, $\varepsilon=0.025=5\Delta x$, and $\psi_b=0.01$. The remaining parameters are $M_\phi=M_\psi=0.1$, $\mathrm{Pi}=1.35$, $\mathrm{Ex}=0.5$, and $\theta_0=120^\circ$. These parameters correspond to $\mathrm{Ca}=0.5$ and $\mathrm{Re}=1.75$.

		\begin{figure}[H]
		\centering
		\includegraphics[scale=0.45]{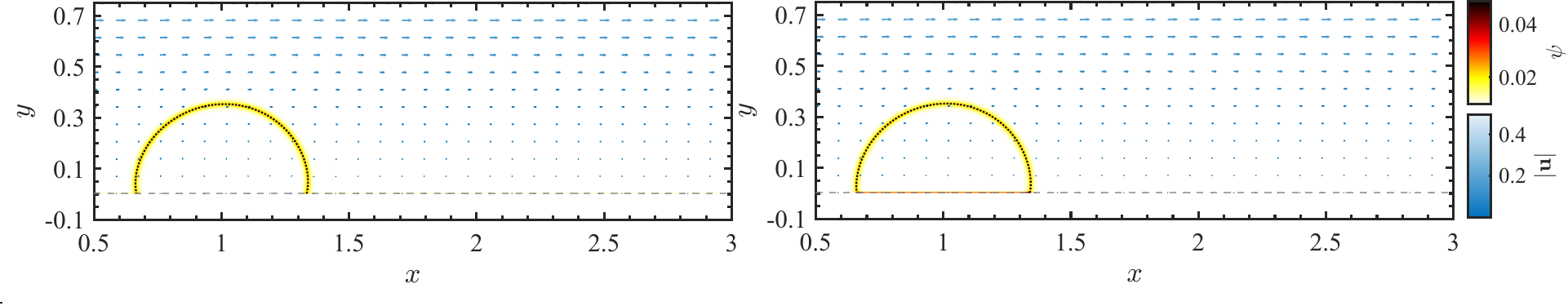} 
		\put(-415,74){(\textit{a})}
		\vspace{1mm}  
		\includegraphics[scale=0.45]{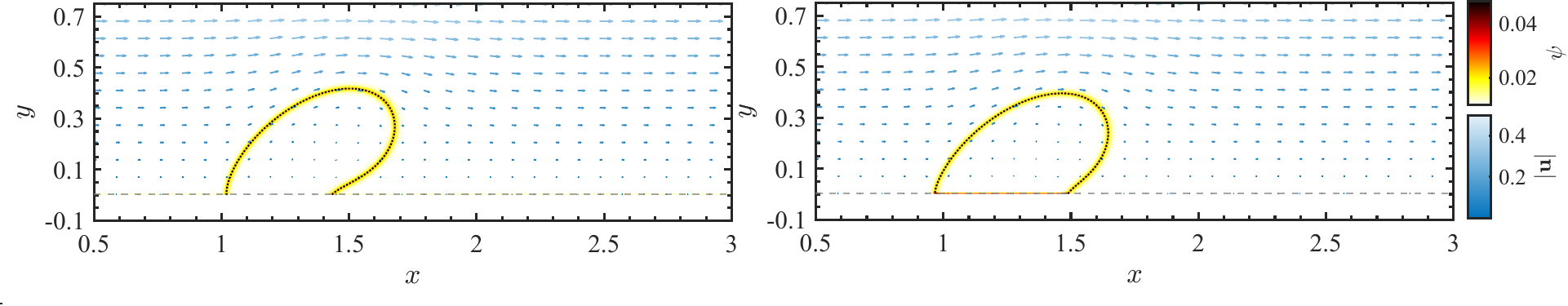} 
		\put(-415,74){(\textit{b})}
		\vspace{1mm}  
		\includegraphics[scale=0.45]{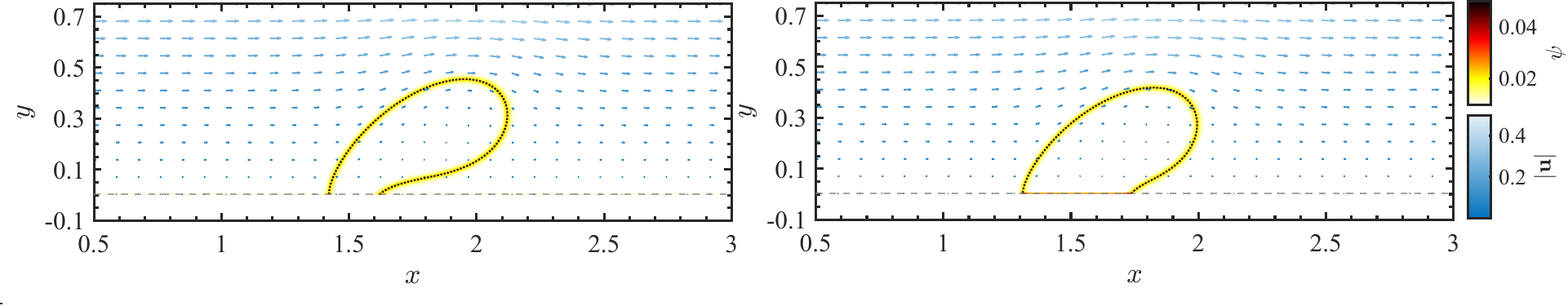} 
		\put(-415,74){(\textit{c})}
		\vspace{1mm}  
		\includegraphics[scale=0.45]{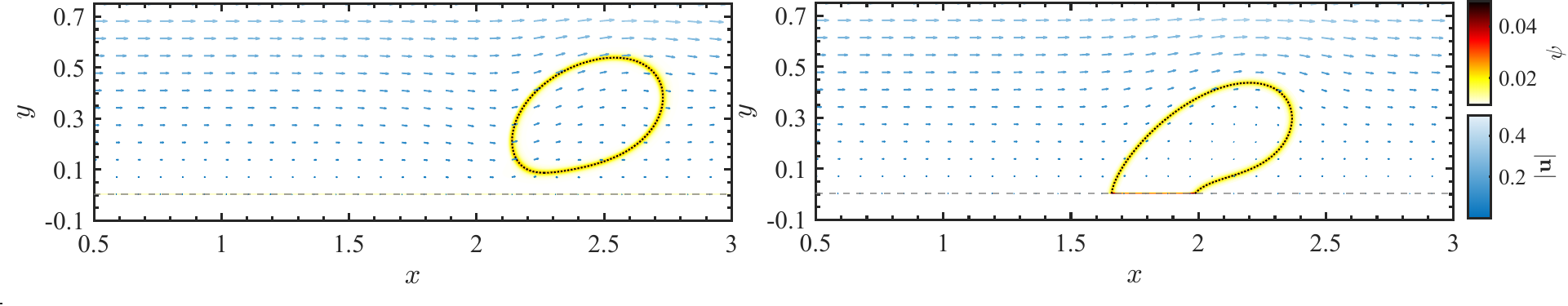} 
		\put(-415,74){(\textit{d})}
		\caption{ Evolution of a surfactant-laden droplet subjected to shear flow with $\mathrm{Ca}=0.5$. From top to bottom, the rows correspond to $t^*=U_0t/(2R_0)=0.595$, $8.929$, $17.857$ and $26.786$. The left column denotes preferential solid--gas adsorption ($\beta_{sg}=0.9$, $\beta_{sl}=0.01$), while the right column represents preferential solid--liquid adsorption ($\beta_{sl}=0.9$, $\beta_{sg}=0.01$). The colour shows the surfactant concentration $\psi$, the black dotted lines denote the fluid--fluid interface with $\phi=0$, and the arrows represent the velocity field, with the colour indicating the magnitude $|\mathbf{u}|$. }
		\label{fig11}
	\end{figure}

	Fig.~\ref{fig11} shows the dynamics of droplet under two preferential solid-adsorption conditions: preferential solid--gas adsorption ($\beta_{sg}=0.9, \beta_{sl}=0.01$) and preferential solid--liquid adsorption ($\beta_{sl}=0.9, \beta_{sg}=0.01$). At early moments, the droplets under both conditions are advected downstream and deform under the shear flow. Under the effect of shear flow, however, the contact-line dynamics of the two cases are significantly different. For preferential solid--gas adsorption, the solid--liquid contact width decreases from about $0.664$ at $t^*=0.595$ to $0.411$ and $0.195$ at $t^*=8.929$ and $17.857$, respectively. The droplet subsequently detaches from the substrate at $t^*=24.524$ and becomes fully suspended above the wall by $t^*=26.786$. In this case, the preferential adsorption at the solid--gas interface reduces $\sigma_{sg}$, thereby decreasing $\sigma_{sg}-\sigma_{sl}$ and favoring contact line recession. The resulting autophobic tendency further facilitates shear-induced detachment. In contrast, preferential solid--liquid adsorption produces a surfactant-rich layer beneath the droplet, with the wall surfactant concentration reaching $\psi_w\simeq0.426$. The associated reduction in $\sigma_{sl}$ increases $\sigma_{sg}-\sigma_{sl}$ and favors wetting and substrate attachment. Consequently, although the droplet is progressively elongated and advected downstream, it remains attached to the substrate, with contact widths being about $0.678$, $0.521$, $0.425$, and $0.326$ at four selected times. These results demonstrate that preferential surfactant adsorption at the solid--gas and solid--liquid interfaces can produce qualitatively different droplet responses under the same shear flow: preferential solid--gas adsorption promotes dewetting and detachment, whereas preferential solid--liquid adsorption stabilizes the wetted state.

	\section{Conclusions}
	\label{sec5}

	In this work, we developed a thermodynamically consistent, conservative Allen–Cahn–Navier–Stokes model for surfactant-laden contact line dynamics with solid adsorption, which is solved efficiently by a multiple-relaxation-time lattice Boltzmann method. In this model, the second-order Allen–Cahn equations are employed for both the phase field and surfactant transport, thereby avoiding the higher-order spatial derivatives associated with Cahn–Hilliard formulations. Surfactant adsorption at fluid–fluid, solid–liquid, and solid–gas interfaces is incorporated through a unified free-energy functional, from which the wetting boundary conditions consistent with the surfactant-modified Young’s relation are derived.

	To enforce the conservation properties, a linear constraint is introduced to ensure conservation of total surfactant mass, while a nonlinear interface-localized constraint is constructed to preserve a nonlinear diffuse phase-volume measure. This conserved mapping provides a superior diffuse approximation to the Heaviside step function while satisfying the energy-dissipation law. It should be noted that the resulting weighted correction vanishes in the bulk phases and only acts within the diffuse interface, eliminating the leading-order curvature-induced bulk-phase shift associated with classical linear models. Through the matched asymptotic analysis, we prove the formal second-order geometric-volume conservation of the formulation. In addition, the numerical tests further show that the model achieves second-order geometric-volume accuracy, exhibits low spurious velocities, and preserves global conservation. For a droplet spreading on a wetting substrate, the nonlinear volume constraint can not only suppress the artificial shrinkage observed with classical linear models, but also retain high accuracy in geometric volume preservation across a wide range of contact angles and interfacial curvatures. Finally, it is also found that the present model can reproduce three characteristic surfactant-mediated wetting behaviors: amplification of intrinsic wettability without solid adsorption, hydrophilization under preferential solid–liquid adsorption, and autophobing under preferential solid–gas adsorption. For a surfactant-laden droplet on the wetting substrate under shear flow, preferential solid–liquid adsorption stabilizes droplet attachment, whereas preferential solid–gas adsorption promotes contact line recession and shear-induced detachment. In future work, we will extend the framework to multiphase flows with solid adsorption in complex geometries.

	\section*{Acknowledgments}
	This work was financially supported by the National Natural Science Foundation of China (Grant Nos. 12672330 and 12501599), the Open Research Fund of State Key Laboratory of Mesoscience and Process Engineering (Grant No. MESO-25-D04) and the Interdisciplinary Research Program of Hust (Grant No. 2024JCYJ001). The computations were performed on the HPC platform of Huazhong University of Science and Technology \& IIRMAS Cluster.

	\appendix
	\section{Matched-asymptotic expansion}
	\label{APA}
	We construct the matched asymptotic expansions by using the intrinsic transition length $\ell_\phi$ as the small parameter. In the homogeneous bulk regions far from $\Gamma_\varepsilon(t)$, the outer fields and the Lagrange multiplier can be expanded in powers of $\ell_\phi$
	\begin{equation}
		u_{\text{out}}^\pm = u_0^\pm + \ell_\phi u_1^\pm + \ell_\phi^2 u_2^\pm + \mathcal{O}(\ell_\phi^3),\quad \mu_{\phi,\text{out}}^\pm = \ell_\phi^{-1}\mu_{-1}^\pm + \mu_0^\pm + \ell_\phi \mu_1^\pm + \mathcal{O}(\ell_\phi^2),\quad 
		\lambda_\phi(t) = \lambda_0(t) + \ell_\phi \lambda_1(t) + \mathcal{O}(\ell_\phi^2).
	\end{equation}
	Substituting the outer expansions into the definitions of the chemical potential and weighting function and collecting terms at successive orders of \(\ell_\phi\) yields
	\begin{equation}
		\begin{aligned}
			&\mu_{-1}^\pm
			=2\gamma_\phi[(u_0^\pm)^3-u_0^\pm],
			\quad
			\mu_0^\pm
			=2\gamma_\phi[3(u_0^\pm)^2-1]u_1^\pm, \quad \mu_1^\pm = 2\gamma_\phi \left\{ [3(u_0^\pm)^2 - 1]u_2^\pm + 3u_0^\pm(u_1^\pm)^2 \right\} - \gamma_\phi \nabla^2 u_0^\pm,\\
			& w(u_{\text{out}}^\pm) = w_0^\pm + \ell_\phi w_1^\pm + \mathcal{O}(\ell_\phi^2),\quad w_0^\pm = a_\phi [1 - (u_0^\pm)^2], \quad w_1^\pm = -2a_\phi u_0^\pm u_1^\pm,
		\end{aligned}
		\label{eq:app_gv_outer_mu_coefficients}
	\end{equation}
	where $a_\phi = 3/(4\phi_b)$. Substituting these expansions into Eq. (\ref{eq_AC}) gives rise to the outer equations at different orders,
	\begin{equation}
		\begin{aligned}
			\mathcal{O}(\ell_\phi^{-1})&: \quad 0=-M_\phi\mu_{-1}^\pm, \\
			\mathcal{O}(1)&: \quad \phi_bD_tu_0^\pm =-M_\phi(\mu_0^\pm-\lambda_0w_0^\pm),\\
			\mathcal{O}(\ell_\phi)&:
			\quad
			\phi_bD_tu_1^\pm
			=-M_\phi(\mu_1^\pm-\lambda_0w_1^\pm-\lambda_1w_0^\pm).
		\end{aligned}
	\end{equation}
	Here, $D_t =\partial_t+\mathbf u\cdot\nabla$. At the order  $\mathcal{O}(\ell_\phi^{-1})$, the outer equation gives $(u_0^\pm)^3 - u_0^\pm = 0$, which leads to the stable bulk equilibrium values $u_0^\pm = \pm1$. At the order $\mathcal{O}(1)$, since $u_0^\pm$ is constant ($D_t u_0^\pm = 0$) and the constraint kernel vanishes at the bulk states ($w_0^\pm = w(\pm 1) = 0$), the equation gives $\mu_0^\pm = 4\gamma_\phi u_1^\pm = 0$, directly yielding $u_1^\pm = 0$. At the order $\mathcal{O}(\ell_\phi)$, using $u_1^\pm = 0$ and $\nabla^2 u_0^\pm = 0$, the equation similarly yields $u_2^\pm = 0$. Consequently, the vanishing of the weighting function at the bulk equilibrium states prevents the nonlocal Lagrange multiplier from shifting the bulk equilibrium values through the first two correction orders. In particular, the bulk plateau corrections vanish at the orders $\mathcal{O}(\ell_\phi)$ and $\mathcal{O}(\ell_\phi^2)$, formally yielding
	\begin{equation}
		u_{\rm out}^\pm = \pm 1 + \mathcal{O}(\ell_\phi^3).
	\end{equation}

	In the tubular neighborhood $\mathcal{U}_\delta = \{\mathbf{x}\in\Omega : \vert{}r(\mathbf{x},t)\vert{} < \delta\}$ of $\Gamma_\varepsilon(t)$, any point $\mathbf x\in\mathcal U_\delta(t)$ can be uniquely represented in the local coordinates $(s,z)$ as $\mathbf{x} = \mathbf{X}(s,t) + \ell_\phi z\,\mathbf{n}(s,t)$ with $z = r(\mathbf{x},t)/\ell_\phi$, where $\mathbf{X}(s,t)$ is the parametrization of $\Gamma_\varepsilon(t)$, $s$ denotes the local tangential coordinates on the interface, $\mathbf{n} = \nabla r$ is the unit normal vector pointing into the positive phase, and $z$ is the stretched normal coordinate. In this coordinate system, the differential operators and the volume element can be expanded as
	\begin{equation}
		\begin{aligned}
			& \nabla = \ell_\phi^{-1} \mathbf{n} \partial_z + \nabla_s + \mathcal{O}(\ell_\phi), \quad \nabla^2 = \ell_\phi^{-2}\partial_{zz} + \ell_\phi^{-1}\mathcal{K}\partial_z + \mathcal{O}(1), \\ 
			& D_t =\ell_\phi^{-1}(v_n - V_n) \partial_z + D_t^\Gamma + \mathcal{O}(\ell_\phi),\quad \mathrm{d}\Omega = \ell_\phi \left[ 1 + \ell_\phi z \mathcal{K} + \mathcal{O}(\ell_\phi^2 z^2) \right] \mathrm{d}z\,\mathrm{d}S,
		\end{aligned}
	\end{equation}
	where $V_n = -\partial_t r\vert{}_{\Gamma_\varepsilon}$ is the interface normal velocity, $v_n = \mathbf{u}\cdot\mathbf{n}$ is the fluid normal velocity, $\mathcal{K} = \nabla \cdot \mathbf{n}$ is the total mean curvature, $D_t^\Gamma$ represents the tangential material derivative along $\Gamma_\varepsilon(t)$, and $\mathrm{d}S$ denotes the hypersurface measure element on $\Gamma_\varepsilon(t)$. Assuming that the velocity field is smooth across the tubular neighborhood, the normal fluid velocity admits the expansion $v_n = v_n^\Gamma + \ell_\phi z \partial_n v_n\vert{}_\Gamma + \mathcal{O}(\ell_\phi^2)$. In the inner region, we can expand the phase field, chemical potential, and normal velocity as
	\begin{equation}
		U(z,s,t) = U_0 + \ell_\phi U_1 + \ell_\phi^2 U_2 + \mathcal{O}(\ell_\phi^3), \quad \overline{\mu}_\phi = \ell_\phi^{-1}\overline{\mu}_{-1} + \overline{\mu}_0 + \ell_\phi \overline{\mu}_1 + \mathcal{O}(\ell_\phi^2), \quad V_n = V_n^{(0)} + \ell_\phi \mathcal{V}_1 + \mathcal{O}(\ell_\phi^2),
	\end{equation}
	which leads to the following expressions of inner chemical potential
	\begin{equation}
		\overline{\mu}_{-1} = \gamma_\phi \left[ 2(U_0^3 - U_0) - \partial_{zz} U_0 \right], \quad \overline{\mu}_0 = \gamma_\phi \left[ 2(3U_0^2 - 1)U_1 - \partial_{zz} U_1 - \mathcal{K}\partial_z U_0 \right].
	\end{equation}
	According to the asymptotic matching principle, the inner solution $U(z,s,t)$ must match the Taylor expansion of the outer solution $u_{\text{out}}^\pm(\mathbf{x},t)$ as $z \to \pm\infty$
	\begin{equation}
		\begin{aligned}
			\lim_{z \to \pm\infty} U(z,s,t) \sim u_{\text{out}}^\pm(\mathbf{X} + \ell_\phi z\mathbf{n}, t) = u_0^\pm(\mathbf{X},t) + &\ell_\phi \left[ u_1^\pm(\mathbf{X},t) + z \partial_n u_0^\pm(\mathbf{X},t) \right] + \\ &\ell_\phi^2 \left[ u_2^\pm(\mathbf{X},t) + z \partial_n u_1^\pm(\mathbf{X},t) + \frac{z^2}{2} \partial_{nn} u_0^\pm(\mathbf{X},t) \right] + \mathcal{O}(\ell_\phi^3).
		\end{aligned}
	\end{equation}
	According to $u_0^\pm = \pm 1$, $u_1^\pm = 0$  and $u_2^\pm = 0$, the matching conditions provide the far-field boundary conditions for the inner equations,
	\begin{equation}
		\lim_{z \to \pm\infty} U_0(z) = \pm 1, \quad \lim_{z \to \pm\infty} U_1(z) = 0,\quad \lim_{z \to \pm\infty} U_2(z) = 0, \quad \lim_{z \to \pm\infty} \partial_z U_1(z) = 0.
		\label{matching_1}
	\end{equation}
	along with $\lim_{z \to \pm\infty} \overline{\mu}_{-1}(z) = 0$ and $\lim_{z \to \pm\infty} \overline{\mu}_0(z) = 0$. Since $r=0$ is defined as the zero-level set of the phase field, the order-by-order centering conditions are given by
	\begin{equation}
		U_0(0) = 0, \quad U_1(0) = 0, \quad U_2(0) = 0.
		\label{matching_2}
	\end{equation}
	The leading-order inner equation is 
	\begin{equation}
		\mathcal{O}(\ell_\phi^{-1}):\phi_b (v_n^\Gamma - V_n^{(0)})\partial_z U_0 = -M_\phi \overline{\mu}_{-1} = -M_\phi \gamma_\phi \left[ 2(U_0^3 - U_0) - \partial_{zz} U_0 \right].
	\end{equation} 
	Multiplying $\partial_z U_0$ and integrating over $z \in \mathbb{R}$, the right-hand integral vanishes because it corresponds to the boundary difference of the symmetric double-well energy,
	\begin{equation}
		\int_{-\infty}^{+\infty} \left[ 2(U_0^3 - U_0)\partial_z U_0 - (\partial_{zz}U_0)(\partial_z U_0) \right] \mathrm{d}z = \left[ \frac{1}{2}(U_0^2 - 1)^2 - \frac{1}{2}(\partial_z U_0)^2 \right]_{-\infty}^{+\infty} = 0.
	\end{equation}
	Due to the fact that $\int_{-\infty}^{+\infty}(\partial_z U_0)^2\,\mathrm{d}z > 0$, we can find $V_n^{(0)} = v_n^\Gamma$ and $\overline{\mu}_{-1} = 0$, and the above equation reduces to $-\partial_{zz} U_0 + 2(U_0^3 - U_0) = 0$, subject to $U_0(\pm\infty) = \pm 1$ and $U_0(0) = 0$. In this case, one can obtain 
	\begin{equation}
		U_0(z) = q_0(z) = \tanh z, \quad q_0'(z) = 1 - q_0^2(z) = \operatorname{sech}^2 z.
	\end{equation}
	We now define the self-adjoint linear operator $\mathcal{L}f = \gamma_\phi [2(3q_0^2-1)f - f'']$, whose nullspace is spanned by the translational zero mode $q_0'$, which satisfies $\mathcal Lq_0'=0$. The inner equation at order $\mathcal O(1)$ then takes the following form
	\begin{equation}
		 \mathcal{O}(1): \quad -\phi_b (\mathcal{V}_1 - z\partial_n v_n\vert{}_\Gamma)\partial_z U_0 = -M_\phi (\overline{\mu}_0 - \lambda_0 w_0)= -M_\phi \left[ \mathcal{L}U_1 - \gamma_\phi \mathcal{K} q_0' - \lambda_0 w(q_0) \right],
		 \label{inner_1}
	\end{equation}
	where $w(q_0) = a_\phi q_0'$. Multiplying this equation by the zero mode $q_0'$ and integrating over $z \in \mathbb{R}$, we have
	\begin{equation}
		-\phi_b \int_{-\infty}^{+\infty} (\mathcal{V}_1 - z\partial_n v_n\vert{}_\Gamma) (q_0')^2 \,\mathrm{d}z = -M_\phi \int_{-\infty}^{+\infty} \left[ \mathcal{L}U_1 - \gamma_\phi \mathcal{K} q_0' - \lambda_0 w(q_0) \right] q_0' \,\mathrm{d}z,
	\end{equation}
	Since $z (q_0')^2$ on the left-hand side of the above equation is an odd function, the second term vanishes by symmetry,
	\begin{equation}
		\text{LHS} = -\phi_b \mathcal{V}_1 \int_{-\infty}^{+\infty} (q_0')^2 \,\mathrm{d}z + \phi_b \partial_n v_n\vert{}_\Gamma \int_{-\infty}^{+\infty} z (q_0')^2 \,\mathrm{d}z = -\phi_b \mathcal{V}_1 c_0, \quad c_0 = \int_{-\infty}^{+\infty} \operatorname{sech}^4 z \,\mathrm{d}z = \frac{4}{3}.
		\label{left0}
	\end{equation}
	Using the self-adjointness of $\mathcal{L}$ ($\int_{-\infty}^{+\infty} q_0' \mathcal{L}U_1\,\mathrm{d}z = \int_{-\infty}^{+\infty} U_1 \mathcal{L}q_0'\,\mathrm{d}z = 0$) and the substitution $u = q_0(z)$, the term on the right-hand side can be written as
	\begin{equation}
		\text{RHS} = M_\phi \gamma_\phi \mathcal{K} \int_{-\infty}^{+\infty} (q_0')^2 \,\mathrm{d}z + M_\phi \lambda_0 \int_{-1}^{1} \frac{3}{4\phi_b}(1 - u^2) \,\mathrm{d}u = M_\phi \gamma_\phi \mathcal{K} c_0 + M_\phi \frac{\lambda_0}{\phi_b}.
		\label{right0}
	\end{equation}
	Combining Eqs. (\ref{left0}) and (\ref{right0}), one can obtain the Fredholm solvability condition,
	\begin{equation}
		\left( \frac{\phi_b \mathcal{V}_1}{M_\phi} + \gamma_\phi \mathcal{K} \right) c_0 + \frac{\lambda_0}{\phi_b} = 0.
		\label{FSC}
	\end{equation}
	To determine the parameter $\lambda_0$, we evaluate the continuous multiplier identity asymptotically. Noting that outer bulk contributions are exponentially small and using the self-adjoint orthogonality $\int_{\mathbb{R}} q_0'\mathcal{L}U_1\,\mathrm{d}z = 0$, the domain integrals can be evaluated at leading order,
	\begin{equation}
		\int_\Omega [w(u)]^2\,\mathrm{d}\Omega = \ell_\phi a_\phi^2 c_0 \vert{}\Gamma_\varepsilon\vert{} + \mathcal{O}(\ell_\phi^2), \qquad \int_\Omega w(u)\mu_\phi\,\mathrm{d}\Omega = -\ell_\phi a_\phi \gamma_\phi c_0 \int_{\Gamma_\varepsilon} \mathcal{K}\,\mathrm{d}S + \mathcal{O}(\ell_\phi^2).
	\end{equation}
	Dividing the two asymptotic expressions and using $a_\phi^{-1} = \phi_b c_0$ yields $\lambda_0 = -\gamma_\phi \phi_b c_0 \overline{\mathcal{K}}$ with $\overline{\mathcal{K}} = \frac{1}{\vert{}\Gamma_\varepsilon\vert{}} \int_{\Gamma_\varepsilon} \mathcal{K}\,\mathrm{d}S$. This self-consistently confirms that $\lambda_\phi = \mathcal{O}(1)$. Substituting $\lambda_0$ back into Eq. (\ref{FSC}) determines the first-order correction to normal velocity $\mathcal{V}_1 = -M_\phi \gamma_\phi \phi_b^{-1} (\mathcal{K} - \overline{\mathcal{K}})$. Thus, the inner equation for $U_1$ in Eq. (\ref{inner_1}) can be simplified by
	\begin{equation}
		\mathcal{L}U_1 = -\frac{\phi_b}{M_\phi} z \partial_n v_n\vert{}_\Gamma q_0'.
		\label{eq_L}
	\end{equation}
	To solve this linear inhomogeneous equation (\ref{eq_L}) under the conditions (\ref{matching_1}) and (\ref{matching_2}), we can apply the reduction-of-order substitution $U_1(z) = v(z)q_0'(z)$ to derive
	\begin{equation}
		\frac{\mathrm{d}}{\mathrm{d}z}\left[ (q_0'(z))^2 v'(z) \right] = \frac{\phi_b}{\gamma_\phi M_\phi} z \partial_n v_n\vert{}_\Gamma (q_0'(z))^2.
	\end{equation}
	Integrating this relation twice with $U_1(0) = 0$ yields the following explicit solution,
	\begin{equation}
		U_1(z) = \frac{\phi_b}{\gamma_\phi M_\phi} \partial_n v_n\vert{}_\Gamma q_0'(z) \int_0^z \frac{1}{(q_0'(\eta))^2} \left[ \int_{-\infty}^\eta \zeta (q_0'(\zeta))^2 \,\mathrm{d}\zeta \right] \mathrm{d}\eta.
	\end{equation}
	By expressing $U_1(z)$ as a quotient and applying L'Hôpital's rule twice, we obtain $\lim_{z\to\pm\infty}U_1(z)=0$. Since $\partial_n v_n\vert{}_\Gamma$ is uniformly bounded on $\Gamma_\varepsilon(t) $, the correction is uniformly bounded, with $\Vert{}U_1\Vert{}_{L^\infty} \le C$.

	To conclude, the outer solution satisfies $u_{\text{out}}^\pm = \pm 1 + \mathcal{O}(\ell_\phi^3)$ and the inner solution is $U(z,s,t) = q_0(z) + \ell_\phi U_1(z,s,t) + \mathcal{O}(\ell_\phi^2)$. Combining the inner and outer solutions through matched asymptotic assembly, we can obtain the second-order uniform asymptotic approximation in the entire domain $\Omega$,
	\begin{equation}
		u_\varepsilon(\mathbf{x},t) = u_{\text{unif},0}(\mathbf{x},t) + \ell_\phi u_{\text{unif},1}(\mathbf{x},t) + \mathcal{O}(\ell_\phi^2),
	\end{equation}
	where the leading-order profile is given by the hyperbolic-tangent transition,
	\begin{equation}
		u_{\text{unif},0}(\mathbf{x},t) = q_0\left(\frac{r(\mathbf{x},t)}{\ell_\phi}\right) = \tanh\left(\frac{r(\mathbf{x},t)}{\ell_\phi}\right),
	\end{equation}
	and there exists a positive constant $C$, independent of $\ell_\phi$, such that the first-order composite correction is uniformly bounded,
	\begin{equation}
		\Vert{}u_{\text{unif},1}\Vert{}_{L^\infty(\Omega)} \le C.
	\end{equation}




\begin{thebibliography}{99}
		
		
		\bibitem{TakagiARFM2011}
		Takagi S, Matsumoto Y. Surfactant effects on bubble motion and bubbly flows. Annual Review of Fluid Mechanics, 2011, 43(1): 615-636.
		
		\bibitem{MyersJWS2020}
		Myers D. Surfactant science and technology. John Wiley \& Sons, 2020.
		
		\bibitem{SeethepalliSPEJ2004}
		Seethepalli A, Adibhatla B, Mohanty K K. Physicochemical interactions during surfactant flooding of fractured carbonate reservoirs. SPE Journal, 2004, 9(04): 411-418.
		
		\bibitem{BaretLC2012}
		Baret J C. Surfactants in droplet-based microfluidics. Lab on a Chip, 2012, 12(3): 422-433.
		
		
		\bibitem{ZhangJCP2006}
		Zhang J, Eckmann D M, Ayyaswamy P S. A front tracking method for a deformable intravascular bubble in a tube with soluble surfactant transport. Journal of Computational Physics, 2006, 214(1): 366-396.
		
		
		\bibitem{MuradogluJCP2008}
		Muradoglu M, Tryggvason G. A front-tracking method for computation of interfacial flows with soluble surfactants. Journal of Computational Physics, 2008, 227(4): 2238-2262.
		
		
		\bibitem{LaiJCP2008}
		Lai M C, Tseng Y H, Huang H. An immersed boundary method for interfacial flows with insoluble surfactant. Journal of Computational Physics, 2008, 227(15): 7279-7293.
		
		
		\bibitem{XuJCP2006}
		Xu J J, Li Z, Lowengrub J, et al. A level-set method for interfacial flows with surfactant. Journal of Computational Physics, 2006, 212(2): 590-616.
		
		
		\bibitem{XuJCP2018}
		Xu J J, Shi W, Lai M C. A level-set method for two-phase flows with soluble surfactant. Journal of Computational Physics, 2018, 353: 336-355.
		
		
		\bibitem{JamesJCP2004}
		James A J, Lowengrub J. A surfactant-conserving volume-of-fluid method for interfacial flows with insoluble surfactant. Journal of Computational Physics, 2004, 201(2): 685-722.
		
	
		\bibitem{XueJCP2026}
		Xue Z H, Magnaudet J, Zhang J. A sharp and conservative method for modeling interfacial flows with insoluble surfactants in the framework of a geometric volume-of-fluid approach. Journal of Computational Physics, 2026, 559, 114936.
		
		
		\bibitem{TeigenCMS2009}
		Teigen K E, Li X, Lowengrub J, et al. A diffuse-interface approach for modeling transport, diffusion and adsorption/desorption of material quantities on a deformable interface. Communications in Mathematical Sciences, 2009, 7(4): 1009–1037.
		
		\bibitem{TeigenJCP2011}
		Teigen K E, Song P, Lowengrub J, et al. A diffuse-interface method for two-phase flows with soluble surfactants. Journal of Computational Physics, 2011, 230(2): 375-393.
		
		
		\bibitem{LiuJFM2018}
		Liu H, Ba Y, Wu L, et al. A hybrid lattice Boltzmann and finite difference method for droplet dynamics with insoluble surfactants. Journal of Fluid Mechanics, 2018, 837: 381-412.
		
		
		\bibitem{HuAML2021}
		Hu Y. A diffuse interface–lattice Boltzmann model for surfactant transport on an interface. Applied Mathematics Letters, 2021, 111: 106614.
		

		\bibitem{YamashitaJCP2024}
		Yamashita S, Matsushita S, Suekane T. Conservative transport model for surfactant on the interface based on the phase-field method. Journal of Computational Physics, 2024, 516: 113292.
		
		
		\bibitem{JainJCP2024}
		Jain S S. A model for transport of interface-confined scalars and insoluble surfactants in two-phase flows. Journal of Computational Physics, 2024, 515: 113277.
		

		
		\bibitem{ZhanPRE2025}
		Zhan C, Liang H, Chai Z, et al. Phase-field-based lattice Boltzmann method for the transport of insoluble surfactant in two-phase flows. Physical Review E, 2025, 112(3): 035303.
		
		
		\bibitem{HaoJCP2025}
		Hao H, Li X, Xu L, et al. Profile-preserving phase-field model for surfactant transport and adsorption-desorption in two-phase flow systems. Journal of Computational Physics, 2026, 546: 114510.
		
		
		\bibitem{HaoJCP2025a}
		Hao H, Li X, Liu T, et al. Enhanced profile-preserving phase-field model of two-phase flow with surfactant interfacial transport and Marangoni effects. Journal of Computational Physics, 2025, 536: 114058.
		
		
		\bibitem{VanRA2006}
		Van der Sman R G M, Van der Graaf S. Diffuse interface model of surfactant adsorption onto flat and droplet interfaces. Rheologica Acta, 2006, 46(1): 3-11.
		
		\bibitem{LiuJCP2010}
		H. Liu, Y. Zhang, Phase-field modeling droplet dynamics with soluble surfactants, Journal of Computational Physics 229 (2010) 9166-9187.
		
		
		\bibitem{YunAMC2014}
		Yun A, Li Y, Kim J. A new phase-field model for a water–oil-surfactant system. Applied Mathematics and Computation, 2014, 229: 422-432.
	
	
		\bibitem{ZongPOF2020}
		Zong Y, Zhang C, Liang H, et al. Modeling surfactant-laden droplet dynamics by lattice Boltzmann method. Physics of Fluids, 2020, 32(12).
		
	
		\bibitem{YangCMAME2021}
		Yang X. A novel fully-decoupled, second-order and energy stable numerical scheme of the conserved Allen–Cahn type flow-coupled binary surfactant model. Computer Methods in Applied Mechanics and Engineering, 2021, 373: 113502.


		\bibitem{LiangJCP2026}
		Liang H, Tong S, Wang L, et al. A thermodynamically consistent and conservative diffuse-interface model for surfactant-laden two-phase dynamics. Journal of Computational Physics, 2026, 548: 114584.		
		
		
		\bibitem{YangJCP2022}
		Yang J, Tan Z, Kim J. Linear and fully decoupled scheme for a hydrodynamics coupled phase-field surfactant system based on a multiple auxiliary variables approach. Journal of Computational Physics, 2022, 452: 110909.		
		
		\bibitem{CahnJCP1977}
		Cahn J W. Critical point wetting. The Journal of Chemical Physics, 1977, 66(8): 3667-3672.
		
		\bibitem{YueJCP2010}
		Yue P, Zhou C, Feng J J. Sharp-interface limit of the Cahn–Hilliard model for moving contact lines. Journal of Fluid Mechanics, 2010, 645: 279-294.
		
		
		
		\bibitem{ZhuJFM2019}
		Zhu G, Kou J, Yao B, et al. Thermodynamically consistent modelling of two-phase flows with moving contact line and soluble surfactants. Journal of Fluid Mechanics, 2019, 879: 327-359.
		
		
		\bibitem{ZhuJCP2020}
		Zhu G, Kou J, Yao J, et al. A phase-field moving contact line model with soluble surfactants. Journal of Computational Physics, 2020, 405: 109170.
		
		
		\bibitem{WangJCP2022}
		Wang C, Guo Y, Zhang Z. Unconditionally energy stable and bound-preserving schemes for phase-field surfactant model with moving contact lines. Journal of Scientific Computing, 2022, 92(1): 20.
		
		
		\bibitem{WangJCP2024}
		Wang C, Lai M C, Zhang Z. An improved phase-field algorithm for simulating the impact of a drop on a substrate in the presence of surfactants. Journal of Computational Physics, 2024, 499: 112722.
		
		
		\bibitem{WuAML2024}
		Wu Y, Tan Z. An immersed boundary-phase field fluid-surfactant model with moving contact lines on curved substrates. Applied Mathematics Letters, 2024, 153: 109072.
		
		
		\bibitem{ZhangICHMT2025}
		Zhang S, Wang X, Zhu Y, et al. A second-order phase field-lattice Boltzmann model for two-phase flows with moving contact line and soluble surfactants. International Communications in Heat and Mass Transfer, 2025, 165: 109104.
		
		
		\bibitem{YaoIJMF2026}
		Yao L, Wang P, Zhang L, et al. A novel phase-field lattice Boltzmann method moving contact line model with soluble surfactants. International Journal of Multiphase Flow, 2026, 197: 105608.
		
		
		\bibitem{WuCMAME2026}
		Wu W, Zhang Z, Wei C. Primal-dual splitting methods for phase-field surfactant model with moving contact lines. Computer Methods in Applied Mechanics and Engineering, 2026, 450: 118670.
		
		
		
		
		\bibitem{StoebeLangmuir1996}
		Stoebe T, Lin Z, Hill R M, et al. Surfactant-enhanced spreading. Langmuir, 1996, 12(2): 337-344.
		
		
		\bibitem{KarapetsasJFM2011}
		Karapetsas G, Craster R V, Matar O K. On surfactant-enhanced spreading and superspreading of liquid drops on solid surfaces. Journal of Fluid Mechanics, 2011, 670: 5-37.
		
		
		\bibitem{BeraLangmuir2016}
		Bera B, Duits M H G, Cohen Stuart M A, et al. Surfactant induced autophobing. Soft Matter, 2016, 12(20): 4562-4571.
		
		
		\bibitem{TadmorLangmuir2019}
		Tadmor R, Baksi A, Gulec S, et al. Drops that change their mind: spontaneous reversal from spreading to retraction. Langmuir, 2019, 35(48): 15734-15738.
		
		
	
		\bibitem{KannanJFM2026}
		Kannan P K, Iqbal K T, Díaz D, et al. Solid adsorption: the missing mechanism for surfactant contact lines–a phase-field approach. Journal of Fluid Mechanics, 2026, 1038: A49.
		
		
		
		
		\bibitem{CahnJCP1958}
		Cahn J W, Hilliard J E. Free energy of a nonuniform system. I. Interfacial free energy. The Journal of Chemical Physics, 1958, 28(2): 258-267.
		
		
		\bibitem{AllenAM1976}
		Allen S M, Cahn J W. Mechanisms of phase transformations within the miscibility gap of Fe-rich Fe-Al alloys. Acta Metallurgica, 1976, 24(5): 425-437.
		
		
		\bibitem{ChaiIJHMT2018}
		Chai Z, Sun D, Wang H, et al. A comparative study of local and nonlocal Allen-Cahn equations with mass conservation. International Journal of Heat and Mass Transfer, 2018, 122: 631-642.
		

		\bibitem{ChiuJCP2011}
		Chiu P H, Lin Y T. A conservative phase field method for solving incompressible two-phase flows. Journal of Computational Physics, 2011, 230(1): 185-204.


		\bibitem{RubinsteinIMAJAM1992}
		Rubinstein J, Sternberg P. Nonlocal reaction—diffusion equations and nucleation. IMA Journal of Applied Mathematics, 1992, 48(3): 249-264.
		
		
		\bibitem{BrasselMMAS2011}
		M. Brassel, E. Bretin, A modified phase field approximation for mean curvature flow with conservation of the volume, Mathematical Methods in the Applied Sciences 34 (2011) 1157-1180.
		
		
		\bibitem{KimIJES2014}
		Kim J, Lee S, Choi Y. A conservative Allen–Cahn equation with a space–time dependent Lagrange multiplier. International Journal of Engineering Science, 2014, 84: 11-17.
		
		

		
		
		
		
		
		
		
		
		\bibitem{KwakAML2022}
		S. Kwak, J. Yang, J. Kim, A conservative Allen-Cahn equation with a curvature-dependent Lagrange multiplier, Applied Mathematics Letters 126 (2022) 107838.
		
		\bibitem{ChoiEABE2023}
		Y. Choi, J. Kim, Maximum principle preserving and unconditionally stable scheme for a conservative Allen-Cahn equation, Engineering Analysis with Boundary Elements 150 (2023) 111-119.
		
		
		\bibitem{YangCNSNS2024}
		J. Yang, J. Kim, Unconditionally maximum principle-preserving linear method for a mass-conserved Allen-Cahn model with local Lagrange multiplier, Communications in Nonlinear Science and Numerical Simulation 139 (2024) 108327.
		
		
		\bibitem{HwangEABE2026}
		Hwang Y, Ma J, Nam Y, et al. Review of the conservative Allen–Cahn equations. Engineering Analysis with Boundary Elements, 2026, 187: 106713.
		
		
		
		\bibitem{ZhouPRSA2025}
		Zhou Z, Jiang W, Qian T, et al. A new phase-field model for anisotropic surface diffusion: anisotropic Cahn–Hilliard equation with improved conservation. Proceedings of the Royal Society A, 2025, 481(2326): 20250649.
		
		\bibitem{MusilArxiv2026}
		Musil J. Third-Order Geometric-Volume Conservation in Cahn--Hilliard Models. arXiv preprint arXiv:2602.01497, 2026.
		
		
		
		\bibitem{ChangCSA1995}
		Chang C H, Franses E I. Adsorption dynamics of surfactants at the air/water interface: a critical review of mathematical models, data, and mechanisms. Colloids and Surfaces A: Physicochemical and Engineering Aspects, 1995, 100: 1-45.
		
		\bibitem{JuJML2017}
		Ju H, Jiang Y, Geng T, et al. Equilibrium and dynamic surface tension of quaternary ammonium salts with different hydrocarbon chain length of counterions. Journal of Molecular Liquids, 2017, 225: 606-612.
		
		\bibitem{SoligoJCP2019}
		Soligo G, Roccon A, Soldati A. Coalescence of surfactant-laden drops by phase field method. Journal of Computational Physics, 2019, 376: 1292-1311.
		
		\bibitem{EngblomCICP2013}
		Engblom S, Do-Quang M, Amberg G, et al. On diffuse interface modeling and simulation of surfactants in two-phase fluid flow. Communications in Computational Physics, 2013, 14(4): 879-915.
		
		
		\bibitem{LiuMMS2022}
		Liu X, Chai Z, Zhan C, et al. A diffuse-domain phase-field lattice Boltzmann method for two-phase flows in complex geometries. Multiscale Modeling \& Simulation, 2022, 20(4): 1411-1436.
		
		
		
		
		
		
		
		
		\bibitem{ChaiPRE2020}
		Chai Z, Shi B. Multiple-relaxation-time lattice Boltzmann method for the Navier-Stokes and nonlinear convection-diffusion equations: Modeling, analysis, and elements. Physical Review E, 2020, 102(2): 023306.
		
		
		\bibitem{JohnsonCS1986}
		Johnson B A, Kreuter J, Zografi G. Effects of surfactants and polymers on advancing and receding contact angles. Colloids and Surfaces, 1986, 17(4): 325-342.
		
		
		\bibitem{GokhaleLangmuir2005}
		Gokhale S J, Plawsky J L, Wayner P C. Spreading, evaporation, and contact line dynamics of surfactant-laden microdrops. Langmuir, 2005, 21(18): 8188-8197.
		
		
		
		
	\end{thebibliography}
\end{document}